\documentclass[english]{article}
\usepackage[latin9]{inputenc}
\usepackage{float}
\usepackage{url}
\usepackage{amsmath}
\usepackage{amssymb}
\usepackage[authoryear]{natbib}

\makeatletter

\providecommand{\tabularnewline}{\\}
\floatstyle{ruled}
\usepackage{graphicx}
\usepackage{bm}
\usepackage{pgf}
\usepackage{url}
\usepackage{float}

\usepackage{amsfonts}% Include AMS fonts
\usepackage{amsthm}% For theorem-like environments

\bibpunct{(}{)}{;}{a}{}{,}

\usepackage{babel}

\usepackage{chngcntr}
\counterwithout{figure}{section}

\makeatother

\usepackage{babel}

\allowdisplaybreaks

\makeatletter
\def\newrmtheorem#1{\@ifnextchar[{\@rmothm{#1}}{\@rmnthm{#1}}}

\def\@rmnthm#1#2{%
\@ifnextchar[{\@rmxnthm{#1}{#2}}{\@rmynthm{#1}{#2}}}

\def\@rmxnthm#1#2[#3]{\expandafter\@ifdefinable\csname #1\endcsname
{\@definecounter{#1}\@addtoreset{#1}{#3}%
\expandafter\xdef\csname the#1\endcsname{\expandafter\noexpand
  \csname the#3\endcsname \@rmthmcountersep \@rmthmcounter{#1}}%
\global\@namedef{#1}{\@rmthm{#1}{#2}}\global\@namedef{end#1}{\@endrmtheorem}}}

\def\@rmynthm#1#2{\expandafter\@ifdefinable\csname #1\endcsname
{\@definecounter{#1}%
\expandafter\xdef\csname the#1\endcsname{\@rmthmcounter{#1}}%
\global\@namedef{#1}{\@rmthm{#1}{#2}}\global\@namedef{end#1}{\@endrmtheorem}}}

\def\@rmothm#1[#2]#3{\expandafter\@ifdefinable\csname #1\endcsname
  {\global\@namedef{the#1}{\@nameuse{the#2}}%
\global\@namedef{#1}{\@rmthm{#2}{#3}}%
\global\@namedef{end#1}{\@endrmtheorem}}}

\def\@rmthm#1#2{\refstepcounter
    {#1}\@ifnextchar[{\@rmythm{#1}{#2}}{\@rmxthm{#1}{#2}}}

\def\@rmxthm#1#2{\@beginrmtheorem{#2}{\csname the#1\endcsname}\ignorespaces}
\def\@rmythm#1#2[#3]{\@opargbeginrmtheorem{#2}{\csname
       the#1\endcsname}{#3}\ignorespaces}

\def\@rmthmcounter#1{\noexpand\arabic{#1}}
\def\@rmthmcountersep{}
\def\@beginrmtheorem#1#2{\rm \trivlist
      \item[\hskip \labelsep{\bf #1\ #2\thmrmcounterend}]}
\def\@opargbeginrmtheorem#1#2#3{\rm \trivlist
      \item[\hskip \labelsep{\bf #1\ #2\ (#3)\thmrmcounterend}]}
\def\@endrmtheorem{\endtrivlist}
\def\thmrmcounterend{\hskip 0em\relax}
\def\newrmwntheorem#1#2{\expandafter\@ifdefinable\csname #1\endcsname%
\global\@namedef{#1}{\@rmwnthm{#1}{#2}}%
\global\@namedef{end#1}{\@endrmwntheorem}}

\def\newsltheorem#1{\@ifnextchar[{\@slothm{#1}}{\@slnthm{#1}}}

\def\@slnthm#1#2{%
\@ifnextchar[{\@slxnthm{#1}{#2}}{\@slynthm{#1}{#2}}}

\def\@slxnthm#1#2[#3]{\expandafter\@ifdefinable\csname #1\endcsname
{\@definecounter{#1}\@addtoreset{#1}{#3}%
\expandafter\xdef\csname the#1\endcsname{\expandafter\noexpand
  \csname the#3\endcsname \@slthmcountersep \@slthmcounter{#1}}%
\global\@namedef{#1}{\@slthm{#1}{#2}}\global\@namedef{end#1}{\@endsltheorem}}}

\def\@slynthm#1#2{\expandafter\@ifdefinable\csname #1\endcsname
{\@definecounter{#1}%
\expandafter\xdef\csname the#1\endcsname{\@slthmcounter{#1}}%
\global\@namedef{#1}{\@slthm{#1}{#2}}\global\@namedef{end#1}{\@endsltheorem}}}

\def\@slothm#1[#2]#3{\expandafter\@ifdefinable\csname #1\endcsname
  {\global\@namedef{the#1}{\@nameuse{the#2}}%
\global\@namedef{#1}{\@slthm{#2}{#3}}%
\global\@namedef{end#1}{\@endsltheorem}}}

\def\@slthm#1#2{\refstepcounter
    {#1}\@ifnextchar[{\@slythm{#1}{#2}}{\@slxthm{#1}{#2}}}

\def\@slxthm#1#2{\@beginsltheorem{#2}{\csname the#1\endcsname}\ignorespaces}
\def\@slythm#1#2[#3]{\@opargbeginsltheorem{#2}{\csname
       the#1\endcsname}{#3}\ignorespaces}

\def\@slthmcounter#1{.\noexpand\arabic{#1}}
\def\@slthmcountersep{}
\def\@beginsltheorem#1#2{\sl \trivlist
      \item[\hskip \labelsep{\bf #1\ #2\thmslcounterend}]}
\def\@opargbeginsltheorem#1#2#3{\sl \trivlist
      \item[\hskip \labelsep{\bf #1\ #2\ (#3)\thmslcounterend}]}
\def\@endsltheorem{\endtrivlist}
\def\thmslcounterend{\hskip 0em\relax}
\def\newslwntheorem#1#2{\expandafter\@ifdefinable\csname #1\endcsname%
\global\@namedef{#1}{\@slwnthm{#1}{#2}}%
\global\@namedef{end#1}{\@endslwntheorem}}

\makeatother

\newsltheorem{theorem}{Theorem}[section]
\newsltheorem{proposition}[theorem]{Proposition}
\newsltheorem{lemma}[theorem]{Lemma}
\newsltheorem{corollary}[theorem]{Corollary}

\newrmtheorem{definition}[theorem]{Definition}
\newrmtheorem{example}[theorem]{Example}
\newrmtheorem{remark}[theorem]{Remark}
\newrmtheorem{problem}[theorem]{Problem}
\newrmtheorem{algorithm}[theorem]{Algorithm}
\newrmtheorem{cond}[theorem]{Condition}
\newrmtheorem{convention}[theorem]{Convention}
\newrmtheorem{assumption}[theorem]{Assumption}

\newcommand{\dd}{\mathrm{d}}

\DeclareMathOperator\erfc{erfc}

\begin{document}

\title{Accounting Noise and the Pricing of CoCos}%

\author{Mike Derksen\thanks{Korteweg-de Vries Institute for Mathematics,
University of Amsterdam, {\tt m.j.m.derksen@uva.nl}} \and
        Peter Spreij\thanks{Korteweg-de Vries Institute for Mathematics,
University of Amsterdam, {\tt spreij@uva.nl}, 
and
Institute for Mathematics, Astrophysics and Particle Physics,
Radboud University
Nij\-me\-gen} 
\and
Sweder van Wijnbergen\thanks{University of Amsterdam, Tinbergen Institute, CEPR,
{\tt s.j.g.vanwijnbergen@uva.nl}}
}

\maketitle

\begin{abstract}\noindent
Contingent Convertible bonds (CoCos) are debt instruments that convert
into equity or are written down in times of distress. Existing pricing
models assume conversion triggers based on market prices and on the
assumption that markets can always observe all relevant firm information.
But all Cocos issued so far have triggers based on accounting ratios
and/or regulatory intervention. We incorporate that markets receive
information through noisy accounting reports issued at discrete time instants,
which allows us to distinguish between market and accounting values,
and between automatic triggers and regulator-mandated conversions.
Our second contribution is to incorporate that coupon payments are
contingent too: their payment is conditional on the Maximum Distributable
Amount not being exceeded. We examine the impact of CoCo design parameters,
asset volatility and accounting noise on the price of a CoCo; and
investigate the interaction between CoCo design features, the capital
structure of the issuing bank and their implications for risk taking
and investment incentives. Finally, we use our model to explain the
crash in CoCo prices after Deutsche Bank's profit warning in February
2016.
\end{abstract}

\section{Introduction}

\setcounter{equation}{0}

Contingent Capital instruments or Contingent Convertible bonds (CoCos)
are debt instruments designed to convert into equity or to be written
down in times of distress. They differ from regular convertibles in
that conversion is not an option to be exercized by the holder; conversion
is either automatically triggered in response to a particular accounting
ratio falling below a specified threshold or at the discretion of
the regulator when a so called Point of Non-Viability (PONV) has been
reached. They were originally proposed by \citet{Flannery} as an
alternative to raising equity in times of distress. Their use has
exploded since the Great Financial Crisis eroded the capital
base of banks across the world and regulators responded by actually
raising capital requirements so as to increase the Loss Absorption
Capacity of the banking system. In this paper we develop a valuation
model that not only takes into account their particular contingent
properties but also explicitly incorporates the fact that markets
get only imperfect information about the underlying firm dynamics
through noisy accounting reports issued at regular discretely spaced
time points.

An academic literature has quickly emerged since Flannery's original
advocacy of CoCos~\citep{Flannery}, and unanimously recommends basing
trigger ratios on market values. In line with that view, the pricing
models proposed so far assume conversion triggers based on market
prices. Moreover, the literature is without exception built on the
assumption that markets can always observe all relevant firm information,
so the existing literature in fact assumes there is no difference
between accounting and market values, which makes the choice between
them obviously irrelevant. Whatever the merit of that view (the preference
for market based triggers), fact is that basing conversion on market
prices actually makes a CoCo ineligible for being counted against
capital requirements under European law (cf.~\citet[art.~54]{CRRart54}, henceforth referred to as CRR) under
the framework implementing Basel III capital requirements in the European
Union. Accordingly, all CoCos issued so far have without exception
triggers for conversion based on accounting ratio's falling below
a particular ratio (the trigger ratio).

In this paper we attempt to bridge this discrepancy between the academic
valuation literature on the one hand and the actual capital market
practice on the other by valuing CoCos under the assumption that the
only information available is noisy accounting information, our first
contribution, which in addition is only received at pre-specified
discrete moments in time. The underlying processes are continuous,
but markets only receive noisy information on those underlying fundamental
processes at discrete time instants: accounting reports are only released
at discretely spaced points in time, the accounting dates (for example
at the end of each quarter). In this way it is possible to distinguish
between market and accounting values. This informational structure
gives rise to potential discrepancies between true underlying values,
accounting ratios and market prices.

We make a second contribution towards a better pricing model for CoCos.
The asset pricing literature has so far concentrated on the conversion
contingency. But the possibility of a write down or conversion into
equity is not the only option-like characteristic embedded in CoCo
designs. The coupon payments are contingent too, they can only be
paid out if that payment does not exceed the so called Maximum Distributable
Amount, a trigger that binds much earlier than the conversion trigger
(see \citep{Kiewiet_ea_2017} for extensive details). The relevance
of this contingency became very clear in the beginning of 2016, when
a profit warning of Deutsche Bank ahead of their first quarter accounting
report set off an accross the board crash in CoCo prices (cf.~again
\citet{Kiewiet_ea_2017}).

The model developed in this paper thus contributes in two different
ways to the existing literature; it distinguishes between market and
book values of assets in the valuation of CoCos and it allows the
coupons of CoCos to be already cancelled at a moment before the conversion
date. The model is based on the approach used by \citet{Duffie2001},
in which debt is valued under the assumption that the only information
available is noisy accounting information which is received at selected
moments in time. This setting is particularly relevant for the pricing
of CoCos since, as pointed out above, the non-discretionary conversion
triggers are always based on imperfect accounting ratios observed
at discrete moments in time, rather than on continuously observable
market prices.

We first set up a comprehensive description of the structural credit
risk model proposed by Duffie and Lando, including the derivation of all the relevant formulas
and proofs. We then go beyond the paper by Duffie and Lando by using their framework to provide explicit formulas and algorithms for
the pricing of CoCos. The setting is applied to the valuation of different
kinds of CoCo bonds, namely CoCos with a (partial) principal write
down and CoCos with a conversion into shares. Also a distinction is
made between CoCos with a discretionary regulatory trigger, for which
conversion could happen at any moment in time, and CoCos that can
only be triggered at one of the accounting dates. The model does not
lead to closed form solutions, but the expressions for CoCo prices
involve integrals that are computed using MCMC-methods.

We first use the model developed in this paper to examine the impact
of several CoCo design parameters on the price of a CoCo; we then
investigate the interaction between CoCo design features, the capital
structure of the issuing bank and their implications for risk taking
and investment incentives. Finally, the model is used to explain the
big downward price jump that CoCos of Deutsche Bank suffered at the
beginning of 2016 after the release of a profit warning. In this particular
case the added value of the proposed model becomes clear as it
allows for the announcement of a bad accounting report and explicitly
allows for the early cancellation of coupon payments (before conversion)
when the payment would exceed the so called Maximum Distributable
Amount (MDA) trigger, cf.~\cite{GoldmanSachs2016}. Market sources indeed indicated
at the time that the sudden price drop was out of fear for the MDA
trigger more than for setting off the conversion trigger, as the conversion
trigger was still far out of reach.

The remainder of this paper is as follows. Section~\ref{sec:cocos}
describes CoCos in detail, their design features and their regulatory
treatment. Section~\ref{sec:literature} surveys the existing asset
price literature on CoCos. Section~\ref{sec:model} sets up the pricing
model, making a distinction between market values and accounting values
and incorporating the possibility of early cancelling of coupons triggered
by the MDA regulations referred to earlier, Section~\ref{Section_Simulations} outlines the MCMC algorithms
used for evaluating the integrals involved in the final pricing expressions.
We outline the derivation of the key pricing formulas, with full details
in the Appendix.
Section~\ref{sec:application} uses the model to analyse the sensitivity
of CoCo valuation to various design features, changes in the firm's
capital structure and various external shocks. We also analyse the
events following Deutsche Bank's profit warning of late February 2016,
and show that our pricing model does quite well in explaining the
observed CoCo price response. Section~\ref{sec:conclusions} summarizes
and concludes. Proofs of the technical results are collected in the
Appendix.

\section{Contingent convertible bonds}
\label{sec:cocos} \setcounter{equation}{0}

A Contingent Convertible bond is a bond which converts into equity
or is (partially) written down at the conversion date. This means
that the design of a CoCo contract is specified by two main characteristics: 
\begin{itemize}
\item The trigger event: when does conversion happen? 
\item The conversion mechanism: what does happen at conversion? 
\end{itemize}

\subsection{The trigger event}

The trigger event specifies at which moment the conversion takes place.
We can distinguish three types of trigger events; an accounting trigger,
a market trigger and a regulatory trigger. In case of an accounting
trigger, the conversion is triggered by an accounting ratio, e.g.
the Common Equity Tier 1 Ratio (defined as the fraction of common
equity over (risk weighted) assets) falling below a certain barrier.
This type of trigger is typical in practice, although it is widely
criticized in the academic world. For example, in \citet{Flannery}
it is argued that a book value will only be triggered long after the
damage has already occurred, because book values are not up-to-date
at any moment. Therefore, the academic literature widely supports
the use of market price based triggers. In the case of a market trigger,
the conversion happens if a market value, e.g.\ the share price of
the issuing bank, falls below a certain threshold. A market price
is thought to better reflect the current situation of the issuing
bank, because a market price is a forward looking parameter; it reflects
the market's opinion on the future of the bank. See \citet{Haldane2011}
and also \citet{Pennacchi2015} for a very articulate defense of this
point of view, to which we return in the next section. Against this
point of view, in \citet{Sundaresan2014} and \citet{Glasserman2012b}
it is argued that a market trigger could lead to a multiple equilibria
problem for the pricing of a CoCo if the terms of conversion are beneficial
to CoCo holders. In this case, a market trigger could also encourage
CoCo holders to short-sell shares of the issuing bank, to profit from
a conversion, which could subsequently lead to a ``death spiral\char`\"{}.
These warnings may well explain why market based triggers are actually
outlawed in the European Union, cf.~CRR. Whatever the EU's
reason for this outlawing of market based triggers, as a consequence
no CoCos with market price based triggers have been issued so far.
A third type of trigger is the regulatory trigger, which allows the
regulator to call for a conversion. All CoCos issued so far have a
trigger mechanism which is a combination of an accounting trigger
and a regulatory trigger, since that is required for the CoCo to count
as regulatory capital in the European Union. The regulatory trigger
has not been discussed in the asset pricing literature yet (but see
\citet{Chan2018} for a corporate finance perspective on CoCo triggers
focusing on risk taking and regulatory forbearance, i.e.\ regulatory
behavior).

\subsection{The Conversion mechanism}

The conversion mechanism specifies what happens at the moment of conversion.
There are two possibilities: a (partial) principal write-down or a
conversion into shares. In case of a (partial) principal write-down
(PWD) mechanism, the principal of the CoCo bond is (partially) written
down at the moment of conversion, to strengthen the capital position
of the issuing bank. In case of a conversion into shares, the principal
of the CoCo bond is converted into a number of shares. Of course,
it needs to be specified how many shares a CoCo holder receives at
conversion. This conversion rule can be designed in two different
ways. One possibility is that the CoCo holder receives a fixed number
of shares $\Delta$ for every monetary unit of principal. This corresponds
to a pre-specified share price 
$1/\Delta$. Another option is a variable number of shares related
to the market price prevailing at the moment of time conversion takes
place. In this case the CoCo holder would ``buy\char`\"{} a number
of shares against a market price based conversion price. Some authors
have warned for the possibility of a ``death spiral\char`\"{}
when CoCo conversion is advantageous to CoCo holders and thus leads
to incentives to short sell the stock (cf.~\citet{Sundaresan2014}).
This possibly leads to an infinitely large dilution of the existing
shareholders. A way to avoid this would be to place a floor under
the conversion price, again a requirement for the CoCo to count as
capital under European law.

\subsection{The Maximum Distributable Amount (MDA) trigger}

As Contingent Convertible bonds qualify as a form of capital in the
Basel III regulations, they are also affected by the concept of the
Maximum Distributable Amount (MDA), which requires regulators to block
earnings distributions when the bank's capital becomes too low. An
example of such earnings distributions are dividends, but also CoCo
coupon payments if the CoCos qualify as AT1 capital. This means that
when the bank's capital falls below some threshold, always (much)
higher than the CoCo's conversion trigger, the payment of coupons
is stopped until the bank's capital is again above the MDA trigger.
See \citet{Kiewiet_ea_2017} for a detailed discussion of the MDA
trigger for coupons. This trigger has not been considered before in
the asset pricing literature, but will be introduced explicitly in
this paper.

\section{Related Literature}

\label{sec:literature} \setcounter{equation}{0}

The existing asset pricing literature on CoCos can be grouped in three
categories (cf.~\citet{Wilkens2014} for an early assessment following
the same classification): structural models, equity derivative models
and credit risk or reduced form models. In a structural model one
starts by describing the value of the assets of a firm by a stochastic
process. Then the liabilities are introduced and equity is the difference
between the assets and those liabilities. Conversion of CoCos occurs
when the market value of the firm's assets or the firm's capital ratio
falls below a predetermined value (the conversion trigger). In most
papers, liquidation of the firm is also incorporated in the model
by assuming that the equity holders liquidate the firm when the value
of assets falls below some optimal threshold, chosen by the shareholders
to maximize equity value. Furthermore, it is assumed that default
cannot occur before conversion. An early example of such a structural
model is \citet{Albul2012}; there the firm's value $A_{t}$ is described
by a geometric Brownian motion (GBM) process under the Risk Neutral
measure given by 
\begin{eqnarray*}
\dd A_{t} & = & \mu A_{t}\,\dd t+\sigma A_{t}\,\dd W_{t},
\end{eqnarray*}
with $\mu$ and $\sigma$ constants and $W$ a standard Brownian
motion. The risk-free rate is assumed constant. The firm also issues
two types of debt: a straight bond and a CoCo, both with perpetual
maturities and both paying coupons at a constant rate. The CoCo converts
into equity the first time the asset value falls below some threshold
$\alpha_{c}$, so the conversion time is $\tau(\alpha_{c})=\inf\{t\geq0:A_{t}\leq\alpha_{c}\}$.
In their set up, the CoCo converts into equity valued at market prices
at a specified conversion ratio $\lambda$, where $\lambda=1$ means
that the CoCo holder receives equity with a market value equal to
the face value of the CoCo at issue. Like in all other papers surveyed
here, the value of the various claims (including the CoCos) is given
by the risk-neutral expectation of the discounted future cashflows
regarding the claim. The simplicity of the model allows for closed
form expressions. In~\citet{Pennacchi2011} a similar model is introduced,
but also proportional jump processes are added to the firm's dynamics by
adding a compound Poisson process; asset values are governed by the
following stochastic process (also under the risk neutral measure)
\begin{eqnarray*}
\dd A_{t} & = & (r-\lambda_{t}k_{t})A_{t}\,\dd t+\sigma A_{t}\,\dd W_{t}+(Y_{q_{t}}-1)\,\dd q_{t}.
\end{eqnarray*}
Here $\lambda_{t}$ is the risk neutral jump intensity of the Poisson process $q_t$, $k_{t}=E_{Q}(Y-1)$
is the expected proportional jump under the risk neutral measure in
case a Poisson jump occurs.\footnote{cf.~\citet[Chapter~11]{Shreve2004} for a discussion of the various types
of Poisson processes.} The model does not yield closed form solutions so Monte Carlo simulation
is used to sketch the solution structure. In~\citet{Chen2013} an
equally involved model is proposed in which the asset value process
also involves a GBM process with Poisson jumps added in, with a distinction
between market wide and firm specific jumps. They are mainly interested in downside shocks and, for tractability, they
assume that minus the log
of the jump sizes have exponential distributions, which allows the
authors to derive closed form solutions. Conversion of CoCos into
equity is triggered the first time the value of assets falls below
some specified threshold. In contrast to the variable conversion share
price featured in \citet{Albul2012} and \citet{Pennacchi2011}, the
CoCo holders receive a fixed number of shares for every dollar of
principal when the CoCo converts, which is the way all Cocos with a conversion into shares are set up in practice, cf.\ \citet{Avdiev2017}. An interesting innovation
is their introduction of finite maturity debt and the associated potential
debt roll over problems. This feature has a significant effect on
risk taking behavior before conversion: even when the share conversion
takes place at a rate favorable to the old shareholders, conversion
leads to higher roll over costs of short term debt, which mitigates
risk taking incentives ex ante. The model in~\citet{Pennacchi2015}
reverts to a straight GBM process driving asset values, and the focus
is on the uniqueness and in fact existence of a price equilibrium
when conversion involves a wealth transfer favoring either the CoCo
holder (dilutive CoCos) or the old equity holder (non-dilutive CoCos).
Academics widely favor conversion triggers based on market prices
and dilutive conversion ratios, but in~\citet{Sundaresan2014} it
has been argued that stipulating triggers based on market prices leads
to multiple price equilibria in the case of dilutive CoCos and in
fact non-existence in the case of non-dilutive CoCos (i.e.\ conversion
at terms favoring the old shareholders). Both in \citet{Glasserman2012b}
and \citet{Pennacchi2015} it is shown that price equilibria will in
fact be unique in the case of dilutive CoCos. In~\citet{Pennacchi2015}
it is furthermore shown that for perpetual CoCos (which is the structure
most seen in practice) non-existence only occurs for implausible parameter
values even when they are non-dilutive.

All these models have in common that a conversion trigger based on
market values is used, as is widely recommended in the academic literature
(cf.~in particular \citet{Haldane2011} and \citet{Pennacchi2015}
for an extensive discussion of why market prices should be used for
conversion triggers). The problem with basing one's analysis on that
view, whatever its merit, is that there is literally no single CoCo
ever issued, at least within the European Union, that follows such
a trigger definition. Without exception, within the European Union,
where the bulk of all CoCos have been issued, trigger ratios are based
on accounting values. In fact in the EU, market based triggers are
illegal under European law, cf.~CRR, or at least cannot
be counted as capital\footnote{European Law explicitly states that in order to qualify as an Additional
Tier 1 instrument for capital purposes, a CoCo instrument should have
a mechanical book-value based trigger which needs to have been mentioned
explicitly in the prospectus.}.

Moreover, none of the models discussed so far actually distinguishes
between market and accounting based valuation. The single exception
in the literature is \citet{Glasserman2012a}, where it is assumed
that markets and accountants agree on whether a firm is solvent (if
the value of the assets exceeds the value of debt based on market
prices, accounting values are assumed to do likewise). But the ratio
between the market value of equity and the accounting value of equity,
roughly similar to the market-to-book (M-to-B) ratio, follows a GBM
process. This approach gives two additional parameters to be used
in calibration: the volatility of the M-to-B value process and its
correlation to the market process. This is an imaginative attempt
to endogenize the M-to-B value process, but there are problems with
this first attempt at endogenizing the difference between market values
and accounting ratios. First of all, \citet{Haldane2011} doubted
casts on their key assumption, that market and accounting valuations
always agree on whether the firm is solvent. A less fundamental but
practically speaking equally serious problem is that the approach
in \citet{Glasserman2012a} assumes that all processes can be observed
continuously. In practice however regulatory capital ratios are only
calculated on a quarterly basis.

In this paper we address both shortcomings. We assume that firm values
are driven by a GBM, without jumps; we omit jumps for reasons explained
in Section~\ref{sec:model}. We do not introduce a separate independent accounting process like~\citet{Brigo2013}.
Instead we follow the approach taken in \citet{Duffie2001} who assume
an underlying GBM process for the dynamic evolvement of the firm's
asset valuation, but stipulate that that process is not directly observable.
Instead, noisy information (``accounting report'') is brought out
at discrete time instants (``quarters''). It is reasonable that
noise in accounting reports has some persistence, so we assume that the
noise term in the accounting report is serially correlated. See Section~\ref{sec:model}
again for a detailed description.

In addition to these structural models, the literature has seen two
other approaches, the credit derivative approach and an equity derivative
approach. A credit derivative approach is a reduced form approach
where a conversion arrival intensity exists by assumption and is subsequently
modeled as a function of latent state variables or predictors of future
conversions. This approach is appealing for its tractability but is
difficult to apply empirically for the simple reason that conversions
have not yet occurred in practice, making the latent variable approach
untestable in practice as of the date of writing this article. However
it could be useful in linking observed credit spreads on CoCos to
presumed drivers of conversion arrival intensities. A third approach,
described in \citet{Wilkens2014}, is the equity derivative approach
where one tries to replicate the CoCo pay off by using equity derivatives
directly. The CoCo is seen as a straight bond plus Knock-in Forwards
minus Binary Down-in options. The long position in Knock-in Forwards
correspond to the possible purchase of shares at the stipulated conversion
price in case the trigger event takes place (i.e.\ when the forwards
knock in). The short position in Binary Down-in options reflects the
loss of (parts of) the coupon payments once the trigger event occurs.
A shortcoming of this approach is that it assumes the investor receives
forwards at conversion; but in an equity converter CoCo the investor
receives shares, not forwards. This is a significant difference when
the trigger event happens a substantial time before expiration of
the CoCo. This matters since CoCos must be perpetuals to qualify as
capital. If dividends are expected to be low for a substantial period
of time after the trigger event occurs, this may not be a major shortcoming.
A bigger problem with both the credit derivative and the equity derivative
approach is that they unavoidably have to assume trigger events conditional
on market price based triggers, which is counterfactual. A third problem
with the two derivative based approaches is that if they are to remain
analytically tractable, one has to assume a Black-Scholes setting
for the market price which cannot easily handle the fat tail risk
observed in CoCo prices and cannot incorporate the link between fat
tail risk and accounting reports observed in practice, cf.~\citet{Kiewiet_ea_2017}.
In~\citet{Corcuera2013} it is shown that this problem can be addressed
by analyzing an equity derivative based model using ''smile conform''
exponential L\'evy processes for stock price dynamics, incorporating
jumps and fat tails, but this approach unavoidably comes at the cost
of having to replace analytical closed form solutions by simulation
based solutions.

\section{The Model}

This section starts with the model description. After that we derive the density of asset values, conditional on accounting information. Finally we present results for the valuation of CoCos for the different trigger events: regulatory triggers,
possibly with conversion into shares,
and accounting triggers.

\label{sec:model} \setcounter{equation}{0}

\subsection{Model Description and the Firm's Debt Structure }

\label{section2_model_desc} The value of assets of the firm, denoted
by $V_{t}$, is modeled by a geometric Brownian motion, that is 
\begin{equation}
\frac{\dd V_{t}}{V_{t}}=\mu\dd t+\sigma\dd W_{t},\label{eq:model}
\end{equation}
for some $\mu\in\mathbb{R}$, $\sigma>0$. We will not explicitly include jumps in the asset value process, because this does not make sense in the noisy accounting information framework, as it would not be possible to distinguish a big price movement caused by the dynamics of the asset process from a reaction to the accounting information.
  Define $Z_{t}=\log V_{t}$
and $m=\mu-\sigma^{2}/2$, then $Z$ is a drifted Brownian motion
with drift $m$ and volatility $\sigma$, that is 
\[
Z_{t}=Z_{0}+mt+\sigma W_{t}.
\]
As mentioned before, we will consider a framework in which investors
do not observe the real asset value, instead they receive imperfect
accounting information at known observation times $t_{1}<t_{2}<\dots$
(typically every three months). At every observation date $t_{i}$
there arrives an imperfect accounting report of the real asset value
$V_{t_{i}}$, denoted by $\hat{V}_{t_{i}}$, where $\log\hat{V}_{t_{i}}$
and $\log V_{t_{i}}$ are assumed to be joint normal. This means that
we can write 
\begin{equation}
Y_{t_{i}}:=\log\hat{V}_{t_{i}}=Z_{t_{i}}+U_{t_{i}},\label{eq:acc}
\end{equation}
where $U_{t_{i}}$ is normally distributed and independent of $Z_{t_{i}}$.
In the following we will use the notation $Y_{i}:=Y_{t_{i}}$ and
similar notations for $Z$ and $U$. Of course, it is reasonable that
there exists some correlation between the accounting noise $U_{1},U_{2},\dots$.
To be more specific, following \citet{Duffie2001}, it is assumed
that 
\[
U_{i}=\kappa U_{i-1}+\epsilon_{i},
\]
for some fixed $\kappa\in\mathbb{R}$ and independent and indentically
distributed $\epsilon_{1},\epsilon_{2},\dots$, which have a normal
distribution with mean $\mu_{\epsilon}\in\mathbb{R}$ and variance
$\sigma_{\epsilon}^2>0$, and are independent of $Z$. \\
 It is assumed that the firm issues two types of debt; straight debt
and contingent convertible debt. The total par value of straight debt
outstanding is denoted by $P_{1}$, over which coupons are paid continuously
at rate $c_{1}$. Furthermore, the straight bonds have a perpetual
maturity and it is assumed that default occurs the first time the
$\log$-value of assets falls below some threshold $z_{b}$, such
that the default time is defined by 
\[
\tau_{b}=\inf\{t\geq0:Z_{t}\leq z_{b}\}.
\]
At the moment of default a fraction $(1-\alpha)$, for $\alpha\in(0,1)$,
of the firm's asset value is lost to bankruptcy costs, so a fraction
$\alpha$ of the asset value is recovered and distributed among the
senior debt holders.\\
 \indent The total par value of CoCos outstanding is denoted by $P_{2}$,
over which coupons are paid continuously at rate $c_{2}$. Furthermore,
the maturity of the contingent convertible bonds is denoted by $T$.
In our accounting report framework, we will consider two different
types of conversion triggers. The first type of conversion trigger
that will be looked into is the regulatory trigger. Banks have the
obligation to report it to their supervisor at the moment they are
approaching a trigger. Then the regulator will call for conversion,
this is called a Point of Non-Viability. Of course, this type of conversion
can also happen in between accounting report dates. This type of conversion
thus is triggered when the $\log$-value of assets falls for the first
time below a conversion threshold $z_{c}$, i.e.\ the conversion
time is given by 
\[
\tau_{c}=\inf\{t\geq0:Z_{t}\leq z_{c}\}.
\]
We will always assume that $z_{b}<z_{c}$, such that conversion will
always happen before default, i.e.\ $\tau_{c}<\tau_{b}$.\\
 There are also also CoCos whose conversion trigger solely depends
on accounting reports. An example is the Coco issued by Barclays on
March 3, 2017 (cf.\ \citet{Barclays2017}). This means that conversion
happens when the reported value of the capital ratio falls below some
threshold and hence conversion can only happen at one of the accounting
report dates $t_{1},t_{2},\dots$. This corresponds to a setting in
which the conversion time is defined as 
\[
\tau_{c}^{A}=\inf\{t_{i}\geq0:Y_{t_{i}}\leq y_{c}\},
\]
for some threshold $y_{c}\geq0$.\\
 In case we consider CoCos with regulatory triggers, the information
available to investors at time $t$ is described by the filtration
$\mathcal{H}_{t}$, where 
\[
\mathcal{H}_{t}=\sigma(\{Y_{t_{1}},\dots Y_{t_{n}},\mathbf{1}_{\{\tau_{c}\leq s\}},\mathbf{1}_{\{\tau_{b}\leq s\}}:s\leq t\}),\text{ for }t_{n}\leq t<t_{n+1}.
\]
Here, the indicators are included to ensure that it is also observed
in the market whether conversion has already occured or the firm is
liquidated before time $t$. In case we deal with CoCos with an accounting
trigger, the market information is described by the filtration 
\[
\mathcal{H}_{t}=\sigma(Y_{t_{1}},\dots Y_{t_{n}}),\text{ for }t_{n}\leq t<t_{n+1}.
\]

\subsection{The Density of Asset Value, Conditional on Accounting Information}

\label{section2_cond_dens} In the previous subsection it was explained
that we will consider two different types of conversion triggers.
For the first one, the regulatory trigger, the conversion time is
determined by the process $Z$ falling below some threshold. In order
to compute the market value of CoCos with such a trigger, we need
to be able to compute the probability of conversion, conditional on
the market information $\mathcal{H}_{t}$. In order to do so, we will
need the conditional density of $Z$, given the market information
$\mathcal{H}_{t}$. In this subsection, following \citet{Duffie2001},
we will derive an expression for this conditional density, which is
intensively used in the remainder of this article.\\
 \indent Consider $t>0$ such that $t_{n}\leq t<t_{n+1}$ and conversion
did not happen until time $t$, that is $\tau_{c}>t$. The goal in
this section is to find an expression for the conditional distribution
of $Z_{t}$, given $\mathcal{H}_{t}$, which we will denote by $f(t,\cdot)$.
Most of the results in this section can be found in the article by \citet{Duffie2001},
but we will consider them shortly, to illustrate how the particular
density is derived and we will provide some additional explicit formulas.\\
 Consider the following notation for the relevant random vectors and
its realisations: 
\begin{align*}
 & Z^{(n)}=(Z_{1},Z_{2},\dots,Z_{n})\text{ and its realisation }z^{(n)}=(z_{1},z_{2},\dots,z_{n}),\\
 & Y^{(n)}=(Y_{1},Y_{2},\dots,Y_{n})\text{ and its realisation }y^{(n)}=(y_{1},y_{2},\dots,y_{n}),\\
 & U^{(n)}=Y^{(n)}-Z^{(n)}\text{ and its realisation }u^{(n)}=y^{(n)}-z^{(n)}.
\end{align*}

As already mentioned, the goal is to compute $f(t,\cdot)$, the conditional
density of $Z_{t}$ given $Y^{(n)}$ and $\tau_{c}>t$. In order to
do so, we will first compute the conditional density of $Z_{t_{n}}$
at the report time $t_{n}$, which we will denote by $g_{t_{n}}(\cdot|Y^{(n)},\tau_{c}>t_{n})$.
To this end, we will first introduce some functions. Firstly, we need
an expression for the probability $\psi(z_{0},x,\sigma\sqrt{t})$
that $\min\{Z_{s}:s\leq t\}>0$, conditional on $Z_{0}=z_{0}>0$ and
$Z_{t}=x>0$. This expression is stated in the following lemma and
can also be found in the paper by \citet{Duffie2001}. \begin{lemma}
\label{lemmapsi} The probability $\psi(z_{0},x,\sigma\sqrt{t})$
that $\min\{Z_{s}:s\leq t\}>0$, conditional on $Z_{0}=z_{0}>0$ and
$Z_{t}=x>0$, is given by 
\[
\psi(z_{0},x,\sigma\sqrt{t})=1-\exp\left(-\frac{2z_{0}x}{\sigma^{2}t}\right).
\]
\end{lemma}
Consider the conditional probability of the intersection $\{Z^{(n)}\leq z^{(n)}\}\cap \{\tau_{c}>t_{n}\}$ given $Y^{(n)}$.  We denote by $b_{n}(\cdot|Y^{(n)})$ its partial derivative w.r.t.\ $z^{(n)}$.
Note that $(Z_{n})_{n\in\mathbb{N}}$ and $(U_{n})_{n\in\mathbb{N}}$
are Markov processes and denote by $p_{Z}(z_{n}|z_{n-1})$ and $p_{U}(u_{n}|u_{n-1})$
their respective transition densities for realisations $z^{(n)},u^{(n)}$.
Furthermore, denote by $p_{Y}(y_{n}|y^{(n-1)})$ the conditional density
of $Y_{n}$ given $Y^{(n-1)}=y^{(n-1)}$. It is then possible (see \citet{Duffie2001}) to write $b_{n}(z^{(n)}|y^{(n)})$
in a recursive way,  
\begin{align}
\lefteqn{b_{n}(z^{(n)}|Y^{(n)})= } \nonumber\\ 
& \frac{\psi(z_{n-1}-z_{c},z_{n}-z_{c},\sigma\sqrt{t_{n}-t_{n-1}})p_{Z}(z_{n}|z_{n-1})p_{U}(y_{n}-z_{n}|y_{n-1}-z_{n-1})b_{n-1}(z^{(n-1)}|y^{(n-1)})}{p_{Y}(y_{n}|y^{(n-1)})}.\label{killeddensity}
\end{align}
It now follows that the conditional density $g_{t_{n}}(\cdot|Y^{(n)},\tau_{c}>t_{n})$
of $Z^{(n)}$ is given by 
\begin{align}
g_{t_{n}}(z^{(n)}|y^{(n)},\tau_{c}>t_{n})=\frac{b_{n}(z^{(n)}|y^{(n)})}{\int_{(z_{c},\infty)^{n}}b_{n}(z^{(n)}|y^{(n)})\dd z^{(n)}}.\label{finaldensity}
\end{align}
It should be noted that there is no explicit expression for the integral
in the denominator of Equation~(\ref{finaldensity}), but note the
important fact that we know the density up to a normalizing constant.
Now the marginal conditional density of $Z_{n}$ at time $t_{n}$
is given by 
\begin{align}
g_{t_{n}}(z_{n}|y^{(n)},\tau_{c}>t_{n})=\int_{(z_{c},\infty)^{n-1}}g_{t_{n}}(z^{(n)}|y^{(n)},\tau_{c}>t_{n})\dd z^{(n-1)}.\label{marginalg}
\end{align}
Now that we found the conditional density for a report time $t_{n}$,
we can use this to find the conditional density $f(t,\cdot)$ for
a general time $t>0$. For this we will need the $\mathcal{H}_{t}$-conditional
density of $Z_{t}$, at a time before the first accounting report
has arrived. Complementing \citet{Duffie2001}, we will now give
an explicit expression for this density.

\begin{lemma} \label{ftildelemma}
$\tilde{f}(t,\cdot,z_{0})$, the $\mathcal{H}_{t}$-conditional density
of $Z_{t}$, at a time $t<\tau_{c}$ before the first accounting report
has arrived, given that $Z$ started in $z_{0}$, is given by 
\begin{align}
\tilde{f}(t,x,z_{0})=\frac{1}{\sigma\sqrt{t}}\frac{\exp\left(\frac{-m(z_{0}-x)}{\sigma^{2}}-\frac{m^{2}t}{2\sigma^{2}}\right)\left(\phi\left(\frac{z_{0}-x}{\sigma\sqrt{t}}\right)-\phi\left(\frac{-z_{0}-x+2z_{c}}{\sigma\sqrt{t}}\right)\right)}{\Phi\left(\frac{z_{0}-z_{c}+mt}{\sigma\sqrt{t}}\right)-e^{-2m(z_{0}-z_{c})/\sigma^{2}}\Phi\left(\frac{z_{c}-z_{0}+mt}{\sigma\sqrt{t}}\right)}.\label{ftilde}
\end{align}
\end{lemma} \begin{proof} The proof of this lemma can be found in
the Appendix. \end{proof} Finally, we are now able time to compute
the conditional density $f(t,\cdot)$ for a general time $t>0$, $t_{n}<t<t_{n+1}$
such that $\tau_{c}>t$. Using the stationarity of $Z$, the $\mathcal{H}_{t}$-conditional
density of $Z_{t}$ can be written as 
\begin{align}
f(t,x)=\int_{z_{c}}^{\infty}\tilde{f}(t-t_{n},x,z_{n})g_{t_{n}}(z_{n}|Y^{(n)},\tau_{c}>t_{n})\dd z_{n}.\label{generaldensity}
\end{align}
Equation (\ref{generaldensity}) should be read as follows; until
time $t_{n}$ the process $Z$ has stayed above $z_{c}$ and ended
in $z_{n}$, then on the time interval $(t_{n},t)$, in which no new
accounting reports arrive, the process has to move from $z_{n}$ to
$x$ and stay above $z_{c}$. Although we do not have an analytical
expression for the density $f(t,\cdot)$, it is important to note
at this point that $f(t,\cdot)$ is written as the integral of $g_{t_{n}}$,
which is known up to normalizing constant, as can be seen from Equation~(\ref{finaldensity}).
This makes it possible to compute integrals with respect to $f(t,\cdot)$,
using Monte Carlo Markov Chain simulations, which means that results
that are stated as an integral weighted by the density $f(t,\cdot)$
can actually be computed. The necessary algorithms are described in
Section~\ref{Section_Simulations}.\\
 As a first use of the density $f(t,\cdot)$, we can for a time $s>t$,
where $t<\tau_{c}$, define the $\mathcal{H}_{t}$-(CoCo) survival
probability $p_{c}(t,s)=\mathbb{P}(\tau_{c}>s|\mathcal{H}_{t})$.
This probability is then given by 
\begin{align}
p_{c}(t,s)=\int_{z_{c}}^{\infty}(1-\pi(s-t,x-z_{c}))f(t,x)\dd x,\label{survprobCoCo}
\end{align}
where, as in \citet{Duffie2001}, $\pi(t,x)$ denotes the probability
that $Z$ hits $0$ before time $t$, starting from $x>0$. This probability
is given by the following lemma, which follows from the well known
expression for the distribution of a Brownian motion's running minimum (see e.g.\ \citet{Harrison1985}, Section 1.8, equation (11)).
\begin{lemma} \label{pi} The probability $\pi(t,x)$ that $Z$ hits
$0$ before time $t$, starting from $x>0$, is given by 
\[
\pi(t,x)=1-\Phi\left(\frac{x+mt}{\sigma\sqrt{t}}\right)+e^{-2mx/\sigma^{2}}\Phi\left(\frac{-x+mt}{\sigma\sqrt{t}}\right).
\]
\end{lemma} 

\subsection{Valuation of CoCos}

In this subsection we will provide formulas for the market values
of the different type of CoCos. Firstly, in subsection~\ref{subsec_PWDreg},
CoCos with a regulatory trigger which suffer a principal write down
at conversion, are valued. Then, it is also shown how to incorporate
early cancelling of coupons, due to the MDA-regulations. In subsection
\ref{subsec_ce_reg}, we then extend the PWD-assumption to CoCos with
a conversion into shares. These first two cases are all for the regulatory
trigger and the results are all in the form of an integral weighted
by the the above derived conditional density $f(t,\cdot)$. It is
postponed to Section~\ref{Section_Simulations} to provide the necessary
Algorithms to compute this integrals. Then, in Section~\ref{subsec_acctrig},
PWD CoCos with only an accounting trigger are valued. 

\subsubsection{Valuation of PWD CoCos with a regulatory trigger}

\label{subsec_PWDreg} In this section we will value CoCos with a
regulatory trigger and a principal write down at conversion. At the
end of the section we will also incorporate the MDA-trigger. Recall that in case of a regulatory trigger, the conversion date
was defined as 
\[
\tau_{c}=\inf\{t\geq0:Z_{t}\leq z_{c}\}.
\]
Also, recall that the firm pays coupons continuously at rate $c_{2}$
until either maturity or conversion. We consider a principal write
down CoCo, which means a fraction $1-R$ of the principal value is
written down at conversion, while a fraction $R$ is recovered to
the bond holder, for $R\in[0,1)$. Furthermore it is assumed that
the risk free rate is constant, denoted by $r$. \\
 Now the value at time $t<\tau_{c}$ of the CoCos, given the imperfect
accounting information $\mathcal{H}_{t}$, is given by 
\begin{align}\label{CoCopricedufflando}
C(t) & =\mathbb{E}\left(P_{2}e^{-r(T-t)}\mathbf{1}_{\{\tau_{c}>T\}}|\mathcal{H}_{t}\right)+\mathbb{E}\left(\int_{t}^{T}c_{2}P_{2}e^{-r(u-t)}\mathbf{1}_{\{\tau_{c}>u\}}\dd u|\mathcal{H}_{t}\right)\nonumber \\
 & \quad+\mathbb{E}\left(RP_{2}e^{-r(\tau_{c}-t)}\mathbf{1}_{\{\tau_{c}\leq T\}}|\mathcal{H}_{t}\right)\nonumber \\
 & =P_{2}e^{-r(T-t)}p_{c}(t,T)+c_{2}P_{2}\int_{t}^{T}e^{-r(u-t)}p_{c}(t,u)\dd u-RP_{2}\int_{t}^{T}e^{-r(u-t)}p_{c}(t,\dd u).
\end{align}
Here the first term represents the payment of the principal, in case
conversion does not happen before maturity, while the second term
accounts for the payment of coupons until either conversion or maturity.
The last term values the recovery of the principal at conversion.
Note that every term is written in terms of the CoCo survival probability
$p_{c}(t,s)$, which was given as an integral, weighted by the density
$f(t,\cdot)$. Unsurprisingly, it turns out that these three terms
together can be written as one integral weighted by the conditional
density $f(t,\cdot)$, which was derived in the previous section.
This leads to the main result of this subsection, which is proved
in the Appendix. \begin{theorem}[Price of a PWD CoCo with a regulatory trigger]
\label{thmCoCoprice1} The secondary market price of the CoCo at time
$t<\tau_{c}$ is given by 
\begin{align}
C(t)=\int_{z_{c}}^{\infty}h(x)f(t,x)\dd x,\label{CoCoprice2}
\end{align}
where $h(x)$ is defined as 
\begin{align}
h(x):=\frac{r-c}{r}P_{2}e^{-r(T-t)}(1-\pi(T-t,x-z_{c}))+\frac{c_{2}}{r}P_{2}+\left(\frac{c_{2}P_{2}}{r}-RP_{2}\right)I(x),\label{hfunction}
\end{align}
in which $I(x)$ is given by 
\begin{align}\label{I(x)}
I(x) & :=\exp\left(-\frac{m(x-z_{c})+(x-z_{c})\sqrt{m^{2}+2r\sigma^{2}}}{\sigma^{2}}\right)\left(\Phi\left(\frac{x-z_{c}-\sqrt{m^{2}+2r\sigma^{2}}(T-t)}{\sigma\sqrt{T-t}}\right)-1\right)\nonumber \\
 & ~~~+\exp\left(-\frac{m(x-z_{c})-(x-z_{c})\sqrt{m^{2}+2r\sigma^{2}}}{\sigma^{2}}\right)\left(\Phi\left(\frac{x-z_{c}+\sqrt{m^{2}+2r\sigma^{2}}(T-t)}{\sigma\sqrt{T-t}}\right)-1\right),
\end{align}
where $\Phi$ denotes the normal cumulative distribution function.
\end{theorem} In the valuation of the firm's convertible debt in
Equation~(\ref{CoCopricedufflando}), it is assumed that coupons
are paid until conversion. However, as pointed out before, CoCos are
affected by the Maximum Distributable Amount (MDA), which requires
regulators to stop earnings distributions when the firm's total capital
falls below some trigger, higher than the conversion trigger. This
we will incorporate in the model by introducing a trigger $z_{cc}>z_{c}$.
If $Z$ is below $z_{cc}$ the firm will not pay coupons, while if
$Z$ is above $z_{cc}$ the firm still pays coupons. To value the
CoCo in this case, only the second term in Equation~(\ref{CoCopricedufflando})
needs to be adjusted. In this case, coupons are only paid at time
$u$ if $Z_{u}>z_{cc}$, so the term 
\[
\mathbb{E}\left(\int_{t}^{T}c_{2}P_{2}e^{-r(u-t)}\mathbf{1}_{\{\tau_{c}>u\}}\dd u|\mathcal{H}_{t}\right),
\]
needs to be replaced with 
\begin{align}
\mathbb{E}\left(\int_{t}^{T}c_{2}P_{2}e^{-r(u-t)}\mathbf{1}_{\{\tau_{c}>u,Z_{u}>z_{cc}\}}\dd u|\mathcal{H}_{t}\right).\label{newterm}
\end{align}
For $\tau_{c}>t$ and $t_{n}\leq t<t_{n+1}$ this term equals 
\[
c_{2}P_{2}\int_{t}^{T}e^{-r(u-t)}\mathbb{P}\left(\tau_{c}>u,Z_{u}>z_{cc}|Y^{(n)},\tau_{c}>t\right)\dd u.
\]
Thus, to value the CoCos while including the effects of the MDA-trigger,
the quantity we need to compute is $\mathbb{P}\left(\tau_{c}>u,Z_{u}>z_{cc}|Y^{(n)},\tau_{c}>t\right)$,
which can be written in a similar way as the CoCo survival probability
$p_{c}(t,s)$. In order to compute this conditional probability, we
first need the following well known result :
the joint distribution of a drifted Brownian motion and its running
minimum (see e.g.~\citet{Harrison1985}, Section 1.8, Corollary 7). \begin{lemma} The joint probability $\tilde{\pi}(t,x,y)$
that $Z$, starting from $x>0$, does not hit 0 before time $t$ and
that $Z_{t}>y$ is given by 
\begin{align}
\tilde{\pi}(t,x,y):=\mathbb{P}(\inf_{0\leq s\leq t}Z_{s}>0,Z_{t}>y)=\Phi\left(\frac{x-y+mt}{\sigma\sqrt{t}}\right)-e^{-2mx/\sigma^{2}}\Phi\left(\frac{-x-y+mt}{\sigma\sqrt{t}}\right).\label{pitilde}
\end{align}
\end{lemma} Now, similarly to Equation~(\ref{survprobCoCo}), we
can write 
\begin{align}
\mathbb{P}\left(\tau_{c}>u,Z_{u}>z_{cc}|Y^{(n)},\tau_{c}>t\right)=\int_{z_{c}}^{\infty}\tilde{\pi}(u-t,x-z_{c},z_{cc}-z_{c})f(t,x)\dd x,\label{MDA_reg_PWD}
\end{align}
such that we again found the solution as an integral weighted by the
density $f(t,\cdot)$. The other two terms in Equation~(\ref{CoCopricedufflando})
do not change, so the CoCo price at time $t<\tau_{c}$ is given by
the sum of the new term in (\ref{newterm}) and the unchanged part
\[
Pe^{-r(T-t)}p_{c}(t,T)-RP\int_{t}^{T}e^{-r(u-t)}p_{c}(t,\dd u).
\]
By an adaption of Equation~(\ref{CoCoprice2}) it is seen that this
unchanged part can be written as 
\[
\int_{z_{c}}^{\infty}\tilde{h}(x)f(t,x)\dd x,
\]
where 
\[
\tilde{h}(x)=Pe^{-r(T-t)}(1-\pi(T-t,x-z_{c}))-RPI(x),
\]
in which $I(x)$ is given by Equation~(\ref{I(x)}). 

\subsubsection{Valuation of CoCos with a conversion into shares and a regulatory
trigger}

\label{subsec_ce_reg} \label{secconversion} In this section we consider
the valuation of contingent convertible bonds which convert into equity
at the conversion date. To recall, we assumed the firm issues two
types of debt; straight debt and contingent convertible debt. The
total par value of straight debt outstanding is denoted by $P_{1}$,
over which coupons are paid continuously at rate $c_{1}$. Furthermore,
the straight bonds have a perpetual maturity and default occurs at
\[
\tau_{b}=\inf\{t\geq0:Z_{t}\leq z_{b}\}.
\]
At the moment of default a fraction $(1-\alpha)$, for $\alpha\in(0,1)$,
of the firm's asset value is lost to bankruptcy costs, so a fraction
$\alpha$ of the asset value is recovered and distributed among the
senior debt holders.\\
 \indent The total par value of CoCos outstanding is denoted by $P_{2}$,
over which coupons are paid continuously at rate $c_{2}$. Furthermore,
the maturity of the contingent convertible bonds is denoted by $T$.
We consider a regulatory trigger, which means the conversion date
is defined as 
\[
\tau_{c}=\inf\{t\geq0:Z_{t}\leq z_{c}\},
\]
where $z_{c}>z_{b}$, to ensure that conversion happens before default.
Following \citet{Chen2013}, we will assume the CoCo holders receive
$\Delta$ shares for every dollar of principal at the moment of conversion.
This means that, if we normalize the number of shares before conversion
to 1, the CoCo holders own a fraction $\rho = \frac{\Delta P_{2}}{\Delta P_{2}+1}$
of the firm's equity after conversion. \\
 To recall, the information in the market at time $t$ is described
by the filtration 
\[
\mathcal{H}_{t}=\sigma(\{Y_{t_{1}},\dots Y_{t_{n}},\mathbf{1}_{\{\tau_{c}\leq s\}},\mathbf{1}_{\{\tau_{b}\leq s\}}:s\leq t\}),\text{ for \ensuremath{t_{n}\leq t<t_{n+1}}}.
\]
In analogy to Equation~(\ref{CoCopricedufflando}), the market price
of the CoCos is given by 
\begin{align}\label{priceeqconv}
C(t) & =\mathbb{E}\left(P_{2}e^{-r(T-t)}\mathbf{1}_{\{\tau_{c}>T\}}|\mathcal{H}_{t}\right)+\mathbb{E}\left(\int_{t}^{T}c_{2}P_{2}e^{-r(u-t)}\mathbf{1}_{\{\tau_{c}>u\}}\dd u|\mathcal{H}_{t}\right)\nonumber \\
 & ~~~+\mathbb{E}\left(\frac{\Delta P_{2}}{\Delta P_{2}+1}E^{PC}(\tau_{c})e^{-r(\tau_{c}-t)}\mathbf{1}_{\{\tau_{c}\leq T\}}|\mathcal{H}_{t}\right).
\end{align}
Of course only the third term has changed compared to Equation~(\ref{CoCopricedufflando}),
because this term describes what happens at the moment of conversion
(note that the second term needs to be replaced by the corresponding
term in Equation~(\ref{newterm}), if we want to include early cancelling
of coupons). The third term now describes that the CoCo holders obtain
a fraction $\frac{\Delta P_{2}}{\Delta P_{2}+1}$ of the firms post-conversion
equity, denoted by $E^{PC}(\tau_{c})$. This post conversion equity
satisfies 
\[
E^{PC}(\tau_{c})=V_{\tau_{c}}-D(\tau_{c})-\mathbb{E}\left(e^{-r(\tau_{b}-\tau_{c})}(1-\alpha)V_{\tau_{b}}|\mathcal{H}_{\tau_{c}}\right).
\]
That is, the firm's value of assets minus the value of straight debt,
denoted by $D(\tau_{c})$, and bankruptcy costs, described by the
last term. Note that the value of straight debt at conversion is given
by 
\[
D(\tau_{c})=\mathbb{E}\left(\int_{\tau_{c}}^{\infty}c_{1}P_{1}e^{-r(u-\tau_{c})}\mathbf{1}_{\{\tau_{b}>u\}}\dd u|\mathcal{H}_{\tau_{c}}\right)+\mathbb{E}\left(\alpha V_{\tau_{b}}e^{-r(\tau_{b}-\tau_{c})}|\mathcal{H}_{\tau_{c}}\right),
\]
where the first term accounts for the continuous payment of coupons
and the second term describes the payment at default. It follows that
the post-conversion equity value is given by 
\begin{align*}
E^{PC}(\tau_{c}) & =V_{\tau_{c}}-\mathbb{E}\left(\int_{\tau_{c}}^{\infty}c_{1}P_{1}e^{-r(u-\tau_{c})}\mathbf{1}_{\{\tau_{b}>u\}}\dd u|\mathcal{H}_{\tau_{c}}\right)-\mathbb{E}\left(e^{-r(\tau_{b}-\tau_{c})}V_{\tau_{b}}|\mathcal{H}_{\tau_{c}}\right)\\
 & =e^{z_{c}}-\mathbb{E}\left(\int_{t}^{\infty}c_{1}P_{1}e^{-r(u-\tau_{c})}\mathbf{1}_{\{\tau_{c}\leq u,\tau_{b}>u\}}\dd u|\mathcal{H}_{\tau_{c}}\right)-e^{z_{b}}\mathbb{E}\left(e^{-r(\tau_{b}-\tau_{c})}|\mathcal{H}_{\tau_{c}}\right).
\end{align*}
So for $\tau_{c}>t$, the third term in Equation~(\ref{priceeqconv})
can be written as 
\begin{align}
\hspace{2em} & \hspace{-2em}\mathbb{E}\left(\frac{\Delta P_{2}}{\Delta P_{2}+1}E^{PC}(\tau_{c})e^{-r(\tau_{c}-t)}\mathbf{1}_{\{\tau_{c}\leq T\}}|\mathcal{H}_{t}\right)\nonumber \\
 & =\frac{\Delta P_{2}}{\Delta P_{2}+1}e^{z_{c}}\int_{t}^{T}e^{-r(u-t)}\mathbb{P}(\tau_{c}\in\dd u|\tau_{c}>t,Y^{(n)})\nonumber \\
 & ~~~-\frac{\Delta P_{2}c_{1}P_{1}}{\Delta P_{2}+1}\int_{t}^{\infty}e^{-r(u-t)}\mathbb{P}(\tau_{c}\leq T\wedge u,\tau_{b}>u|\tau_{c}>t,Y^{(n)})\dd u\nonumber \\
 & ~~~-\frac{\Delta P_{2}}{\Delta P_{2}+1}e^{z_{b}}\int_{t}^{\infty}e^{-r(u-t)}\mathbb{P}(\tau_{c}\leq T,\tau_{b}\in\dd u|\tau_{c}>t,Y^{(n)}).
\end{align}
So in this case, the key to valuation is finding an expression for
the joint conditional distribution of $\tau_{c}$ and $\tau_{b}$,
as needed in the above integrals. These expressions can again be written
as (double) integrals, weighted by the density $f(t,\cdot)$. Which
leads to the following theorem, for which a proof is provided in the
Appendix. \begin{theorem}[Price of a CoCo with a regulatory trigger and a conversion into shares]
\label{thmCoCoprice2} The secondary market price at time $t<\tau_{c}$
of the CoCo with a regulatory trigger and a conversion into shares
is given by 
\begin{align}
C(t)=\int_{z_{c}}^{\infty}(h_{0}(x)+h_{1}(x))f(t,x)\dd x+\int_{z_{c}}^{\infty}\int_{z_{c}}^{\infty}f(t,x)\hat{f}(x,z_{c},\tilde{z},T-t)h_{2}(\tilde{z})\dd\tilde{z}\dd x,\label{CoCoprice3}
\end{align}
where $\hat{f}(x,y,\tilde{z},t_{2})$ is given by 
\begin{align}
\hat{f}(x,y,\tilde{z},t_{2})=\frac{1}{\sigma\sqrt{t_{2}}}\exp\left(\frac{-m(x-\tilde{z})}{\sigma^{2}}-\frac{m^{2}t_{2}}{2\sigma^{2}}\right)\left(\phi\left(\frac{x-\tilde{z}}{\sigma\sqrt{t_{2}}}\right)-\phi\left(\frac{-x-\tilde{z}+2y}{\sigma\sqrt{t_{2}}}\right)\right)\label{fhatfunctie}
\end{align}
and where 
\begin{align*}
 & h_{0}(x)=\frac{r-c_{2}}{r}P_{2}e^{-r(T-t)}(1-\pi(T-t,x-z_{c}))+\frac{c_{2}P_{2}}{r}+\frac{c_{2}P_{2}}{r}I(x),\\
 & h_{1}(x)=\frac{\Delta P_{2}}{\Delta P_{2}+1}\left(e^{z_{b}}J_{b}(x)+c_{1}P_{1}\tilde{I}(x)-c_{1}P_{1}\tilde{J}_{b}(x)-e^{z_{c}}I(x)\right),\\
 & h_{2}(\tilde{z})=\frac{\Delta P_{2}}{\Delta P_{2}+1}e^{-r(T-t)}(c_{1}P_{1}\tilde{J}_{b}(\tilde{z})-e^{z_{b}}J_{b}(\tilde{z})),
\end{align*}
in which $I(x)$ is given by Equation~(\ref{I(x)}), $\tilde{I}(x)$
equals 
\[
\tilde{I}(x)=-\frac{1}{r}e^{-r(T-t)}(1-\pi(T-t,x-z_{c}))+\frac{1}{r}+\frac{1}{r}I(x),
\]
$J_{b}(x)$ is given by 
\begin{align*}
J_{b}(x) & =-\exp\left(-\frac{m(x-z_{b})+(x-z_{b})\sqrt{m^{2}+2r\sigma^{2}}}{\sigma^{2}}\right)
\end{align*}
and $\tilde{J}_{b}(x)=\frac{1}{r}+\frac{1}{r}J_{b}(x)$.

\end{theorem} 

\subsubsection{Valuation of PWD CoCos with an accounting trigger}

\label{subsec_acctrig} In this section we will consider PWD CoCos
which conversion trigger solely depends on accounting reports, for
example the CoCos issued by Barclays. This means that conversion happens
when the reported value of the capital ratio falls below some threshold
and hence conversion can only happen at one of the accounting report
dates $t_{1},t_{2},\dots$. This corresponds to a setting in which
the conversion time is defined as 
\[
\tau_{c}^{A}=\inf\{t_{i}\geq0:Y_{t_{i}}\leq y_{c}\},
\]
for some threshold $y_{c}\geq0$. In this case the available information
at time $t$ would reduce to 
\[
\mathcal{H}_{t}=\sigma(Y_{t_{1}},\dots,Y_{t_{n}}),
\]
for the largest $n$ such that $t_{n}\leq t$. This means that we
are interested in the probability that, given $n$ accounting reports
and $\tau_{c}^{A}>t_{n}$, the $(n+i)$th accounting report will cause
a trigger event, for $i=1,2,\dots$. This probability is given in
the next proposition, of which the proof can be found in the Appendix.
\begin{proposition} \label{prop_istep} The $i$-step conditional
survival probability (concerning $\tau_{c}^{A}$), conditional on
$n$ previous accounting reports, is given by 
\begin{align}
\mathbb{P}\left(\tau_{c}^{A}>t_{n+i}|Y^{(n)}=y^{(n)}\right) & =\int_{\mathbb{R}^{n}}\mathbb{P}(\xi(z_{n})\in(y_{c},\infty)^{i})p_{Z}(z^{(n)}|y^{(n)})\dd z^{(n)},\label{survprobonlyrep}
\end{align}
where 
\begin{align}
p_{Z}(z^{(n)}|y^{(n)}) & =\frac{\prod_{i=1}^{n}p_{Z}(z_{i}|z_{i-1})p_{U}(y_{i}-z_{i}|y_{i-1}-z_{i-1})}{p_{Y}(y_{n}|y^{(n-1)})},\label{neededdensity}
\end{align}
and where $\xi(z_{n})$ denotes a multivariate normal distributed
random variable with mean vector $\hat{\mu}_{i}$ and covariance matrix
$\Sigma_{i}$, for which formulas, depending on $z_{n}$, are provided
in the Appendix. \end{proposition} Note that that $p_{Z}(z_{i}|z_{i-1})$
is a Gaussian density with mean $z_{i-1}+m\Delta t$ and variance
$\sigma^{2}\Delta t$ and that $p_{U}(u_{i}|u_{i-1})$ is a Gaussian
density with mean $\kappa u_{i-1}+\mu_{\epsilon}$ and variance $\sigma_{\epsilon}^{2}$.
We did not provide a formula for $p_{Y}(y_{n}|y^{n-1})$, but it turns
out in Section~\ref{Section_Simulations} that we do not need this
to compute the integral of Equation~(\ref{survprobonlyrep}).\\
 Using this proposition, it is now possible to value the contingent
convertible bond with, as before, principal $P_{2}$, continuous coupon
rate $c_{2}$, maturity $T$ and a principal write-down with recovery
rate $R$. As in Equation~(\ref{CoCopricedufflando}), this CoCo
has secondary market price 
\begin{align}\label{CoCopriceacc}
C'(t) & =\mathbb{E}\left(Pe^{-r(T-t)}\mathbf{1}_{\{\tau_{c}^{A}>T\}}|\mathcal{H}_{t}\right)+\mathbb{E}\left(\int_{t}^{T}cPe^{-r(u-t)}\mathbf{1}_{\{\tau_{c}^{A}>u\}}\dd u|\mathcal{H}_{t}\right)\nonumber \\
 & ~~~+\mathbb{E}\left(RPe^{-r(\tau_{c}^{A}-t)}\mathbf{1}_{\{\tau_{c}^{A}\leq T\}}|\mathcal{H}_{t}\right),
\end{align}
which can be written in terms of the above derived $i$-step survival
probability, as is stated in the next result, which is proved in the
Appendix. \begin{theorem}[Price of a PWD CoCo with a sole accounting trigger]\label{thmCoCoprice3}
For $t_{n}\leq t<t_{n+1}$, $T=t_{n+m}$ for some $m\in\mathbb{N}$
and $Y^{(n)}=y^{(n)}$, where $y_{i}>y_{c},1\leq i\leq n$, the market
price $C'(t)$ of the CoCos is given by 
\begin{align}\label{CoCopriceacc2}
C'(t) & =(1-R)Pe^{-r(T-t)}\mathbb{P}(\tau_{c}^{A}>t_{n+m}|Y^{(n)}=y^{(n)})\nonumber \\
 & ~~~~+\sum_{i=1}^{m-1}\left(\frac{cP}{r}-RP\right)(e^{-r(t_{n+i}-t)}-e^{-r(t_{n+i+1}-t)})\mathbb{P}(\tau_{c}^{A}>t_{n+i}|Y^{(n)}=y^{(n)})\nonumber \\
 & ~~~~+\frac{cP}{r}(1-e^{-r(t_{n+1}-t)})+RPe^{-r(t_{n+1}-t)},
\end{align}
\end{theorem} It should be noted that the only things left to compute
are the $i$-step survival probabilities $\mathbb{P}(\tau_{c}^{A}>t_{n+i}|Y^{(n)}=y^{(n)})$,
for which the formula is provided in Proposition~\ref{prop_istep}
in terms of an integral, which can be evaluated using the method described
in Section~\ref{Section_Simulations}.\\
 \indent As in the case of the regulatory trigger, we can also incorporate
the MDA-regulations, which imply that coupons are already cancelled
at a moment before the conversion date. It is now assumed that coupons
over the time interval $[t_{i},t_{i+1})$ are only paid if $Y_{i}>y_{cc}$,
for some trigger level $y_{cc}>y_{c}$. To value the CoCo in this
case, the second term in Equation~(\ref{CoCopriceacc}) needs to
be changed to 
\[
\mathbb{E}\left(\sum_{i=1}^{m-1}\int_{t_{n+i}}^{t_{n+i+1}}cPe^{-r(u-t)}\mathbf{1}_{\{\tau_{c}^{A}>u,Y_{n+i}>y_{cc}\}}\dd u+\mathbf{1}_{\{Y_{n}>y_{cc}\}}\int_{t}^{t_{n+1}}cPe^{-r(u-t)}\dd u\Big|\mathcal{H}_{t}\right),
\]
where $t_{n}\leq t<t_{n+1}$, $T=t_{n+m}$ for some $m\in\mathbb{N}$.
\\
 This leads us to the next result, stating the value of PWD CoCo with
an trigger, when we also take into account the early cancelling of
coupons, due to the MDA regulations. \begin{theorem}[Price of PWD CoCo with a sole accounting trigger, including MDA regulations]\label{thmCoCoprice4}
When we include the MDA trigger, the CoCo price of Equation~(\ref{CoCopriceacc2})
modifies into 
\begin{align}
C'(t) & =Pe^{-r(T-t)}\mathbb{P}(\tau_{c}^{A}>t_{n+m}|Y^{(n)}=y^{(n)})+\mathbf{1}_{\{Y_{n}>y_{cc}\}}\frac{cP}{r}(1-e^{-r(t_{n+1}-t)})\nonumber \\
 & ~~~~+\sum_{i=1}^{m-1}\frac{cP}{r}(e^{-r(t_{n+i}-t)}-e^{-r(t_{n+i+1}-t)})\mathbb{P}(\tau_{c}^{A}>t_{n+i},Y_{n+i}>y_{cc}|Y^{(n)}=y^{(n)})\nonumber \\
\nonumber \\
 & ~~~~+RP\sum_{i=1}^{m}e^{-r(t_{n+j}-t)}\left(\mathbb{P}(\tau_{c}^{A}>t_{n+i-1}|Y^{(n)}=y^{(n)})-\mathbb{P}(\tau_{c}^{A}>t_{n+i}|Y^{(n)}=y^{(n)})\right),\nonumber \\
\end{align}
where, similar to Equation~(\ref{survprobonlyrep}), 
\begin{align}
\mathbb{P}(\tau_{c}^{A}>t_{n+i},Y_{n+i}>y_{cc}|Y^{(n)}=y^{(n)}) & =\int_{\mathbb{R}^{n}}\mathbb{P}(\xi(z_n) \in(y_{c},\infty)^{i-1}\times(y_{cc},\infty))p_{Z}(z^{(n)}|y^{(n)})\dd z^{(n)}.\label{survprobonlyrep2}
\end{align}
\end{theorem}

\section{MCMC algorithms for simulating the model}

\label{Section_Simulations} \setcounter{equation}{0}

In this section the algorithms that are necessary to compute all the
derived CoCo values, are provided. The results in the previous section
contain three kind of expressions, for which three different algorithms
are proposed in this subsection. The first expressions we will consider
are those of the form 
\[
\int_{z_{c}}^{\infty}h(x)f(t,x)\dd x.
\]
That is, integrals of a function $h$, weighted by the density $f(t,\cdot)$.
This type of expression is needed in the valuation of a PWD CoCo with
a regulatory trigger (cf.\ Theorem~\ref{thmCoCoprice1}), when we
include the MDA trigger (cf.\ Equation~(\ref{MDA_reg_PWD})) and
in the first part of the formula for the value of a CoCo with a conversion
into shares (cf.\ Theorem~\ref{thmCoCoprice2}).\\
 First note that we can write 
\begin{align}\label{uitschrijven}
\int_{z_{c}}^{\infty}h(x)f(t,x)\dd x & =\int_{z_{c}}^{\infty}h(x)\int_{z_{c}}^{\infty}\tilde{f}(t-t_{n},x,z_{n})g_{t_{n}}(z_{n}|Y^{(n)},\tau_{c}>t_{n})\dd z_{n}\dd x\nonumber \\
 & =\int_{z_{c}}^{\infty}\int_{(z_{c},\infty)^{n}}h(x)\tilde{f}(t-t_{n},x,z_{n})g_{t_{n}}(z^{(n)}|Y^{(n)},\tau_{c}>t_{n})\dd z^{(n)}\dd x\nonumber \\
 & =\int_{(z_{c},\infty)^{n+1}}h(z_{n+1})\tilde{f}(t-t_{n},z_{n+1},z_{n})g_{t_{n}}(z^{(n)}|Y^{(n)},\tau_{c}>t_{n})\dd z^{(n+1)}.
\end{align}
So we will need a sample $((z^{(n+1)})^{1},\dots,(z^{n+1})^{G})$
from the $(n+1)$-dimensional distribution on $(z_{c},\infty)^{n+1}$
with density $\tilde{f}(t-t_{n},z_{n+1},z_{n})g_{t_{n}}(z^{(n)}|Y^{(n)},\tau_{c}>t_{n})$,
in order to approximate this integral as 
\begin{align}
C(t)\approx\frac{1}{G}\sum_{g=1}^{G}h(z_{n+1}^{g}).\label{capprox2}
\end{align}
The algorithm used to obtain the sample, is the following MCMC-algorithm.
\newpage
\begin{algorithm}
\label{rwmh}\hfill{}
\begin{enumerate}
\item In each iteration $g$, $g=1,\dots n_{0}+G$, given the current value
$(z^{(n+1)})^{g}$, the proposal $(z^{n+1})'$ is drawn according
to 
\[
(z^{(n+1)})'=(z^{(n+1)})^{g}+X,\text{ for }X\sim\text{N}_{n+1}(0,\Sigma),
\]
where the $(n+1)\times(n+1)$-covariance matrix $\Sigma$ is chosen
to reach some desired acceptance rate. 
\item Set 
\[
(z^{(n+1)})^{(g+1)}=\left\{ \begin{matrix}(z^{(n+1)})' & \text{ with prob.} & \alpha((z^{(n+1)})^{g},(z^{(n+1)})')\\
z^{(n+1)} & \text{ with prob.} & 1-\alpha((z^{(n+1)})^{g},(z^{(n+1)})')
\end{matrix}\right.,
\]
where the acceptance-probability $\alpha(z^{(n+1)},(z^{(n+1)})')$
is given by 
\begin{align*}
\alpha(z^{(n+1)},(z^{(n+1)})') & =\min\left\{ 1,\frac{\tilde{f}(t-t_{n},z_{n+1}',z_{n}')g_{t_{n}}((z^{(n)})'|y^{(n)},\tau_{c}>t_{n})}{\tilde{f}(t-t_{n},z_{n+1},z_{n})g_{t_{n}}(z^{(n)}|y^{(n)},\tau_{c}>t_{n})}\right\} \\
 & =\min\left\{ 1,\frac{\tilde{f}(t-t_{n},z_{n+1}',z_{n}')b_{n}((z^{(n)})'|y^{(n)})}{\tilde{f}(t-t_{n},z_{n+1},z_{n})b_{n}(z^{(n)}|y^{(n)})}\right\} .
\end{align*}
\item Discard the draws from the first $n_{0}$ iterations and save the
sample 
 $(z^{(n+1)})^{n_{0}+1},\dots,(z^{(n+1)})^{n_{0}+G}$. 
\end{enumerate}
\end{algorithm}

\noindent The acceptance probability involves the term ${\displaystyle \frac{b_{n}((z^{(n)})'|y^{(n)})}{b_{n}(z^{(n)}|y^{(n)})}}$.
It follows from Equation~(\ref{killeddensity}) that this fraction
is explicitly given by 
\begin{align*}
\frac{b_{n}((z^{(n)})'|y^{(n)})}{b_{n}(z^{(n)}|y^{(n)})}=\frac{\prod_{i=1}^{n}\psi(z'_{i-1}-z_{c},z'_{i}-z_{c},\sigma\sqrt{t_{i}-t_{i-1}})p_{Z}(z_{i}'|z'_{i-1})p_{U}(y_{i}-z'_{i}|y_{i-1}-z'_{i-1})}{\prod_{i=1}^{n}\psi(z_{i-1}-z_{c},z_{i}-z_{c},\sigma\sqrt{t_{i}-t_{i-1}})p_{Z}(z_{i}|z_{i-1})p_{U}(y_{i}-z_{i}|y_{i-1}-z_{i-1})},
\end{align*}
under the convention that $t_{0}=0$ and $p_{U}(\cdot|u_{0})=p_{U}(\cdot)$
is a Gaussian density with mean $\mu_{\epsilon}$ and variance $\sigma_{\epsilon}^{2}$.
Note that $p_{Z}(z_{i}|z_{i-1})$ is a Gaussian density with mean
$z_{i-1}+m(t_{i}-t_{i-1})$ and variance $\sigma^{2}(t_{i}-t_{i-1})$,
that $p_{U}(u_{i}|u_{i-1})$ is a Gaussian density with mean $\kappa u_{i-1}+\mu_{\epsilon}$
and variance $\sigma_{\epsilon}^{2}$ and that an expression for $\psi$
is provided in Lemma~\ref{lemmapsi}. Algorithm~\ref{rwmh}, in combination
with Equations (\ref{capprox2}) and (\ref{uitschrijven}), allows
us two compute all expressions which are of the form of an integral
of a function, weighted by the density $f(t,\cdot)$.\\
 \indent The second expression that occurs in the valuation of CoCos
in the previous section, is the expression we see in the second part
of the solution for a CoCo with a regulatory trigger and a conversion
into shares, as in Theorem~\ref{thmCoCoprice2}. Which is the following
double integral 
\[
\int_{z_{c}}^{\infty}\int_{z_{c}}^{\infty}f(t,x)\hat{f}(x,z_{c},\tilde{z},T-t)h_{2}(\tilde{z})\dd\tilde{z}\dd x.
\]
Note that this integral can be, similarly to the above, written as
\begin{align}\label{lastterm}
\hspace{2em} & \hspace{-2em}\int_{z_{c}}^{\infty}\int_{(z_{c},\infty)^{n+1}}h_{2}(\tilde{z})\hat{f}(z_{n+1},z_{c},\tilde{z},T-t)\tilde{f}(t-t_{n},z_{n+1},z_{n})g_{t_{n}}(z^{(n)}|Y^{(n)},\tau_{c}>t_{n})\dd z^{(n+1)}\dd\tilde{z}\nonumber \\
 & =\int_{(z_{c},\infty)^{n+2}}h_{2}(z_{n+2})\hat{f}(z_{n+1},z_{c},z_{n+2},T-t)\tilde{f}(t-t_{n},z_{n+1},z_{n})g_{t_{n}}(z^{(n)}|Y^{(n)},\tau_{c}>t_{n})\dd z^{(n+2)}.
\end{align}
Now note that, by definition of $\hat{f}$, it holds that 
\begin{align*}
\int_{z_{c}}^{\infty}\hat{f}(z_{n+1},z_{c},z_{n+2},T-t)\dd z_{n+2} & =\int_{z_{c}}^{\infty}\mathbb{P}\left(\inf_{0\leq s\leq T-t}Z_{s}>z_{c},Z_{T-t}\in\dd z_{n+2}\Big|Z_{0}=z_{n+1}\right)\\
 & =\mathbb{P}\left(\inf_{0\leq s\leq T-t}Z_{s}>z_{c}|Z_{0}=z_{n+1}\right)\\
 & =1-\pi(T-t,z_{n+1}-z_{c}).
\end{align*}
Hence, $\hat{f}(z_{n+1},z_{c},z_{n+2},T-t)\tilde{f}(t-t_{n},z_{n+1},z_{n})g_{t_{n}}(z^{(n)}|Y^{(n)},\tau_{c}>t_{n})$
is not a density function on $(z_{c},\infty)^{n+2}$, so it is not
possible to proceed in the same way as in the previous case. However,
by the above we know that 
\[
\frac{\hat{f}(z_{n+1},z_{c},z_{n+2},T-t)\tilde{f}(t-t_{n},z_{n+1},z_{n})g_{t_{n}}(z^{(n)}|Y^{(n)},\tau_{c}>t_{n})}{1-\pi(T-t,z_{n+1}-z_{c})}
\]
is a density function on $(z_{c},\infty)^{n+2}$. \\
 \indent So if we have a sample $((z^{(n+2)})^{1},\dots,(z^{n+2})^{G})$
from the $(n+2)$-dimensional distribution with this density, we can
approximate the integral in Equation~(\ref{lastterm}) by 
\begin{align}
\frac{1}{G}\sum_{g=1}^{G}h_{2}(z_{n+2}^{g})(1-\pi(T-t,z_{n+1}^{g}-z_{c})).\label{approx2ndterm}
\end{align}
This sample is, in analogy to Algorithm~\ref{rwmh}, obtained by the
following MCMC-algorithm. 
\begin{algorithm}
\label{rwmh3} \hfill{}
\begin{enumerate}
\item In each iteration $g$, $g=1,\dots n_{0}+G$, given the current value
$(z^{(n+2)})^{g}$, the proposal $(z^{n+2})'$ is drawn according
to 
\[
(z^{(n+2)})'=(z^{(n+2)})^{g}+X,\text{ for }X\sim\text{N}_{n+2}(0,\Sigma),
\]
where the $(n+2)\times(n+2)$-covariance matrix $\Sigma$ is chosen
to reach some desired acceptance rate. 
\item Set 
\[
(z^{(n+2)})^{(g+1)}=\left\{ \begin{matrix}(z^{(n+2)})' & \text{ with prob.} & \alpha((z^{(n+2)})^{g},(z^{(n+2)})')\\
z^{(n+2)} & \text{ with prob.} & 1-\alpha((z^{(n+2)})^{g},(z^{(n+2)})')
\end{matrix}\right.,
\]
where the acceptance-probability $\alpha(z^{(n+2)},(z^{(n+2)})')$
is given by 
\begin{align*}
\min\left\{ 1,\frac{\hat{f}(z_{n+1}',z_{c},z_{n+2}',T-t)\tilde{f}(t-t_{n},z_{n+1}',z_{n}')b_{n}((z^{(n)})'|y^{(n)})(1-\pi(T-t,z_{n+1}-z_{c}))}{\hat{f}(z_{n+1},z_{c},z_{n+2},T-t)\tilde{f}(t-t_{n},z_{n+1},z_{n})b_{n}(z^{(n)}|y^{(n)})(1-\pi(T-t,z'_{n+1}-z_{c}))}\right\} .
\end{align*}
\item Discard the draws from the first $n_{0}$ iterations and save the
sample 
 $(z^{(n+2)})^{n_{0}+1},\dots,(z^{(n+2)})^{n_{0}+G}$. 
\end{enumerate}
\end{algorithm}

The last type of expression that occurs in the valuation of CoCos
in the previous section, are the $i$-step survival probabilities
in the valuation of a CoCo with an accounting report trigger, as in
Theorem~\ref{thmCoCoprice3} and Theorem~\ref{thmCoCoprice4}. This
expressions are in Equation~(\ref{survprobonlyrep}) and Equation~(\ref{survprobonlyrep2})
given in the form
\begin{align*}
\int_{\mathbb{R}^{n}}\mathbb{P}(\xi(z_{n})\in\Xi_{i})p_{Z}(z^{(n)}|y^{(n)})\dd z^{(n)},
\end{align*}
for a set $\Xi_{i}$ which equals $(y_{c},\infty)^{i}$ or $(y_{c},\infty)^{i-1}\times(y_{cc},\infty)$
and a multivariate normally distributed random variable $\xi(z_{n})$.
For a sample $((z^{(n)})^{1},\dots,(z^{(n)})^{G})$ from $p_{Z}(z^{(n)}|y^{(n)})$,
this type of integral can be approximated by 
\begin{align}
\frac{1}{G}\sum_{g=1}^{G}\mathbb{P}(\xi(z_{n}^{g})\in\Xi_{i}).\label{survapprox}
\end{align}
The necessary sample is again obtained using a MCMC-algorithm, as
follows. 
\begin{algorithm}
\label{rwmh4}\hfill{}
\begin{enumerate}
\item In each iteration $g$, $g=1,\dots,n_{0}+G$, given the current value
$(z^{(n)})^{g}$, the proposal $(z^{n})'$ is drawn according to 
\[
(z^{(n)})'=(z^{(n)})^{g}+X,\text{ for }X\sim\text{N}_{n}(0,\Sigma),
\]
where the $n\times n$-covariance matrix $\Sigma$ is chosen to reach
some desired acceptance rate. 
\item Set 
\[
(z^{(n)})^{(g+1)}=\left\{ \begin{matrix}(z^{(n)})' & \text{ with prob.} & \alpha((z^{(n)})^{g},(z^{(n)})')\\
z^{(n)} & \text{ with prob.} & 1-\alpha((z^{(n)})^{g},(z^{(n)})')
\end{matrix}\right.,
\]
where the acceptance-probability $\alpha(z^{(n)},(z^{(n)})')$ is
given by 
\begin{align*}
\alpha(z^{(n)},(z^{(n)})') & =\min\left\{ 1,\frac{p_{Z}((z^{(n)})'|y^{(n)})}{p_{Z}(z^{(n)}|y^{(n)})}\right\} \\
 & =\min\left\{ 1,\frac{\prod_{i=1}^{n}p_{Z}(z_{i}'|z'_{i-1})p_{U}(y_{i}-z'_{i}|y_{i-1}-z'_{i-1})}{\prod_{i=1}^{n}p_{Z}(z_{i}|z_{i-1})p_{U}(y_{i}-z_{i}|y_{i-1}-z_{i-1})}\right\} .
\end{align*}
\item Discard the draws from the first $n_{0}$ iterations (because the
Markov chain needs a burn-in period to converge to the target distribution)
and save the sample 
 $(z^{(n)})^{n_{0}+1},\dots,(z^{(n)})^{n_{0}+G}$. 
\end{enumerate}
\end{algorithm}

\section{Applying the model}

\label{sec:application} \setcounter{equation}{0}

In this section we use the model to shed light on a variety of questions
related to the basic valuation model itself and its sensitivity to
design and ``environmental'' variables such as volatility shocks.
We then explore the interaction between CoCos and other elements of
the capital structure, and in particular look at risk taking and investment
incentives when CoCos are used instead of other types of funding,
like straight debt or equity. Finally we use the fact that we incorporate
the MDA trigger and the coupon payment contingency by comparing the
Deutsche Bank profit scare and its impact on CoCo prices with model
predictions we obtain with our valuation model.

\begin{table}[h]
\centering{}%
\begin{tabular}{|c|c|}
\hline 
Parameter  & Value\tabularnewline
\hline 
\hline 
Initial asset value $V_{0}$  & 100\tabularnewline
\hline 
$n$, the number of accounting reports until time $t$  & 2\tabularnewline
\hline 
Conversion trigger $v_{c}$  & 80\tabularnewline
\hline 
Default trigger $v_{b}$  & 65\tabularnewline
\hline 
Recovery rate at default $\alpha$  & 0.5\tabularnewline
\hline 
Total principal straight debt $P_{1}$  & 50\tabularnewline
\hline 
Coupon straight debt  & 0.04\tabularnewline
\hline 
Total principal CoCos $P_{2}$  & 5\tabularnewline
\hline 
Coupon CoCos $c_{2}$  & 0.07\tabularnewline
\hline 
Maturity CoCos $T$  & t+5\tabularnewline
\hline 
Drift asset process $m$  & 0.01\tabularnewline
\hline 
Volatility asset process $\sigma$  & 0.1\tabularnewline
\hline 
Mean accounting noise $\mu_{\epsilon}$  & 0\tabularnewline
\hline 
Volatility accounting noise $\sigma_{\epsilon}$  & 0.1\tabularnewline
\hline 
Risk free rate  & 0.03\tabularnewline
\hline 
\end{tabular}\label{BaseCaseParameters}\caption{Base Case Parameters}
\end{table}

\subsection{Parametrization of the base case}

Table~\ref{BaseCaseParameters} lists the values of the base case
parameters. For the choice of the base case parameters, some restraints
should be taken into account. For example, the conversion trigger
should be higher than the default trigger. Also, a CoCo should pay
a higher coupon than straight debt, to compensate for the higher risk.
Furthermore, we have no empirical evidence for a reasonable level
of accounting noise, so we set the volatility of accounting noise equal
to the base case parameter chosen by \citet{Duffie2001}, where the
accounting noise variance is chosen to match short run default probabilities
implicit in short run CDS spreads.

In the base case, we will assume the CoCo has a regulatory trigger,
i.e.\ the regulator has access to the true state of the bank and
conversion can take place at any time, not just at accounting dates,
but the market has to evaluate conversion probabilities given this
trigger rule using accounting information only (cf.~Section~\ref{secconversion}
for the mathematics of this trigger). We will also explore other trigger
mechanisms.\\
 Furthermore, we define the dilution ratio $\rho$ (cf.~Section~\ref{secconversion})
as the fraction shares owned by the CoCo holder post-conversion: 
\begin{eqnarray}
\rho & = & \frac{\Delta P_{2}}{\Delta P_{2}+1}
\end{eqnarray}
where $P_{2}$ is the face value of the CoCo before conversion, and
$\Delta$ equals the number of shares the CoCo holder receives at
conversion. The number of old shares is normalized to 1. A dilution
ratio of $\rho=0$ means that the CoCo suffers a principal write-down
(PWD) at conversion, while $\rho=1$ corresponds to the extreme case
that the original shareholders are completely wiped out at conversion.\\
 To compute prices for PWD CoCos, we make use of Theorem~\ref{thmCoCoprice1}.
The integral involved is approximated as in Equation~(\ref{capprox2}),
for which the necessary sample is obtained by using Algorithm~\ref{rwmh}.
To compute prices for CoCos with a conversion into shares, we make
use of Theorem~\ref{thmCoCoprice2}, where the first term in the pricing
formula follows again by using Algorithm~\ref{rwmh} and the second
term is approximated as in Equation~(\ref{approx2ndterm}), for which
the necessary sample is obtained by execution of Algorithm~\ref{rwmh3}.
Then the figures are produced by repeatedly following this procedures
for different values of the parameters.

\begin{figure}[h]
\begin{center}
\includegraphics[width=0.75\textwidth,height=0.35\textheight]{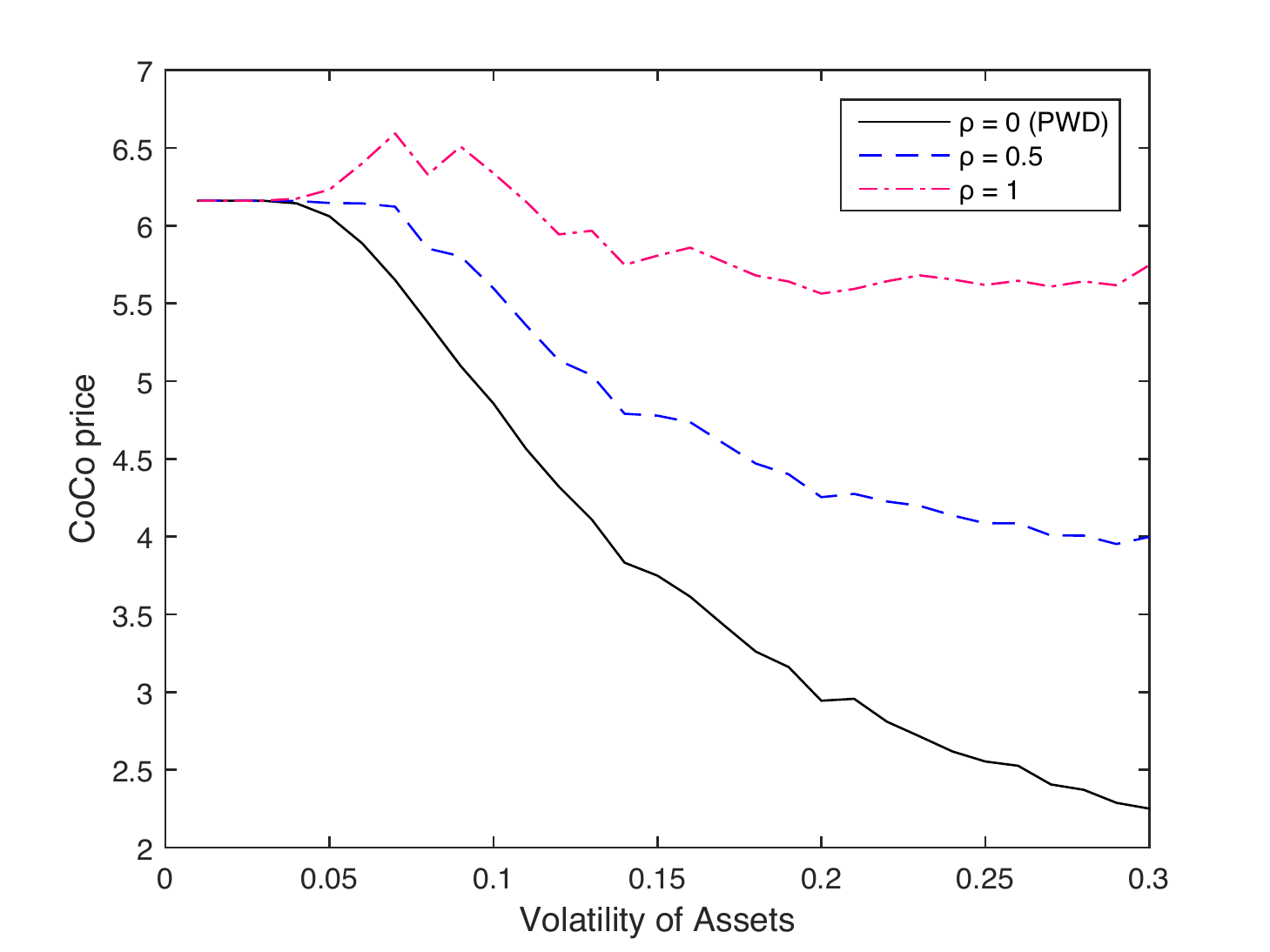} 
\caption{CoCo Prices and asset volatility for different CoCo design parameters}
\label{Fig asset volatility and CoCo pricces} 
\end{center}
\end{figure}

\subsection{Asset Volatility, Accounting Noise and CoCo design parameters}

\begin{figure}[h]
\begin{center}
\includegraphics[width=0.75\textwidth,height=0.35\textheight]{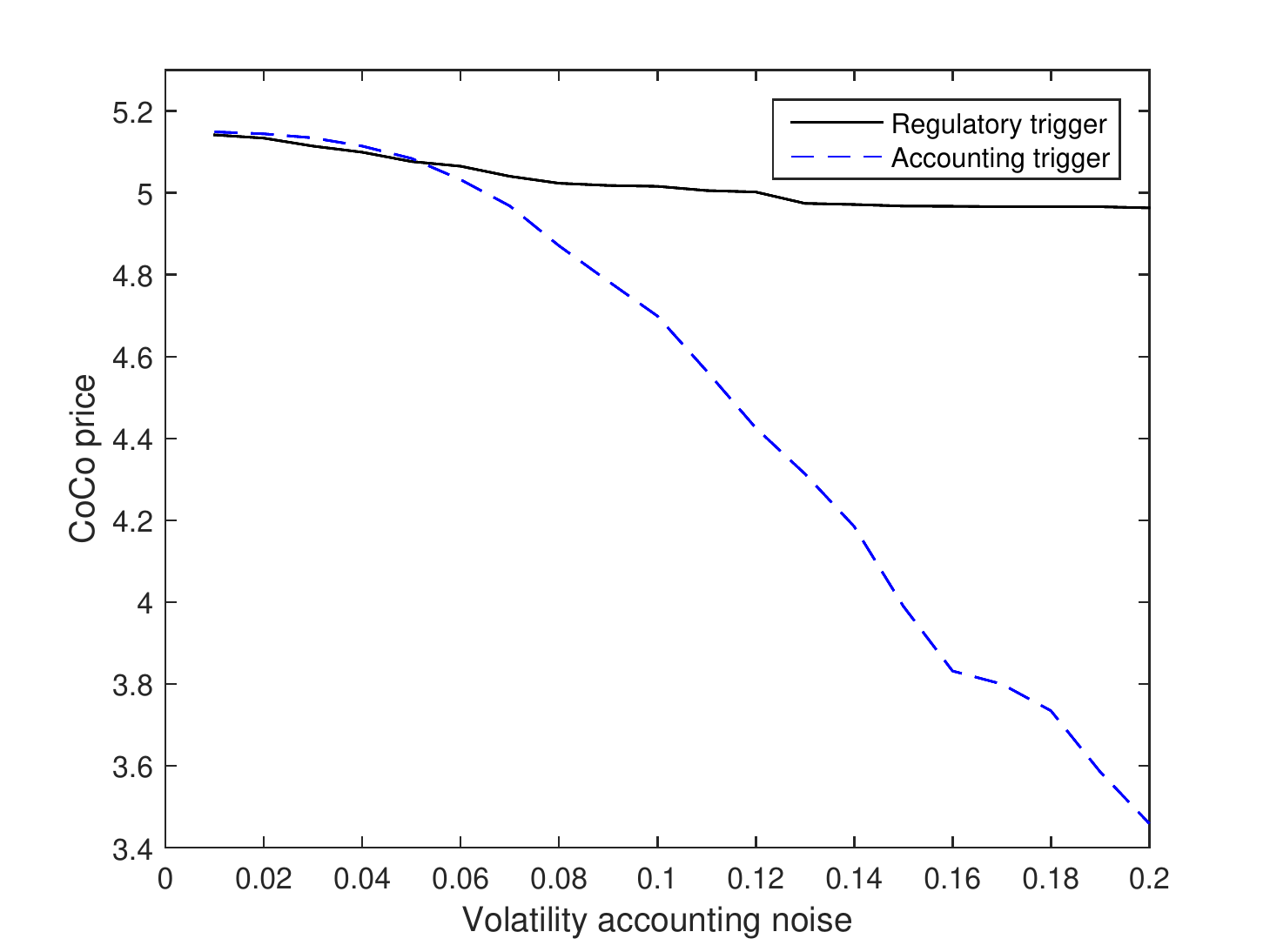}\caption{Accounting noise and trigger design}
\label{Accounting_Noise_Trigger_Design} 
\end{center}
\end{figure}

In this subsection we study the impact of changes in asset volatility
and accounting noise on CoCo prices as a function of different design
parameters.

\subsubsection{Asset Volatility shocks}

We first look at the price impact of changes in volatility of the
underlying asset value process for different CoCo designs.
In Figure~\ref{Fig asset volatility and CoCo pricces} several CoCo
prices are plotted against the volatility of assets $\sigma$, see
Equation~\eqref{eq:model}. The solid line corresponds to a PWD CoCo.
Clearly, the price of a PWD CoCo decreases when assets become more
volatile. This is of course as one would expect, as a higher $\sigma$
increases the probability of the principal write-down happening, causing
the CoCo price to decrease. The dashed line, corresponding to $\rho=0.5$,
shows already that this negative effect from volatility on the CoCo
price is weaker when terms of conversion are more favorable to the
CoCo investor in that her loss is lower, at least some shares are
received after conversion, although not yet enough to compensate for
the loss of principal. In the extreme case that shareholders are completely
wiped out at conversion, corresponding to the dashed-dotted line,
this negative effect is even partially reversed. In this case, the
price first increases with volatility as the (now favorable) conversion
becomes more likely. However, for higher volatility levels the increasing
probability of default and associated costs of bankruptcy push the
price down again.

\subsubsection{Accounting Noise shocks}

We next consider the relationship between accounting noise $\sigma_{\epsilon}$
and the price of a CoCo. In Figure~\ref{Accounting_Noise_Trigger_Design},
CoCos with a regulatory trigger and CoCos with an accounting based
trigger are considered. The book value CoCo is priced by the formula
given in Equation~(\ref{CoCopriceacc2}); This value is computed
using the approximation in Equation~(\ref{survapprox}), for which
the necessary samples are obtained by using~\ref{rwmh4}.

Figure~\ref{Accounting_Noise_Trigger_Design} shows the importance of taking into account the trigger design for the pricing of the CoCo. The
increase in accounting volatility has almost no impact on the value
of the CoCo with a regulatory trigger (the solid line in Figure~\ref{Accounting_Noise_Trigger_Design});
but the dashed line shows that the value of CoCos with a trigger that
depends on accounting reports, the CoCo price is seriously (and obviously
negatively) affected by accounting noise. This is in line with the
results of ~\citet{Duffie2001}: they find that the default probability
increases when the reports become more noisy. In our CoCo setting,
this means that the probability of conversion increases when $\sigma_{\epsilon}$
increases, causing the CoCo price to go down. Figure~\ref{Accounting_Noise_Trigger_Design}
shows that the price of a CoCo with the trigger depending on accounting
reports (slotted line) is much more sensitive to accounting noise
than the price of a CoCo with a regulatory (PONV) trigger (solid line).

\subsubsection{Accounting news and correlation in the accounting report error}
Consider next the impact of the correlation coefficient  $\kappa$ in the accounting
noise error term. In Figure~\ref{fig:kappa plot} we show the price response of a PWD CoCo to a bad news accounting report. The set up
is as follows. After the first report ($Y_{1}=\log 100)$, a second report
is issued: $Y_{2}= \log 85$. The conversion trigger is set at $\log 80$, with a PONV trigger type. The plots show a clear and immediate price response
to the arrival of the bad news.
\begin{figure}[H]
\begin{center}
\includegraphics[width=0.75\textwidth,height=0.35\textheight]{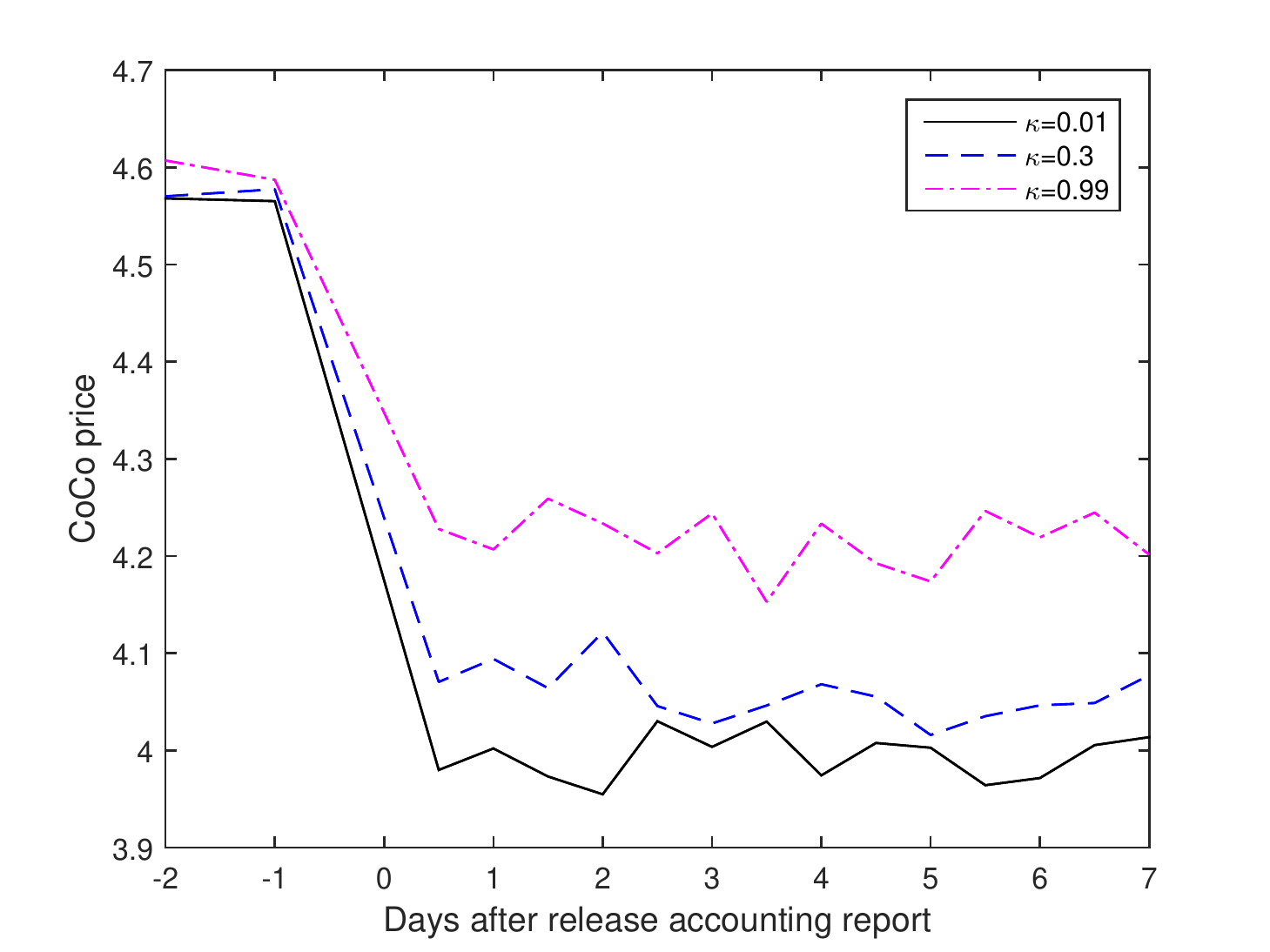}\caption{Price response to ``bad accounting news'' for
different values of the autocorrelation parameter $\kappa$.}
\label{fig:kappa plot}
\end{center}
\end{figure}
Interestingly, a clear pattern emerges if the exercise is repeated
for different values of the autocorrelation parameter $\kappa$:
although the pattern is similar over the entire range from almost
no correlation in accounting noise ($\kappa=0.01)$ to almost complete
persistence of accounting noise innovations ($\kappa=0.99)$, for
higher values of the correlation parameter the price response is more
muted.  Since
the accounting report is known to be contaminated by accounting noise
each time a new report is issued, a higher value of $\kappa$ means
that more of the past noise arrivals survive in the current one, while
at the same time the variance of the accounting noise term $U_i$ increases
with $\kappa$, as it is, in a stationary regime, proportional to $1/(1-\kappa^{2})$. This in turn
lowers the information value of accounting news and explains why a
bad (i.e. worse than the previous one) report leads to a smaller negative
price response for higher $\kappa:$ the signal is less informative
so triggers a smaller price response.

\subsubsection{Time lapsed since last accounting report}

In Figure~\ref{fig:Time-since-last-report} we report on a different experiment:
we show how different CoCo designs are influenced by time lapsed since
the last accounting report. The plot shows the value of three differently
structured CoCo's, each with a different degree of shareholder dilution
after conversion as a function of time lapsed since the last accounting
report. The black line represents a PWD CoCo where the CoCo is written
off upon conversion and no subsequent dilution of the old shareholder
takes place; the other two lines represent equity converters, one
with partial dilution of the old shareholder ($\rho=0.5)$, the dashed-dotted
line, and one where the old shareholder is completely wiped out after
conversion $(\rho=1)$, the dashed line.
\begin{figure}[h]
\begin{center}
\includegraphics[width=0.75\textwidth,height=0.35\textheight]{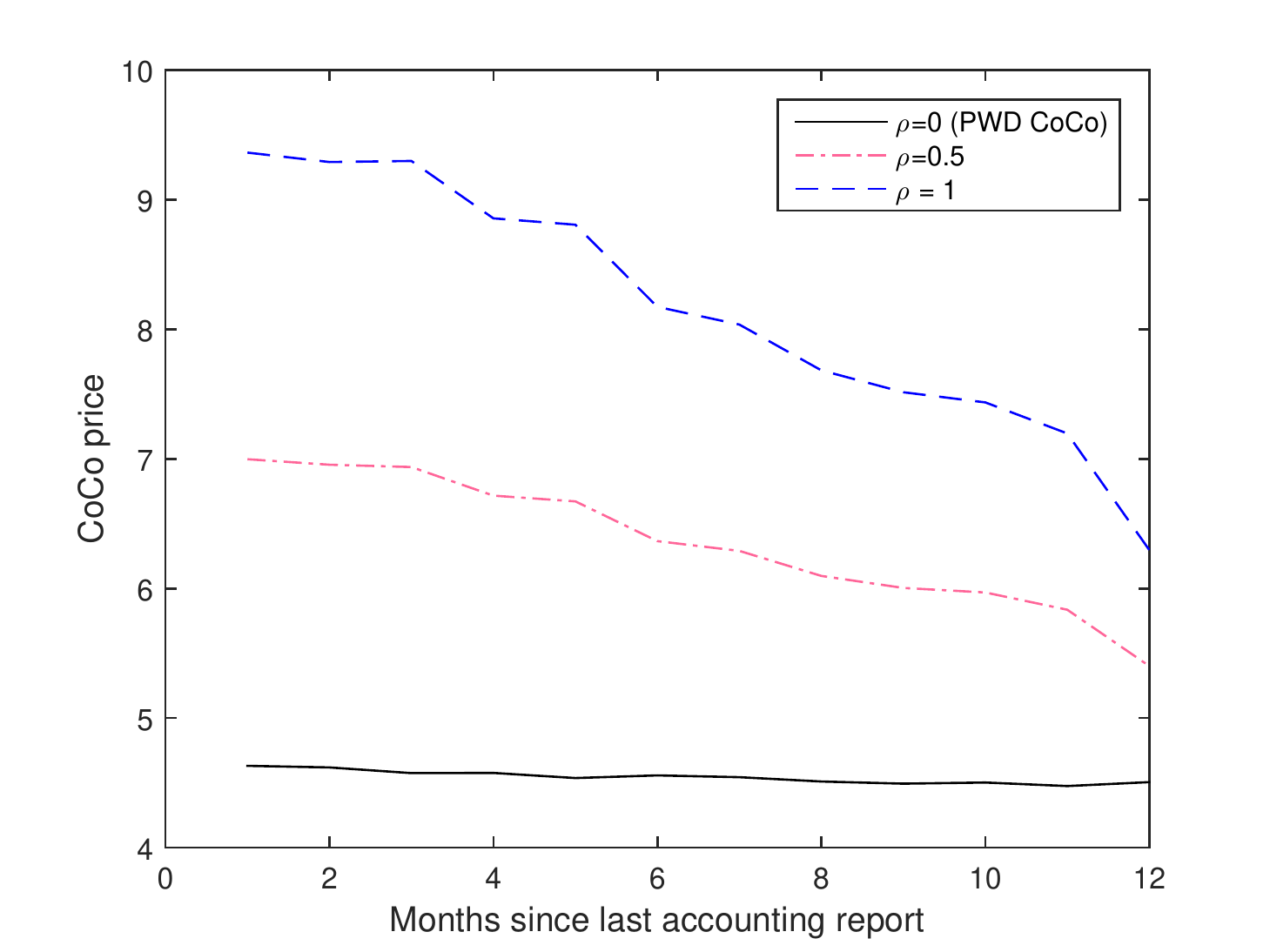}\caption{Time since last accounting report.}
\label{fig:Time-since-last-report}
\end{center}
\end{figure}
The plots show very little impact on the PWD CoCo while the two equity converters
decline in value as the time since the last accounting report increases.
A longer time lapse does not change the asset price dynamics but leads
to a higher uncertainty as to where the asset value is at the time of
valuation. This is similar to moving more weights in the tails and
thus a larger probability of bankruptcy. Since bankruptcy follows
conversion, a higher probability of bankruptcy does not influence
the PWD, they will then already have lost everything because conversion
precedes bankruptcy. But the more shares the CoCo holder receives
upon conversion, the more she loses from a subsequent bankruptcy,
so the price decline increases more for higher values of the dilution
parameter $\rho$.

\subsection{Design parameters and CoCo Valuation}

\label{secparamimpact} Consider next the impact on pricing of the
main characteristics of the CoCo design: the trigger level and the
number of shares received upon conversion.

\subsubsection{The conversion trigger}

In Figure~\ref{Conversion Trigger}, the CoCo price is plotted against
the conversion trigger for different degrees of dilution. The solid
line corresponds to a PWD CoCo, the other lines to CoCos with varying
degrees of dilution of the original shareholders upon conversion as
specified in the legend.

\begin{figure}[h]
\begin{center}
\includegraphics[width=0.75\textwidth,height=0.35\textheight]{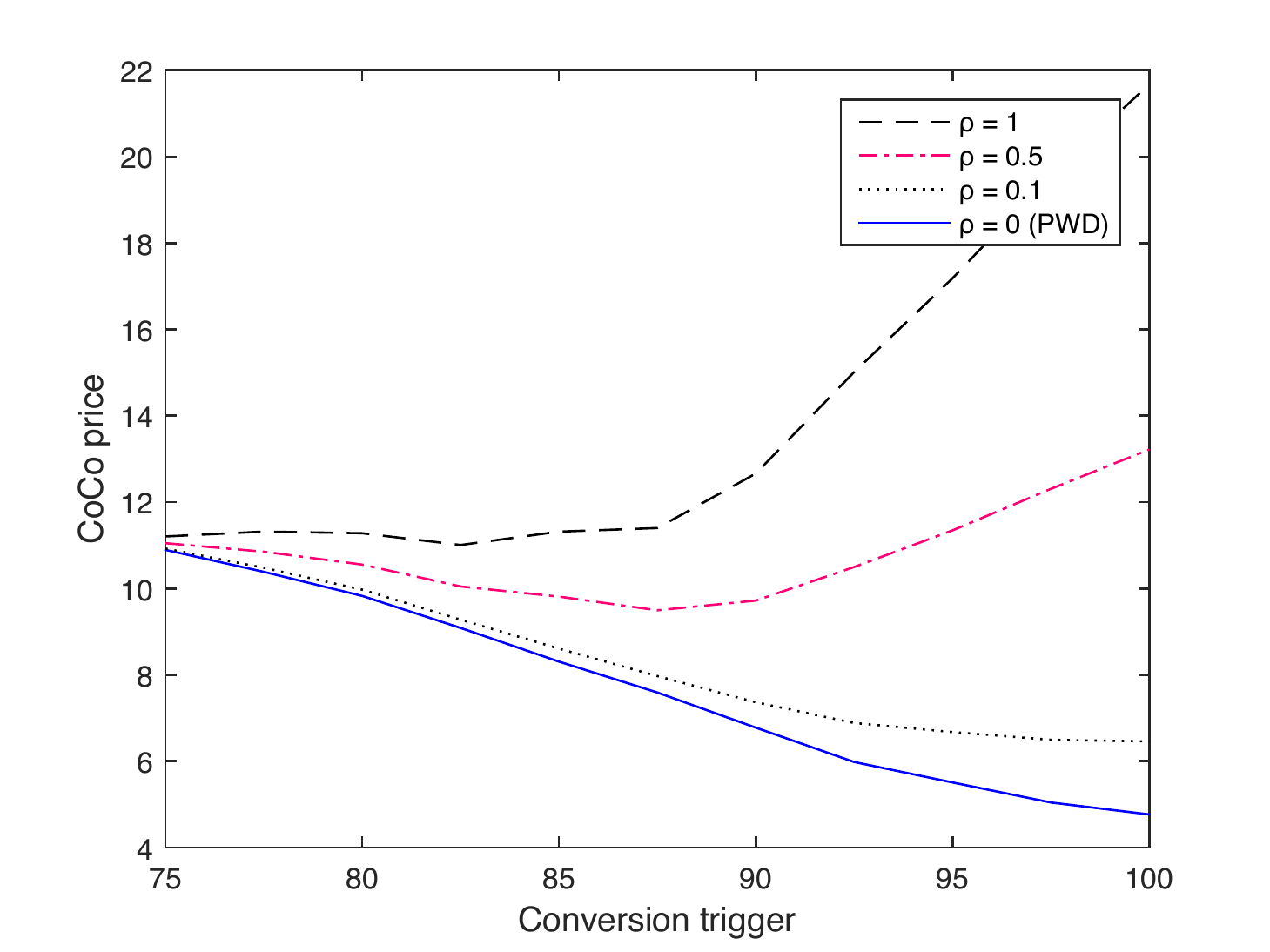} \caption{Conversion Trigger}
\label{Conversion Trigger} 
\end{center}
\end{figure}
As one would expect, the price of a PWD CoCo (the solid line) is lower
for a higher conversion trigger\footnote{Note that the trigger is defined as a percentage of the asset value
with the losses coming from the top (i.e.\ equity above debt on the
liability side), so a higher trigger value means a higher probability
of conversion, as is done in the rest of the academic literature.
In the banking and supervision literature, it is more conventional
to define the trigger value also as a percentage of (risk weighted)
assets, but with the losses coming from the bottom, with equity below
debt; in that definition a higher trigger ratio leads to a lower probability
of conversion.}: a higher conversion trigger increases the probability of a principal
write-down and its associated loss of principal. However, the other
lines show that if conversion terms are more favorable to the CoCo
investor, the impact of the trigger level changes: price will increase
with the conversion trigger if the dilution ratio favors the CoCo
holder enough. In the extreme case that the dilution ratio $\rho$
equals one (the dashed line in Figure~\ref{Conversion Trigger}),
the CoCo price goes up with the conversion trigger. For less extreme
dilution parameters, this positive effect is weaker, and the corresponding
lines are in between the two extremes (no dilution versus complete
dilution).

\subsubsection{On dilution and Leverage}

In Figure~\ref{Dilution and leverage}, the price of a CoCo is plotted
against $\Delta$, the number of shares received at conversion per
unit of principal, for different values of straight debt in the firm's
capital structure. The case $\Delta=0$ corresponds to a principal
write-down CoCo, while $\Delta=\infty$ corresponds to the case in
which all of the original shareholders are wiped out at conversion
and the CoCo investors are then the only shareholders left. Figure
\ref{Dilution and leverage} clearly shows that the CoCo price increases
with $\Delta$. This is of course as expected, as a higher $\Delta$
means a higher payout at conversion. Furthermore, the figure shows
that a CoCo with a conversion into shares has a higher price when
there is a lower amount of straight debt issued. Hence the CoCo is
more valuable when the firm has a lower leverage. This can also easily
be explained, as the CoCo investors receive a fraction of the firm's
equity value at conversion and the equity value is higher in case
there are less liabilities.

\begin{figure}[h]
\begin{center}
\includegraphics[width=0.75\textwidth,height=0.35\textheight]{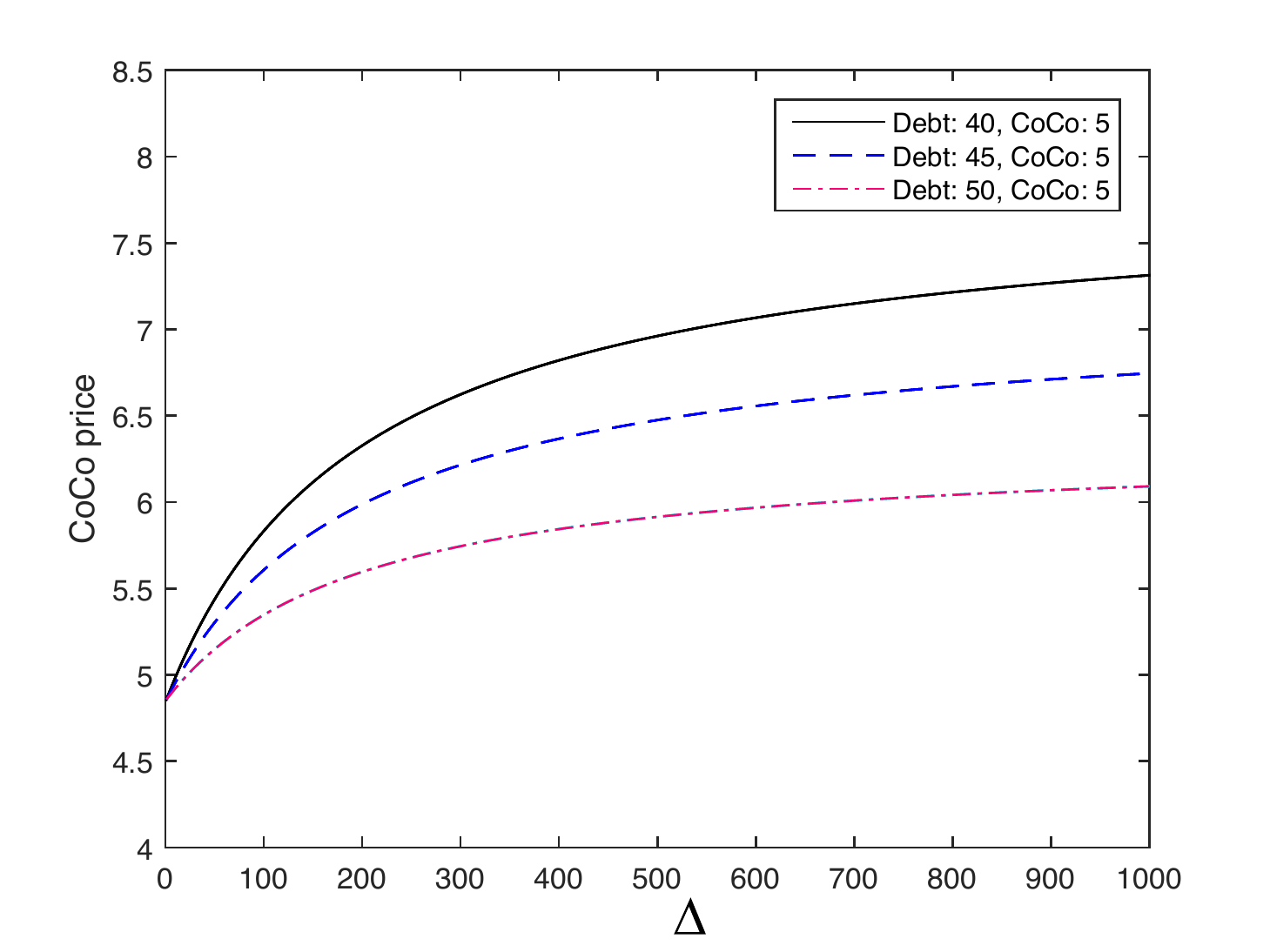}\caption{CoCo prices against dilution, for different leverage ratios.}
\label{Dilution and leverage} 
\end{center}
\end{figure}

The lines for different leverage converge to the same point on the
vertical axis as $\Delta\rightarrow0;$ for a PWD CoCo, leverage has
no impact on the price since both the CoCo and equity are junior to
debt. This result does depend on the assumption that the variance
of the asset value process is exogenously chosen; if it would be endogenously
chosen, higher leverage would lead to more risk taking and a higher
variance, which would have an impact on the value of the CoCo even
if it has a PWD structure (this point is made in \citet{Chan2016;revisedNovember2017}).

\subsection{Capital structure, CoCos and Risk Taking Incentives}

\subsubsection{Issuing CoCos to replace straight debt}

Consider next the impact of the issuance of CoCos on the capital structure
of the bank and, deduced from that, on incentives for shareholders\footnote{The computation of the prices and the production of the figures is
performed following the same procedures as in Section~\ref{secparamimpact}.}. First consider the case in which straight debt is replaced with
CoCos. In Figure~\ref{CoCoForDebt} we show the \emph{change} in
equity value (on the vertical axis) as a consequence of replacing
5 units of straight debt with 5 units of CoCos, set off against different
trigger prices. The different lines correspond to different degrees
of the dilution parameter $\rho$, again ranging from 0 to 1 (from
no dilution at all to infinite dilution). The solid line and the dotted
line indicate that shareholders only benefit from replacing debt with
CoCos when the terms of conversion are favorable enough to the shareholders
and the trigger is high enough. For low trigger ratios, the conversion
possibility becomes very small and the exercise comes down to swapping
debt for debt. That actually turns out to have a negative impact on
equity values because CoCos then are just a more expensive form of
debt so replacing debt with CoCos then actually destroys equity value.
As the trigger ratio goes up (move to the right in Figure~\ref{CoCoForDebt}),
the probability of getting the benefit of wiping out the CoCo debt
at favorable terms becomes more likely and starts to dominate, hence
the positive sign for high trigger ratios. Of course that second effect
does not take place for highly dilutive CoCos, for low probability
of conversion the impact of the debt for CoCo swap is negative, as
with non-dilutive CoCo's. But as the conversion trigger rises
and with it the conversion probability, the negative impact of a highly
dilutive conversion comes closer, so the price impact turns even more
negative. So the dashed line and the dashed-dotted line (highly dilutive
cases) show that shareholders have no incentive to swap debt for highly
dilutive CoCos, and increasingly less so as the probability of conversion
increases with higher trigger levels (academic convention).

\begin{figure}[h]
\begin{center}
\includegraphics[width=0.75\textwidth,height=0.35\textheight]{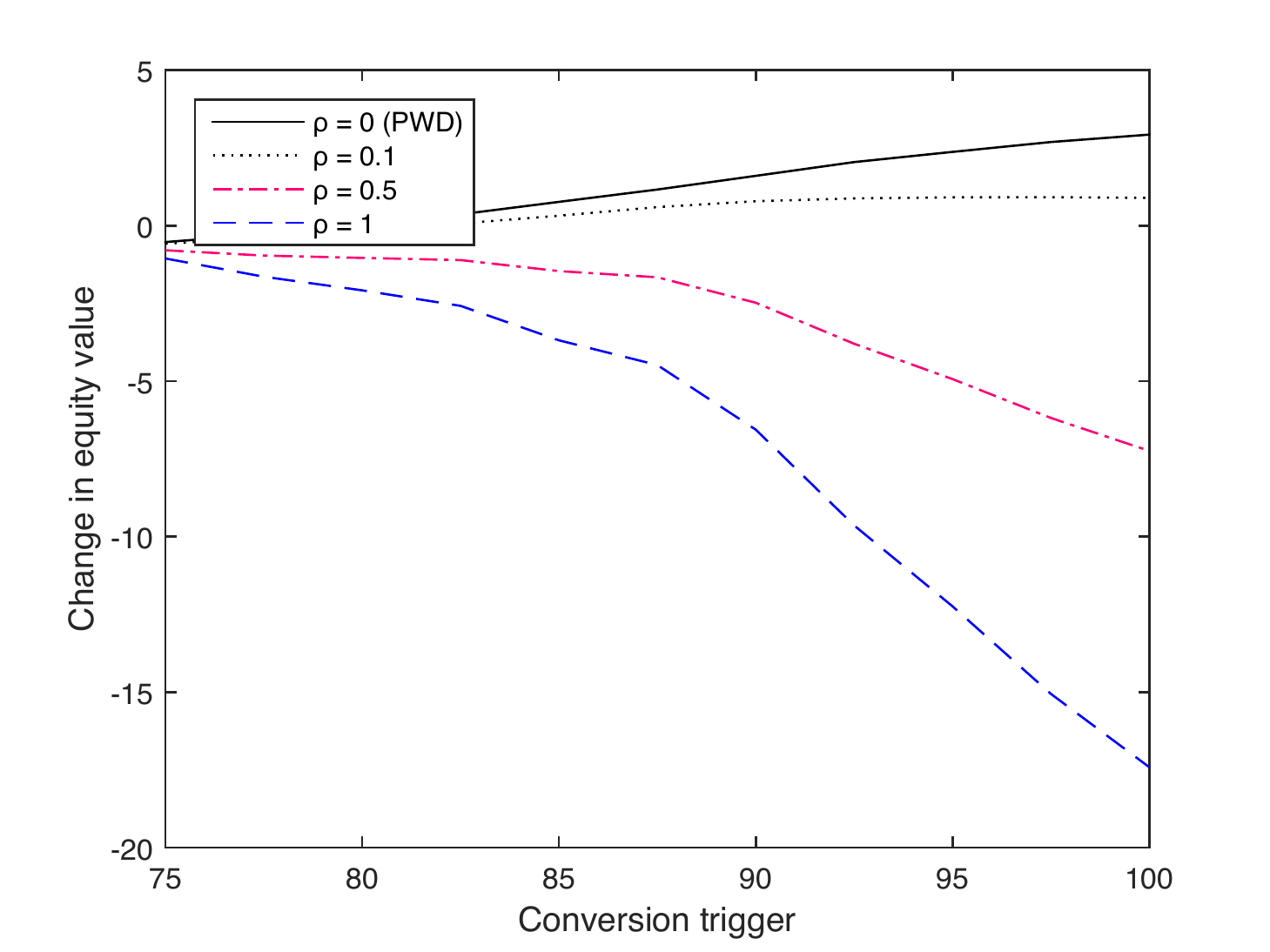}\caption{Change in Equity Value when 5 units of debt are replaced with 5 units
of CoCo (in market value terms). }
\label{CoCoForDebt} 
\end{center}
\end{figure}

\subsubsection{Issuing CoCos to replace equity}

Consider next a change in capital structure in the other direction,
where equity instead of debt is replaced by CoCos. Specifically, we
assume a CoCo is issued and the proceeds are used to buy back equity
at market value. The consequences on equity values are shown in Figure~\ref{CoCosForEquity},
again for different trigger levels (on the horizontal axis with the
different lines representing different degrees of dilution after conversion).

\begin{figure}[h]
\begin{center}
\includegraphics[width=0.75\textwidth,height=0.35\textheight]{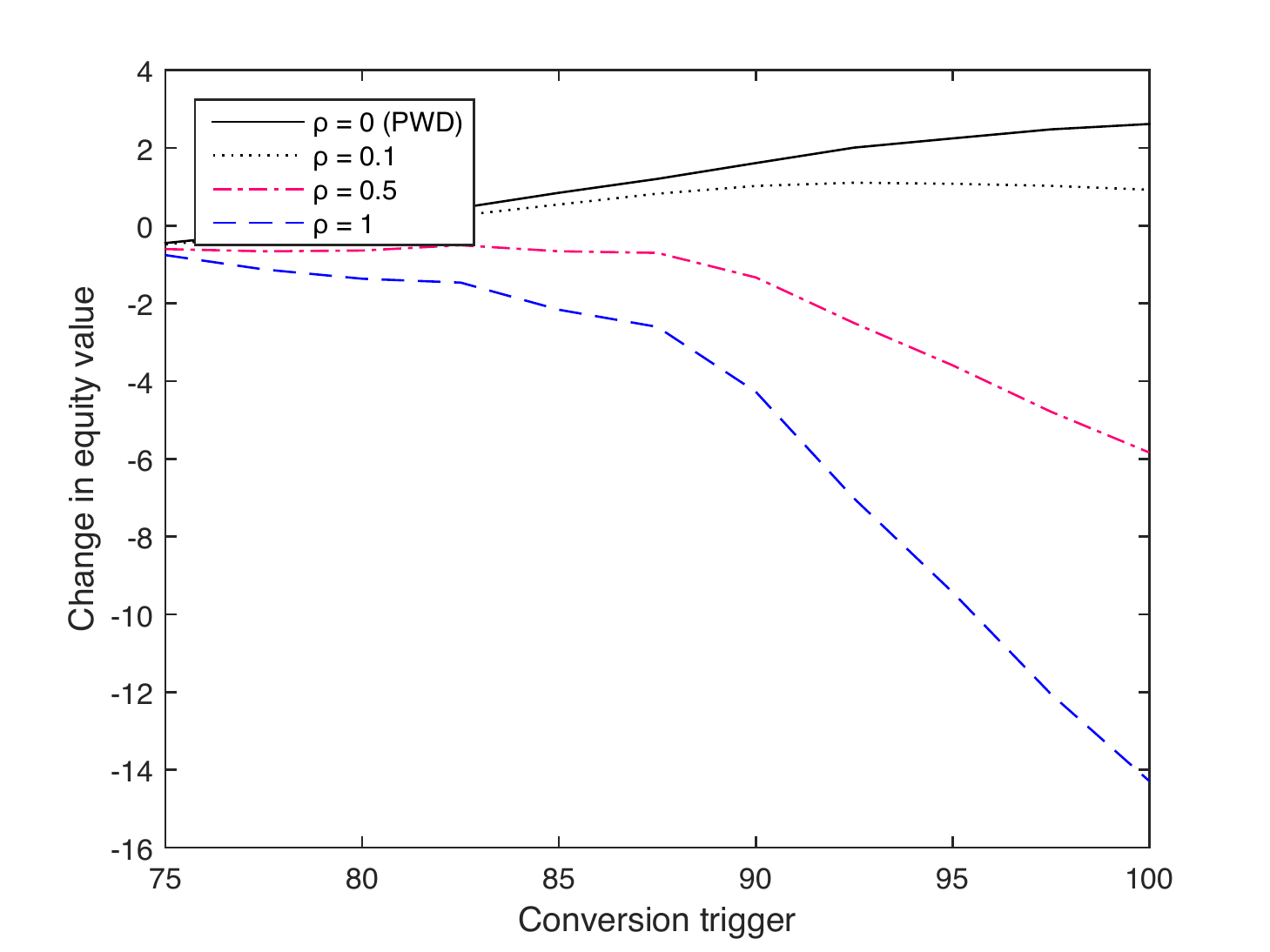}\caption{Change in equity value when 5 units of equity are replaced with 5
units of CoCos (in market value)}
\label{CoCosForEquity} 
\end{center}
\end{figure}
The pattern is very similar to the debt for CoCo swaps analyzed in
Figure~\ref{CoCoForDebt}. Equity holders have a strong incentive
to issue PWD CoCos with $\rho=0$ instead of new equity (or even issue
PWD CoCos to buy back debt as is done in this policy experiment),
since they actually gain on conversion. However for lower trigger
ratios the probability of conversion becomes too small, turning CoCos
\emph{de facto} into expensive debt, so the impact on equity value
turns negative for low values of the trigger ratio. And equity holders
will never want to issue dilutive CoCos ($\rho=1$ is the extreme
case with infinite dilution) for any level of the trigger ratio: before
conversion these CoCos are an expensive form of debt and after conversion
or rather at conversion time equity holders will actually loose out
when conversion takes place, making the instrument unambiguously unattractive
to shareholders when structured this way. These results may well explain
why some 60\% of all CoCos are PWD CoCos, cf.~\citet{Avdiev2017}, instead of
the dilutive CoCos favored by the academic literature (\citet{Calomiris_Herring_2013}
is an early and eloquent example of what is a widely shared view in
the academic literature arguing CoCos should be highly dilutive).

\subsection{Debt Overhang: on CoCos and Investment Incentives}

Debt overhang arises when the firm's loss absorption capacity has
become too low to protect the debtholders from fluctuations in asset
values (cf.\ \citet{Merton_1974}, \citet{Myers1977}), possibly
to the point of arrears having emerged already. One consequence of
debt overhang is that investment incentives are reduced for equity
holders, since part of the benefits of a new project will in effect
have to be shared with the creditors. Even if there are no actual
arrears yet, but debt is trading under par, part of the asset value
increase will go into increased market value of the debt, at the (partial)
expense of a higher market value for equity. In a structural model
without CoCos, the shareholders then do not have an incentive to invest
exactly at the moment the firm most needs an increase in asset values,
i.e.\ when the firm is near bankruptcy. Almost all of the value of
the investment will then be captured by the debt holders, as the value
of debt increases when the probability of a bankruptcy is reduced.
In which way CoCos interact with a situation with debt overhang is an
interesting question; CoCos introduce additional loss absorption capacity
which is good for debt holders, but CoCos may also have their own
impact on equity values. Debt holders may also profit in another way
in that, depending on the design of the CoCos, shareholders may have
an increased incentive to make an investment to avoid conversion.

The debt overhang and incentive issue can be looked at within the
context of our model by looking at what happens when assets are increased
by one unit, financed through one unit of equity (issued at market
value). If the total market value of equity goes up by more than one
unit, the shareholders would make a profit when they invest, giving
them an incentive to do so. However, when equity increases by less
than one unit, the investment is not beneficial to shareholders to
offset the expense, all or part of the benefits are apparently captured
by debt holders. We therefore consider the case in which a new accounting
report has just be released, with an asset value, see Equation~\eqref{eq:acc},
of $Y_{t_{n}}=100$; we can then examine what happens when this asset
value increases by one unit. The profit of this investment of one
unit is plotted against the conversion trigger in Figure~\ref{InvestmentIncentives}.

\begin{figure}[h]
\begin{center}
\includegraphics[width=0.75\textwidth,height=0.35\textheight]{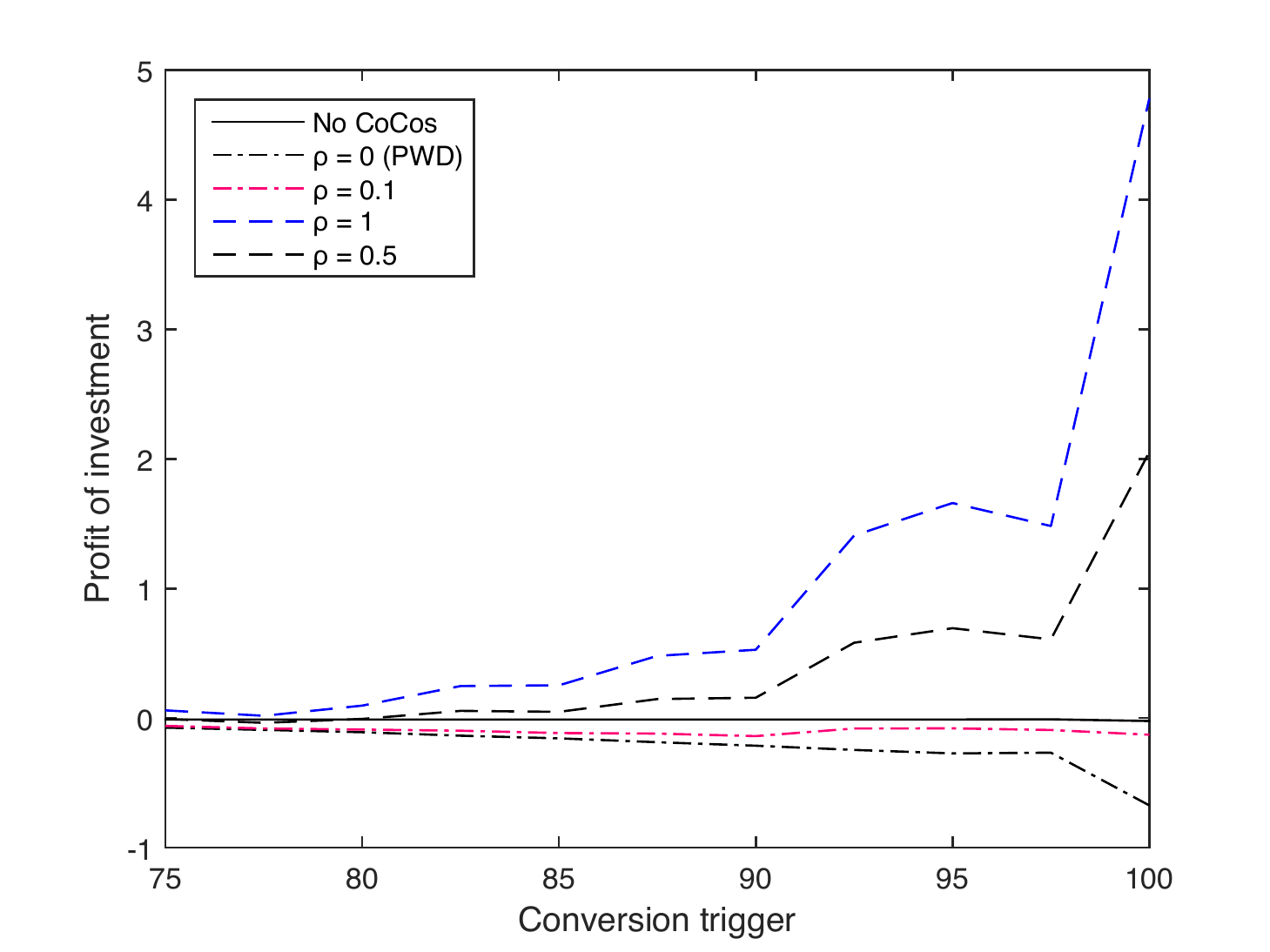}\caption{Debt overhang, CoCos and investment incentives}
\label{InvestmentIncentives} 
\end{center}
\end{figure}
The solid black line is our benchmark case with only straight debt
in addition to equity. The simulation shows the impact of debt overhang:
without CoCos (the solid line) the shareholders do not make a profit
when they invest, they actually suffer a small loss. The dashed-dotted
lines show that when the terms of conversion are favorable to shareholders
(i.e.\ CoCo holders loose out upon conversion), the shareholders
have even less of an incentive to engage in additional investment,
actually worsening the debt overhang problem. The black dashed-dotted
line corresponds to the existence of a PWD CoCo in the capital structure
of the firm and shows that the PWD CoCo indeed makes the investment
incentive for shareholders more negative, especially close to the
conversion trigger. So the strongest increase in Debt Overhang is
with the CoCos that most favors shareholders, the CoCos with a principal
write-down. The same happens to a somewhat lesser degree with CoCos
at slightly less non-dilutive terms but still favorable to shareholders.
Thus PWD or insufficiently dilutive CoCos are not capable of solving
the problem of debt overhang. However, highly dilutive CoCos do strengthen
shareholders' incentives to invest because they want to avoid conversion.
See in particular the dashed lines in Figure~\ref{InvestmentIncentives},
which correspond to highly dilutive CoCos; clearly such CoCos improve
the shareholders' investment incentives because they wish to avoid
conversion. Especially close to the conversion trigger, the shareholders
have in this case a substantive incentive to invest in a last attempt
to avoid the unfavorable conversion. To summarize, when terms of conversion
are beneficial enough to CoCo investors instead of favoring the old
shareholders, CoCos are capable of creating more of an investment
incentive for the shareholders. However, PWD CoCos and in general
less dilutive CoCos actually lead to lower investment incentives and
worsen the debt overhang problem when compared to straight debt.

\subsection{Coupon payments, the MDA trigger and the Deutsche Bank CoCo scare
of February 2016}

In the literature it is generally assumed that coupons are paid until
conversion. However coupon payments are affected by the so called
Maximum Distributable Amount trigger, under which regulators stop
the payment of coupons (and dividends) when the firm's capital value
falls below some trigger that is higher than the conversion trigger.
Coupon payments can start again when the capital value goes back up
and exceeds the trigger value again. This means that in the valuation
of a CoCo, we can apply Theorem~\ref{thmCoCoprice2} and Algorithm
\ref{rwmh3}, but with the coupon term defined as in Equation~(\ref{newterm}).
To demonstrate the relevance of the inclusion of this trigger in the
valuation of CoCos, we will look at the big price drop that the CoCos
of Deutsche Bank suffered at the beginning of 2016. On January 28
Deutsche Bank reported a net loss of 2.1 Billion EUR over the last
quarter of 2015. The relevant report furthermore reported for its
Risk-Weighted Assets a value of 397 Billion EUR, down from 408 Billion
EUR in the previous accounting report. Also, the Common Equity Tier
1 (CET1) ratio, defined as the fraction of the common equity and the
risk weighted asset (RWA), fell from 11.5\% to 11.1\%, primarily reflecting
the net loss over the quarter. The information is taken from the Financial Data Supplement 4Q2015, \citet{DeutscheBank2016},
%\footnote{Information comes from the Financial Data Supplement 4Q2015, which can be found at \url{https://www.db.com/ir/en/download/FDS_4Q2015_11_03_2016.pdf}}.
the report that caused a big downward move in the price of the CoCos
of Deutsche Bank.

At this time, Deutsche Bank had four different CoCos issued (two in
USD, one in EUR, one in GBP, all PWD CoCos). To avoid having to deal
with an additional exchange rate risk factor, we will only consider
the EUR CoCo. This CoCo's write-down is triggered when the CET1-ratio
hits the level of 5.125\% and it pays a coupon of 6\%. As is clear
from the above, the CET1-ratio did not even come close to the low
trigger level. Still, the CoCo price tumbled 19.5\% percent within
the week after the announcement of the report. Market publications
at the time widely argued that this happened out of fear for reaching
the MDA trigger and the subsequent cancelling of coupon payments.
The model developed in this paper is particularly relevant to analyze
this case, as we can include the announcement of a bad accounting
report in the valuation, as well as the early cancelling of coupons
when the MDA trigger is hit. The precise value of the MDA trigger
is not publicly known, so it is not possible to use the real value
of the MDA trigger. However, it is still interesting to examine how
much of a price drop the model can explain by taking the MDA trigger
close to the reported values. Unless stated otherwise, we use the
same parameters as in Table~\ref{BaseCaseParameters}. Before the
bad accounting report arrives, we assume there is one accounting report,
with a value $Y_{t_{1}}$= EUR~408~bn. Then the new accounting report
arrives, so we now have two accounting reports with values $Y_{t_{1}}$=EUR~408~bn
and $Y_{t_{2}}$=EUR~397~bn. The triggers are chosen such that they
correspond with CET1 ratios at the moment of the accounting report.
That is, we choose $v_{c}$ such that it corresponds to a CET1 ratio
of 5.125\%. We know the CET1 ratio is 11.1\% where RWA is EUR~397~bn,
so the total amount of debt (only CoCos and straight debt in the model)
is EUR~397~bn $\times$ 0.889 = EUR~352.93~bn. So a CET1 ratio
of 5.125\% would then correspond to a RWA value of EUR~352.93/(1-0.05125)
= EUR~372~bn, which is thus the value of the conversion trigger
$v_{c}$. The value of the MDA trigger $v_{cc}$ can be chosen in
the same way, a MDA trigger at a CET1 ratio of 10\% would correspond
to a RWA value of EUR~352.93/(1-0.1)~bn = EUR~392~bn. The coupon
of the CoCo is $c_{2}$ = 0.06. As the relevant CoCo has a perpetual
maturity, we choose the first call date, 10/10/18, as the maturity.
Because we assumed that the second accounting report arrives at 01/28/16,
t=0 corresponds to 07/28/15. Hence T = 3 + 2/12 + 13/365. In Figure
\ref{Price_Drop_DB} the price change after the announcement of a
bad accounting report is illustrated for different choices of the
MDA trigger.

\begin{figure}[h]
\begin{center}
\includegraphics[width=0.75\textwidth,height=0.35\textheight]{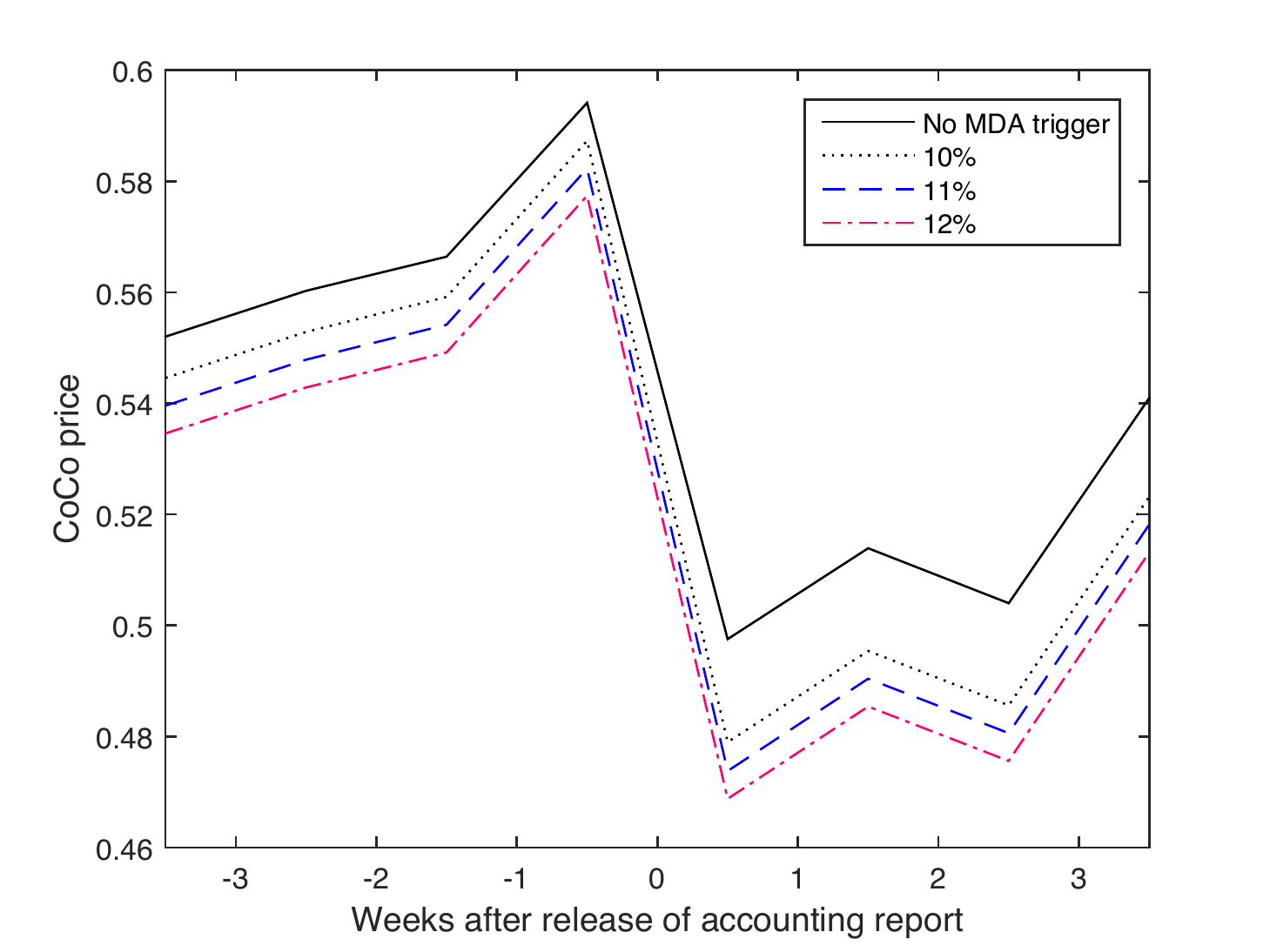}
\caption{CoCo price after the release of the bad accounting report for different
values of the MDA trigger.}
\label{Price_Drop_DB} 
\end{center}
\end{figure}
The solid line corresponds to the case where the MDA-trigger is not
included in the model, in this case only a drop of 14.3\% in the CoCo
price occurs, when looking at the price just before the release of
the accounting report and afterwards. However, if we add the MDA trigger
to the model, a stronger negative price change follows. The dashed
line corresponds to the case that we take the MDA trigger at 11\%,
i.e.\ just beneath the reported CET1 value. This gives a price drop
of 18.7\%. If we take the MDA trigger to be 10\%, the price drops
by 17.3\%, which is illustrated by the dotted line. However, we have
to take the MDA trigger above the reported CET1 ratio of 11.1\% to
create a price drop of 19.5\%, cf.~the dashed-dotted line. That is,
a price drop of 19.5\% corresponds in the model to the situation that
the MDA trigger is already breached, which was not the case. However,
it is clear that a significant part of the price change is driven
by the MDA trigger, not by the conversion trigger. The above illustrates
the added value of explicitly incorporating accounting reports into
the analysis and taking the MDA trigger into account in the valuation
of a CoCo, especially when the MDA trigger is coming close, but the
conversion trigger is still far away.

\section{Conclusions}

\label{sec:conclusions} \setcounter{equation}{0}

CoCos are debt instruments that are written down or converted into
equity when the value of the issuing bank becomes too low. CoCos have
taken European capital markets by storm. Over 560 bn Euro has been issued
over the past five years, with more likely to come. Apparently banks
see CoCos as an attractive alternative to issuing new equity when
faced with a capital shortage. The academic literature has rapidly
developed attempting to analyse and price the new debt instruments,
but at the same time a remarkable divergence has opened up between
this academic literature and the type of CoCos issued in actual practice.
Without exception, the academic literature argues for conversion triggers
based on market values instead of accounting ratios. Accordingly,
with the exception of \citet{Glasserman2012a}, the asset pricing
literature on CoCos has analysed market based conversion triggers
only. Yet, at least in the European Union and Switzerland, market
based triggers disqualify the instrument as capital under EU regulation,
so without a single exception all CoCos issued so far base their conversion
trigger on accounting ratios. In addition, they have to keep open
the possibility of regulatory intervention when a so called Point
of Non-Viability is reached. Moreover, the literature has paid no
attention to the triggers in place for suspension of coupon payments,
although those triggers have most likely caused most of the recent
volatility in CoCo prices. In this paper we bridge the gap between
the academic literature and actual practice by explicitly introducing
coupon suspension triggers and, arguably more important, explicitly
introducing key features of the accounting process in the model which
allows us to analyse accounting ratio based conversion triggers and
PONV interventions in a meaningful way.

In order to do so we model the basic stochastic process driving asset
values as a standard geometric Brownian Motion, as does most of the
literature. Where we diverge is in our assumption that the process
is not directly observable. Instead information reaches the market
based on noisy accounting reports issued at regular but discrete time
epochs only. In this way we can take into account differences between
accounting values and market values. Our model is based on the premise
that the market cannot observe the true asset value process, it only
has access to noisy accounting reports which moreover, are only published
at discrete moments in time. In this way, the price of CoCos can only
be based on the information from the accounting reports, not on the
underlying true asset process as this is not observed directly. The
model does not lead to closed form solutions for CoCo prices, but
Markov Chain Monte Carlo methods are used to compute prices.

The model was remarkably successful in reproducing the price response
of CoCos to a widely reported adverse profit warning issued by Deutsche
Bank in February 2016. This exercise has shown the importance of incorporating
the so called MDA trigger in valuation models, the trigger that governs
suspension of coupon payments. Using the model as a tool of analysis
yields a rich set of results on the relation between valuation, CoCo
design and environment variables such as asset volatility and accounting
noise. Moreover, we have shown that CoCos depending on their design
have a significant impact on shareholder incentives to take on additional
risk or on investment incentives in situations of debt overhang, and
interact in interesting ways with the capital structure of the bank
issuing them. The various results, such as the attractiveness of PWD
CoCos for equity holders, can help explain the design choices made
in practice where about 60\% of all CoCos issued are of that variety.
The explicit incorporation of the accounting process as providing
noisy reports on the underlying unobservable firm fundamentals, which are
issued at regular but discrete time instants, allows us to analyse CoCo
designs based on accounting ratios as well as triggers based on reaching
a so called Point of Non-Viability. The relation between risk taking
incentives, leverage and CoCo design should be of interest to regulators.
We show that CoCos, again depending on their design features, may
significantly change the sensitivity of equity values to risk, thereby
possibly opening the door to risk arbitrage for given capital requirements,
taking into account asset characteristics only, as is the case under
the BIS based capital regime. 

\bibliographystyle{plainnat}
%\bibliography{CocoValuation_biblio_new}

\begin{thebibliography}{99}
\providecommand{\natexlab}[1]{#1}
\providecommand{\url}[1]{\texttt{#1}}
\expandafter\ifx\csname urlstyle\endcsname\relax
  \providecommand{\doi}[1]{doi: #1}\else
  \providecommand{\doi}{doi: \begingroup \urlstyle{rm}\Url}\fi

\bibitem[Albul et~al.(2012)Albul, Jaffee, and A.]{Albul2012}
B.~Albul, D.~Jaffee, and Tchistyi A.
\newblock Contingent convertible bonds and capital structure decisions.
\newblock Working paper, University of California, 2012.
\newblock Available at SSRN: \url{https://ssrn.com/abstract=2772612}.

\bibitem[Avdiev et~al.(2017)Avdiev, Bogdanova, Bolton, Jiang, and
  Kartasheva]{Avdiev2017}
S.~Avdiev, B.~Bogdanova, P.~Bolton, Wei Jiang, and A.~Kartasheva.
\newblock Coco issuance and bank fragility.
\newblock Working paper, BIS, 2017.
\newblock Available at SSRN: \url{https://ssrn.com/abstract=3066030}.

\bibitem[Barclays(2017)]{Barclays2017}
Barclays.
\newblock Barclays {C}oco {IPO}, March 3, 2017.

\bibitem[Brigo et~al.(2015)Brigo, Garcia, and Pede]{Brigo2013}
D.~Brigo, J.~Garcia, and N.~Pede.
\newblock Coco bonds valuation with equity-and-credit risk calibrated first
  passage structural models.
\newblock \emph{International Journal of Theoretical and Applied Finance},
  18\penalty0 (03):\penalty0 1550015, 2015.

\bibitem[Calomiris and Herring(2013)]{Calomiris_Herring_2013}
C.W. Calomiris and R.J. Herring.
\newblock How to design a contingent convertible debt requirement that helps
  solve our too-big-to-fail problem*.
\newblock \emph{Journal of Applied Corporate Finance}, 25\penalty0
  (2):\penalty0 39--62, 2013.
\newblock ISSN 1745-6622.

\bibitem[{C}apital {R}equirements~{R}egulation 575/2013/{EU}(2013)]{CRRart54}
{C}apital {R}equirements~{R}egulation 575/2013/{EU}, 2013.
\newblock URL
  \url{http://eur-lex.europa.eu/legal-content/en/TXT/?uri=celex%3A32013R0575}.

\bibitem[Chan and van Wijnbergen(2016; revised November
  2017)]{Chan2016;revisedNovember2017}
S.~Chan and S.~van Wijnbergen.
\newblock Coco design, risk shifting incentives and capital regulation.
\newblock Working paper, CEPR, 2016; revised November 2017.
\newblock URL \url{https://www.ceps.eu/system/files/ECMI%20WP%20No%202_0.pdf}.

\bibitem[Chan and van Wijnbergen(2018)]{Chan2018}
S.~Chan and S.~van Wijnbergen.
\newblock Regulatory forbearance, cocos and bank risk shifting.
\newblock Working paper, University of Amsterdam, 2018.
\newblock In preparation.

\bibitem[Chen et~al.(2013)Chen, Glasserman, Nouri, and Pelger]{Chen2013}
N.~Chen, P.~Glasserman, B.~Nouri, and M.~Pelger.
\newblock Cocos, bail-in and tail risk.
\newblock Working paper 0004, Office of Financial Research, 2013.
\newblock Available at SSRN: \url{https://ssrn.com/abstract=2296462}.

\bibitem[Corcuera et~al.(2013)Corcuera, De~Spiegeleer, Ferreiro-Castilla,
  Kyprianou, Madan, and Schoutens]{Corcuera2013}
J.~Corcuera, J.~De~Spiegeleer, A.~Ferreiro-Castilla, A.~Kyprianou, D.~Madan,
  and W.~Schoutens.
\newblock Pricing of contingent convertibles under smile conform models.
\newblock \emph{Journal of credit risk}, 9\penalty0 (3):\penalty0 121--140,
  2013.

\bibitem[Duffie and Lando(2001)]{Duffie2001}
D.~Duffie and D.~Lando.
\newblock Term structures of credit spreads with incomplete accounting
  information.
\newblock \emph{Econometrica}, 69\penalty0 (3):\penalty0 633--664, 2001.

\bibitem[Flannery(2005)]{Flannery}
M.J. Flannery.
\newblock No pain, no gain: Effecting market discipline via "reverse
  convertible debentures".
\newblock In Hal~S. Scott, editor, \emph{Capital Adequacy beyond Basel:
  Banking, Securities, and Insurance}, pages 171--196. Oxford University Press,
  2005.

\bibitem[Glasserman and Nouri(2012)]{Glasserman2012a}
P.~Glasserman and B.~Nouri.
\newblock Contingent capital with a capital-ratio trigger.
\newblock \emph{Management Science}, 58\penalty0 (10):\penalty0 1816--1833,
  2012.

\bibitem[Glasserman and Nouri(2016)]{Glasserman2012b}
P.~Glasserman and B.~Nouri.
\newblock Market-triggered changes in capital structure: equilibrium price
  dynamics.
\newblock \emph{Econometrica}, 84\penalty0 (6):\penalty0 2113--2153, 2016.

\bibitem[Haldane(2011)]{Haldane2011}
A.G. Haldane.
\newblock Capital discipline.
\newblock \emph{Bank of England, based on a speech given at the meetings of the
  AEA, Denver January 9th 2011}, 2011.
\newblock URL \url{https://www.bis.org/review/r110325a.pdf}.

\bibitem[Harrison(1985)]{Harrison1985}
J.M. Harrison.
\newblock \emph{Brownian motion and stochastic flow systems}.
\newblock John Wiley and Sons, 1985.

\bibitem[Kiewiet et~al.(2017)Kiewiet, van Lelyveld, and van
  Wijnbergen]{Kiewiet_ea_2017}
G.~Kiewiet, I.~van Lelyveld, and S.~van Wijnbergen.
\newblock Contingent convertibles: can the market handle them?
\newblock Working paper DP12359, CEPR, 2017.
\newblock URL
  \url{www.cepr.org/active/publications/discussion_papers/dp.php?dpno=12359}.

\bibitem[\mbox{Deutsche~Bank}(2016)]{DeutscheBank2016}
\mbox{Deutsche~Bank}.
\newblock Financial data supplement {Q}42015.
\newblock \url{https://www.db.com/ir/en/download/FDS_4Q2015_11_03_2016.pdf},
  2016.

\bibitem[Merton(1974)]{Merton_1974}
Robert Merton.
\newblock On the pricing of corporate debt: the risk structure of interest
  rates.
\newblock \emph{Journal of Finance}, 29\penalty0 (2):\penalty0 449--470, May
  1974.

\bibitem[Myers(1977)]{Myers1977}
Stuart Myers.
\newblock Determinants of corporate borrowing.
\newblock \emph{Journal of Financial Economics}, 5\penalty0 (2):\penalty0 147
  -- 175, 1977.

\bibitem[Pennacchi(2011)]{Pennacchi2011}
G.~Pennacchi.
\newblock A structural model of contingent bank capital.
\newblock Working paper, University of Illinois, 2011.
\newblock URL \url{https://business.illinois.edu/gpennacc/ConCap061811.pdf}.

\bibitem[Pennacchi and Tschistyi(2015)]{Pennacchi2015}
G.~Pennacchi and A.~Tschistyi.
\newblock A reexamination of contingent convertibles with stock price triggers.
\newblock Working paper, University of Illinois, 2015.
\newblock Available at SSRN: \url{https://ssrn.com/abstract=2773335}.

\bibitem[Pitt and Dewji(2016)]{GoldmanSachs2016}
L.~Pitt and B.~Dewji.
\newblock Company update {D}eutsche {B}ank.
\newblock Goldman Sachs Credit Research, October 6, 2016.

\bibitem[Shreve(2004)]{Shreve2004}
S.~Shreve.
\newblock \emph{Stochastic Calculus and Finance II: Continuous Time}.
\newblock Springer Verlag, 2nd edition, 2004.

\bibitem[Sundaresan and Wang(2015)]{Sundaresan2014}
S.~Sundaresan and Z.~Wang.
\newblock On the design of contingent capital with a market trigger.
\newblock \emph{The Journal of Finance}, 70\penalty0 (2):\penalty0 881--920,
  2015.

\bibitem[Wilkens and Bethkens(2014)]{Wilkens2014}
S.~Wilkens and N~Bethkens.
\newblock Contingent convertible ({CoCo}) bonds: A first empirical assesment of
  selected pricing models.
\newblock \emph{Financial Analysts Journal}, 70\penalty0 (2), 2014.
\end{thebibliography}

\appendix

\section{Appendix}

\label{sec:appendix} \setcounter{equation}{0}

In this Appendix, all the mathematical details and proofs that are
left out in the main text, are provided.\medskip{}
\\
 \textbf{Proof of Lemma~\ref{ftildelemma}} $\tilde{f}(t,\cdot,z_{0})$
is defined by 
\[
\mathbb{P}\left(Z_{t}\in\dd x|\tau_{b}>t\right)=\tilde{f}(t,x,z_{0})\dd x.
\]
By Bayes' rule we can write 
\[
\mathbb{P}\left(Z_{t}\in\dd x|\tau_{b}>t\right)=\frac{\mathbb{P}\left(Z_{t}\in\dd x,\tau_{b}>t\right)}{\mathbb{P}(\tau_{b}>t)}.
\]
The denominator of this expression is given by 
\[
\mathbb{P}(\tau_{b}>t)=1-\pi(t,z_{0}-z_{b})=\Phi\left(\frac{z_{0}-z_{b}+mt}{\sigma\sqrt{t}}\right)-e^{-2m(z_{0}-z_{b})/\sigma^{2}}\Phi\left(\frac{z_{b}-z_{0}+mt}{\sigma\sqrt{t}}\right).
\]
In order to compute the numerator, we will rely on the following result
by \citet{Harrison1985}, which can be found in Section 1.8, Proposition 1. Denote
by $X_{t}$ a Brownian motion with drift $\mu$, variance $\sigma^{2}$
and $X_{0}=0$. Furthermore define $M_{t}:=\max\{X_{s}:0\leq s\leq t\}$.
Then the joint distribution of $X_{t}$ and $M_{t}$ satisfies 
\begin{align}
\mathbb{P}\left(X_{t}\in\dd x,M_{t}\leq y\right)=\frac{1}{\sigma\sqrt{t}}\exp\left(\frac{\mu x}{\sigma^{2}}-\frac{\mu^{2}t}{2\sigma^{2}}\right)\left(\phi\left(\frac{x}{\sigma\sqrt{t}}\right)-\phi\left(\frac{x-2y}{\sigma\sqrt{t}}\right)\right)\dd x,\label{harrlemma}
\end{align}
where $\phi$ denotes the standard normal density function. Now, denote
$X_{t}=-Z_{t}+z_{0}$, which is a Brownian motion with drift $-m$,
variance $\sigma^{2}$ and $X_{0}=0$. Furthermore, denote $M_{t}=\max\{X_{s}:0\leq s\leq t\}$.
Then Equation~(\ref{harrlemma}) implies that 
\begin{align}\label{zevenenhalf}
\mathbb{P}\left(Z_{t}\in\dd x,\tau_{b}>t\right) & =\mathbb{P}\left(Z_{t}\in\dd x,\inf_{0\leq s\leq t}Z_{s}>z_{b}\right)\nonumber \\
 & =\mathbb{P}\left(X_{t}\in\dd(z_{0}-x),M_{t}\leq z_{0}-z_{b}\right)\nonumber \\
 & =\frac{1}{\sigma\sqrt{t}}\exp\left(\frac{-m(z_{0}-x)}{\sigma^{2}}-\frac{m^{2}t}{2\sigma^{2}}\right)\left(\phi\left(\frac{z_{0}-x}{\sigma\sqrt{t}}\right)-\phi\left(\frac{-z_{0}-x+2z_{b}}{\sigma\sqrt{t}}\right)\right)\dd x.
\end{align}
So we conclude that 
\begin{align*}
\tilde{f}(t,x,z_{0})=\frac{1}{\sigma\sqrt{t}}\frac{\exp\left(\frac{-m(z_{0}-x)}{\sigma^{2}}-\frac{m^{2}t}{2\sigma^{2}}\right)\left(\phi\left(\frac{z_{0}-x}{\sigma\sqrt{t}}\right)-\phi\left(\frac{-z_{0}-x+2z_{b}}{\sigma\sqrt{t}}\right)\right)}{\Phi\left(\frac{z_{0}-z_{b}+mt}{\sigma\sqrt{t}}\right)-e^{-2m(z_{0}-z_{b})/\sigma^{2}}\Phi\left(\frac{z_{b}-z_{0}+mt}{\sigma\sqrt{t}}\right)}.
\end{align*}
\hfill{}$\square$\medskip\\
\textbf{Proof of Theorem~\ref{thmCoCoprice1}}
Recall that the CoCo price was written as 
\begin{align}
C(t)=Pe^{-r(T-t)}p_{c}(t,T)+cP\int_{t}^{T}e^{-r(u-t)}p_{c}(t,u)\dd u-RP\int_{t}^{T}e^{-r(u-t)}p_{c}(t,\dd u).\label{CoCopricedufflando2}
\end{align}
The integral in this last term can be written as 
\begin{align*}
\int_{t}^{T}e^{-r(u-t)}p_{c}(t,\dd u) & =\int_{t}^{T}e^{-r(u-t)}\frac{\partial}{\partial u}p_{c}(t,u)\dd u\\
 & =\int_{t}^{T}e^{-r(u-t)}\int_{z_{c}}^{\infty}\frac{\partial}{\partial u}(1-\pi(u-t,x-z_{c}))f(t,x,\omega)\dd x\dd u\\
 & =\int_{z_{c}}^{\infty}f(t,x,\omega)\int_{t}^{T}e^{-r(u-t)}\frac{\partial}{\partial u}(1-\pi(u-t,x-z_{c}))\dd u\dd x\\
 & =\int_{z_{c}}^{\infty}f(t,x,\omega)I(x)\dd x,
\end{align*}
where 
\begin{align*}
I(x)=\int_{t}^{T}e^{-r(u-t)}\frac{\partial}{\partial u}(1-\pi(u-t,x-z_{c}))\dd u.
\end{align*}
Furthermore, the integral in the second term of Equation~(\ref{CoCopricedufflando2})
can be written as 
\begin{align*}
\int_{t}^{T}e^{-r(u-t)}p_{c}(t,u)\dd u & =\int_{t}^{T}e^{-r(u-t)}\int_{z_{c}}^{\infty}(1-\pi(u-t,x-z_{c}))f(t,x,\omega)\dd x\dd u\\
 & =\int_{z_{c}}^{\infty}f(t,x,\omega)\int_{t}^{T}e^{-r(u-t)}(1-\pi(u-t,x-z_{c}))\dd u\dd x\\
 & =\int_{z_{c}}^{\infty}f(t,x,\omega)\tilde{I}(x)\dd x,
\end{align*}
where 
\begin{align*}
\tilde{I}(x) & =\int_{t}^{T}e^{-r(u-t)}(1-\pi(u-t,x-z_{c}))\dd u\\
 & =\left[-\frac{1}{r}e^{-r(u-t)}(1-\pi(u-t,x-z_{c}))\right]_{u=t}^{T}+\frac{1}{r}I(x)\\
 & =-\frac{1}{r}e^{-r(T-t)}(1-\pi(T-t,x-z_{c}))+\frac{1}{r}+\frac{1}{r}I(x).
\end{align*}
Putting the above together allows us to write the CoCo price $C(t)$
as a single integral, weighted by the density $f(t,x)$, as follows
\begin{align*}
C(t) & =\int_{z_{c}}^{\infty}\left(Pe^{-r(T-t)}(1-\pi(T-t,x-z_{c}))+cP\tilde{I}(x)-RPI(x)\right)f(t,xa)\dd x\\
 & =\int_{z_{c}}^{\infty}\left(\frac{r-c}{r}Pe^{-r(T-t)}(1-\pi(T-t,x-z_{c}))+\frac{c}{r}P+\left(\frac{cP}{r}-RP\right)I(x)\right)f(t,x)\dd x.
\end{align*}
It now remains to find an analytical expression for $I(x)$. First
consider 
\begin{align*}
\lefteqn{\frac{\partial}{\partial u}(1-\pi(u-t,x-z_{c}))}\\
 & \quad=\frac{\partial}{\partial u}\left(\Phi\left(\frac{x-z_{c}+m(u-t)}{\sigma\sqrt{u-t}}\right)-e^{-2m(x-z_{c})/\sigma^{2}}\Phi\left(\frac{-(x-z_{c})+m(u-t)}{\sigma\sqrt{u-t}}\right)\right)\\
 & \quad=\phi\left(\frac{x-z_{c}+m(u-t)}{\sigma\sqrt{u-t}}\right)\left(\frac{m}{2\sigma\sqrt{u-t}}-\frac{x-z_{c}}{2\sigma(u-t)^{3/2}}\right)\\
 & \quad\quad-e^{-2m(x-z_{c})/\sigma^{2}}\phi\left(\frac{-(x-z_{c})+m(u-t)}{\sigma\sqrt{u-t}}\right)\left(\frac{m}{2\sigma\sqrt{u-t}}+\frac{x-z_{c}}{2\sigma(u-t)^{3/2}}\right)\\
 & \quad=\frac{z_{c}-x}{\sigma(u-t)^{3/2}}\phi\left(\frac{x-z_{c}+m(u-t)}{\sigma\sqrt{u-t}}\right),
\end{align*}
which implies 
\begin{align}\label{I2}
I(x) & =\int_{t}^{T}e^{-r(u-t)}\frac{z_{c}-x}{\sigma(u-t)^{3/2}}\frac{1}{\sqrt{2\pi}}\exp\left(-\frac{(x-z_{c}+m(u-t))^{2}}{2\sigma^{2}(u-t)}\right)\dd u\nonumber \\
 & =\frac{z_{c}-x}{\sqrt{2\pi\sigma^{2}}}\exp\left(-\frac{m(x-z_{c})}{\sigma^{2}}\right)\int_{0}^{T-t}e^{-ru}\frac{1}{u^{3/2}}\exp\left(-\frac{(x-z_{c})^{2}}{2\sigma^{2}u}-\frac{m^{2}u}{2\sigma^{2}}\right)\dd u\nonumber \\
 & =\frac{z_{c}-x}{\sqrt{2\pi\sigma^{2}}}\exp\left(-\frac{m(x-z_{c})}{\sigma^{2}}\right)\int_{0}^{T-t}\frac{1}{u^{3/2}}\exp\left(-\frac{(x-z_{c})^{2}}{2\sigma^{2}}\frac{1}{u}-\left(\frac{m^{2}}{2\sigma^{2}}+r\right)u\right)\dd u\nonumber \\
 & =2\frac{z_{c}-x}{\sqrt{2\pi\sigma^{2}}}\exp\left(-\frac{m(x-z_{c})}{\sigma^{2}}\right)\int_{(T-t)^{-1/2}}^{\infty}\exp\left(-Av^{2}-B\frac{1}{v^{2}}\right)\dd v,
\end{align}
where the last line follows by substitution of $v=u^{-1/2}$ and by
setting 
\[
A=\frac{(x-z_{c})^{2}}{2\sigma^{2}},~B=\frac{m^{2}}{2\sigma^{2}}+r.
\]
Now, by noting that $(Av^{2}+B/v^{2})=(\sqrt{A}v-\sqrt{B}/v)^{2}+2\sqrt{AB}$,
as well as $(Av^{2}+B/v^{2})=(\sqrt{A}v+\sqrt{B}/v)^{2}-2\sqrt{AB}$,
the remaining integral can be evaluated, by doing the substitutions
$u=\sqrt{A}v-\sqrt{B}/v$ and $u=\sqrt{A}v+\sqrt{B}/v$, as follows
\begin{align}\label{int3}
\hspace{2em} & \hspace{-2em}\int_{(T-t)^{-1/2}}^{\infty}\exp\left(-Av^{2}-B\frac{1}{v^{2}}\right)\dd v\nonumber \\
 & =\frac{1}{2\sqrt{A}}\int_{(T-t)^{-1/2}}^{\infty}\exp\left(-(\sqrt{A}v-\sqrt{B}/v)^{2}-2\sqrt{AB}\right)(\sqrt{A}+\sqrt{B}\frac{1}{v^{2}})\dd v\nonumber \\
 & ~~~+\frac{1}{2\sqrt{A}}\int_{(T-t)^{-1/2}}^{\infty}\exp\left(-(\sqrt{A}v+\sqrt{B}/v)^{2}+2\sqrt{AB}\right)(\sqrt{A}-\sqrt{B}\frac{1}{v^{2}})\dd v\nonumber \\
 & =\frac{1}{2\sqrt{A}}e^{-2\sqrt{AB}}\int_{\sqrt{A/(T-t)}-\sqrt{B(T-t)}}^{\infty}e^{-u^{2}}\dd u\nonumber \\
 & ~~~+\frac{1}{2\sqrt{A}}e^{2\sqrt{AB}}\int_{\sqrt{A/(T-t)}+\sqrt{B(T-t)}}^{\infty}e^{-u^{2}}\dd u\nonumber \\
 & =\frac{\sqrt{\pi}}{4\sqrt{A}}\left(e^{-2\sqrt{AB}}\erfc\left(\sqrt{A/(T-t)}-\sqrt{B(T-t)}\right)\right.\nonumber \\
 & ~~~+\left.e^{2\sqrt{AB}}\erfc\left(\sqrt{A/(T-t)}+\sqrt{B(T-t)}\right)\right),
\end{align}
where $\erfc(x)$ is the complementary error function, which is defined
by 
\[
\erfc(x):=\frac{2}{\sqrt{\pi}}\int_{x}^{\infty}e^{-u^{2}}\dd u
\]
and satisfies 
\[
\frac{1}{2}\erfc(x/\sqrt{2})=1-\Phi(x).
\]
Combining Equations (\ref{I2}) and (\ref{int3}) and substituting
back the expressions for $A$ and $B$, finally leads to the expression
for $I(x)$: 
\begin{align}\label{I(x)2}
I(x) & =2\frac{z_{c}-x}{\sqrt{2\pi\sigma^{2}}}\exp\left(-\frac{m(x-z_{c})}{\sigma^{2}}\right)\int_{(T-t)^{-1/2}}^{\infty}\exp\left(-Av^{2}-B\frac{1}{v^{2}}\right)\dd v\nonumber \\
 & =\exp\left(-\frac{m(x-z_{c})}{\sigma^{2}}\right)\left(-e^{-2\sqrt{AB}}\frac{1}{2}\erfc\left(\sqrt{A/(T-t)}-\sqrt{B(T-t)}\right)\right.\nonumber \\
 & ~~~-\left.e^{2\sqrt{AB}}\frac{1}{2}\erfc\left(\sqrt{A/(T-t)}+\sqrt{B(T-t)}\right)\right)\nonumber \\
 & =\exp\left(-\frac{m(x-z_{c})+(x-z_{c})\sqrt{m^{2}+2r\sigma^{2}}}{\sigma^{2}}\right)\left(\Phi\left(\frac{x-z_{c}-\sqrt{m^{2}+2r\sigma^{2}}(T-t)}{\sigma\sqrt{T-t}}\right)-1\right)\nonumber \\
 & ~~~+\exp\left(-\frac{m(x-z_{c})-(x-z_{c})\sqrt{m^{2}+2r\sigma^{2}}}{\sigma^{2}}\right)\left(\Phi\left(\frac{x-z_{c}+\sqrt{m^{2}+2r\sigma^{2}}(T-t)}{\sigma\sqrt{T-t}}\right)-1\right).
\end{align}
\hfill{}$\square$ \\
 \textbf{Proof of Theorem~\ref{thmCoCoprice2}} The market price
of the CoCos is given by 
\begin{align}\label{priceeqconv2}
C(t) & =\mathbb{E}\left(P_{2}e^{-r(T-t)}\mathbf{1}_{\{\tau_{c}>T\}}|\mathcal{H}_{t}\right)+\mathbb{E}\left(\int_{t}^{T}c_{2}P_{2}e^{-r(u-t)}\mathbf{1}_{\{\tau_{c}>u\}}\dd u|\mathcal{H}_{t}\right)\nonumber \\
 & ~~~+\mathbb{E}\left(\frac{\Delta P_{2}}{\Delta P_{2}+1}E^{PC}(\tau_{c})e^{-r(\tau_{c}-t)}\mathbf{1}_{\{\tau_{c}\leq T\}}|\mathcal{H}_{t}\right).
\end{align}
The first two terms together, are captured in the integral 
\[
\int_{z_{c}}^{\infty}h_{0}(x)f(t,x)\dd x,
\]
as is clear from taking $R=0$ in the PWD case, cf.\ Theorem~\ref{thmCoCoprice1}.
\\
 Recall that the third term in Equation~(\ref{priceeqconv2}) was
written as 
\begin{align} \label{newterm2.1}
\hspace{2em} & \hspace{-2em}\mathbb{E}\left(\frac{\Delta P_{2}}{\Delta P_{2}+1}E^{PC}(\tau_{c})e^{-r(\tau_{c}-t)}\mathbf{1}_{\{\tau_{c}\leq T\}}|\mathcal{H}_{t}\right)\nonumber \\
 & =\frac{\Delta P_{2}}{\Delta P_{2}+1}e^{z_{c}}\int_{t}^{T}e^{-r(u-t)}\mathbb{P}(\tau_{c}\in\dd u|\tau_{c}>t,Y^{(n)})\nonumber \\
 & ~~~-\frac{\Delta P_{2}c_{1}P_{1}}{\Delta P_{2}+1}\int_{t}^{\infty}e^{-r(u-t)}\mathbb{P}(\tau_{c}\leq T\wedge u,\tau_{b}>u|\tau_{c}>t,Y^{(n)})\dd u\nonumber \\
 & ~~~-\frac{\Delta P_{2}}{\Delta P_{2}+1}e^{z_{b}}\int_{t}^{\infty}e^{-r(u-t)}\mathbb{P}(\tau_{c}\leq T,\tau_{b}\in\dd u|\tau_{c}>t,Y^{(n)}).
\end{align}
Note that the first integral in this equation is already computed
in the proof of Theorem~\ref{thmCoCoprice1} and given by 
\begin{align}
e^{z_{c}}\int_{t}^{T}e^{-r(u-t)}\mathbb{P}(\tau_{c}\in\dd u|\tau_{c}>t,Y^{(n)})=-e^{z_{c}}\int_{z_{c}}^{\infty}f(t,x)I(x)\dd x,\label{oldterm}
\end{align}
where $I(x)$ is given by Equation~(\ref{I(x)2}). \\
 To compute the other integrals in Equation~(\ref{newterm2.1}),
it is sufficient to find expressions for 
\[
\mathbb{P}(\tau_{c}\leq T,\tau_{b}>u|\tau_{c}>t,Y^{(n)}=y^{(n)})\text{ and }\mathbb{P}(\tau_{c}\leq u,\tau_{b}>u|\tau_{c}>t,Y^{(n)}=y^{(n)}).
\]
In order to find expressions for this joint probabilities, we first
need the following lemma. \begin{lemma} The joint probability $\gamma(x,y,z,t_{1},t_{2})$
that $Z$, starting from $x$, does not hit $z$ before time $t_{1}$
but does hit $y$ before time $t_{2}$, is for $x>y>z$ given by 
\begin{align*}
\gamma(x,y,z,t_{1},t_{2}) & =\mathbb{P}(\inf_{0\leq s\leq t_{1}}Z_{s}>z,\inf_{0\leq s\leq t_{2}}Z_{s}\leq y)\\
 & =\left\{ \begin{matrix}\pi(t_{2},x-y)-\pi(t_{1},x-z) & \text{ for }t_{1}\leq t_{2},\\
1-\pi(t_{1},x-z)-\int_{y}^{\infty}(1-\pi(t_{1}-t_{2},\tilde{z}-z))\hat{f}(x,y,\tilde{z},t_{2})\dd\tilde{z} & \text{ for }t_{1}>t_{2},
\end{matrix}\right.
\end{align*}
where 
\begin{align}
\hat{f}(x,y,\tilde{z},t_{2})=\frac{1}{\sigma\sqrt{t_{2}}}\exp\left(\frac{-m(x-\tilde{z})}{\sigma^{2}}-\frac{m^{2}t_{2}}{2\sigma^{2}}\right)\left(\phi\left(\frac{x-\tilde{z}}{\sigma\sqrt{t_{2}}}\right)-\phi\left(\frac{-x-\tilde{z}+2y}{\sigma\sqrt{t_{2}}}\right)\right)\label{gfunctie}
\end{align}
\end{lemma} \begin{proof} \hfill{}
\begin{itemize}
\item For $t_{1}\leq t_{2}$, we can write 
\begin{align*}
\mathbb{P}(\inf_{0\leq s\leq t_{1}}Z_{s}>z,\inf_{0\leq s\leq t_{2}}Z_{s}\leq y) & =\mathbb{P}(\inf_{0\leq s\leq t_{2}}Z_{s}\leq y)-\mathbb{P}(\inf_{0\leq s\leq t_{1}}Z_{s}\leq z,\inf_{0\leq s\leq t_{2}}Z_{s}\leq y)\\
 & =\mathbb{P}(\inf_{0\leq s\leq t_{2}}Z_{s}\leq y)-\mathbb{P}(\inf_{0\leq s\leq t_{1}}Z_{s}\leq z)\\
 & =\pi(t_{2},x-y)-\pi(t_{1},x-z).
\end{align*}
\item For $t_{1}>t_{2}$, note that 
\begin{align*}
\mathbb{P}(\inf_{0\leq s\leq t_{1}}Z_{s}>z,\inf_{0\leq s\leq t_{2}}Z_{s}\leq y) & =\mathbb{P}(\inf_{0\leq s\leq t_{1}}Z_{s}>z)-\mathbb{P}(\inf_{0\leq s\leq t_{1}}Z_{s}>z,\inf_{0\leq s\leq t_{2}}Z_{s}>y)\\
 & =1-\pi(t_{1},x-z)-\mathbb{P}(\inf_{0\leq s\leq t_{1}}Z_{s}>z,\inf_{0\leq s\leq t_{2}}Z_{s}>y),
\end{align*}
where 
\begin{align*}
\lefteqn{\mathbb{P}(\inf_{0\leq s\leq t_{1}}Z_{s}>z,\inf_{0\leq s\leq t_{2}}Z_{s}>y)}\\
 & \quad=\int_{y}^{\infty}\mathbb{P}\left(\inf_{t_{2}\leq s\leq t_{1}}Z_{s}>z,\inf_{0\leq s\leq t_{2}}Z_{s}>y|Z_{t_{2}}=\tilde{z}\right)\mathbb{P}(Z_{t_{2}}\in\dd\tilde{z})\\
 & \quad=\int_{y}^{\infty}\mathbb{P}\left(\inf_{t_{2}\leq s\leq t_{1}}Z_{s}-Z_{t_{2}}>z-\tilde{z}\right)\mathbb{P}(\inf_{0\leq s\leq t_{2}}Z_{s}>y,Z_{t_{2}}\in\dd\tilde{z})\\
 & \quad=\int_{y}^{\infty}\mathbb{P}\left(\inf_{0\leq s\leq t_{1}-t_{2}}Z_{s}>z-\tilde{z}+x\right)\mathbb{P}(\inf_{0\leq s\leq t_{2}}Z_{s}>y,Z_{t_{2}}\in\dd\tilde{z})\\
 & \quad=\int_{y}^{\infty}(1-\pi(t_{1}-t_{2},\tilde{z}-z))\mathbb{P}(\inf_{0\leq s\leq t_{2}}Z_{s}>y,Z_{t_{2}}\in\dd\tilde{z}),
\end{align*}
where is used that $Z$ has independent and stationary increments.\\
 Now the result follows by noting that by a modification of Equation~(\ref{zevenenhalf})
to the current setting, it holds that 
\end{itemize}
\[
\mathbb{P}(\inf_{0\leq s\leq t_{2}}Z_{s}>y,Z_{t_{2}}\in\dd\tilde{z})=\hat{f}(x,y,\tilde{z},t_{2})\dd\tilde{z}.
\]
\end{proof} Now the desired probabilities are, in analogy to Equation~(\ref{survprobCoCo}),
given by 
\[
\mathbb{P}(\tau_{c}\leq T,\tau_{b}>u|\tau_{c}>t,Y^{(n)}=y^{(n)})=\int_{z_{c}}^{\infty}\gamma(x,z_{c},z_{b},u-t,T-t)f(t,x)\dd x
\]
and 
\[
\mathbb{P}(\tau_{c}\leq u,\tau_{b}>u|\tau_{c}>t,Y^{(n)}=y^{(n)})=\int_{z_{c}}^{\infty}\gamma(x,z_{c},z_{b},u-t,u-t)f(t,x)\dd x.
\]
Recall that the objective was to compute the last two integrals in
Equation~(\ref{newterm2.1}). Let us first consider the second one,
that is 
\begin{align*}
-\int_{t}^{\infty}e^{-r(u-t)}\mathbb{P}(\tau_{c}\leq T,\tau_{b}\in\dd u|\tau_{c}>t,Y^{(n)}) & =\int_{t}^{\infty}e^{-r(u-t)}\frac{\partial}{\partial u}\mathbb{P}(\tau_{c}\leq T,\tau_{b}>u|\tau_{c}>t,Y^{(n)})\dd u\\
 & =(\mathrm{I})+(\mathrm{II}),
\end{align*}
where 
\begin{align*}
(\mathrm{I}) & =\int_{t}^{T}e^{-r(u-t)}\frac{\partial}{\partial u}\mathbb{P}(\tau_{c}\leq T,\tau_{b}>u|\tau_{c}>t,Y^{(n)})\dd u\\
 & =\int_{z_{c}}^{\infty}f(t,x)\int_{t}^{T}e^{-r(u-t)}\frac{\partial}{\partial u}\gamma(x,z_{c},z_{b},u-t,T-t)\dd u\dd x\\
 & =\int_{z_{c}}^{\infty}f(t,x)\int_{t}^{T}e^{-r(u-t)}\frac{\partial}{\partial u}(-\pi(u-t,x-z_{b}))\dd u\dd x\\
 & =\int_{z_{c}}^{\infty}f(t,x)I_{b}(x)\dd x,
\end{align*}
in which 
\begin{align}\label{Ib}
I_{b}(x) & =\int_{t}^{T}e^{-r(u-t)}\frac{\partial}{\partial u}(-\pi(u-t,x-z_{b}))\dd u\nonumber \\
 & =\exp\left(-\frac{m(x-z_{b})+(x-z_{b})\sqrt{m^{2}+2r\sigma^{2}}}{\sigma^{2}}\right)\left(\Phi\left(\frac{x-z_{b}-\sqrt{m^{2}+2r\sigma^{2}}(T-t)}{\sigma\sqrt{T-t}}\right)-1\right)\nonumber \\
 & ~~~+\exp\left(-\frac{m(x-z_{b})-(x-z_{b})\sqrt{m^{2}+2r\sigma^{2}}}{\sigma^{2}}\right)\left(\Phi\left(\frac{x-z_{b}+\sqrt{m^{2}+2r\sigma^{2}}(T-t)}{\sigma\sqrt{T-t}}\right)-1\right),
\end{align}
which follows from Equation~(\ref{I(x)2}), by replacing $z_{c}$
by $z_{b}$. Furthermore, we have 
\begin{align*}
(\mathrm{II}) & =\int_{T}^{\infty}e^{-r(u-t)}\frac{\partial}{\partial u}\mathbb{P}(\tau_{c}\leq T,\tau_{b}>u|\tau_{c}>t,Y^{(n)})\dd u\\
 & =\int_{z_{c}}^{\infty}f(t,x)\int_{T}^{\infty}e^{-r(u-t)}\frac{\partial}{\partial u}\gamma(x,z_{c},z_{b},u-t,T-t)\dd u\dd x\\
 & =\int_{z_{c}}^{\infty}f(t,x)\int_{T}^{\infty}e^{-r(u-t)}\frac{\partial}{\partial u}(-\pi(u-t,x-z_{b}))\dd u\dd x\\
 & ~~~-\int_{z_{c}}^{\infty}\int_{z_{c}}^{\infty}f(t,x)\hat{f}(x,z_{c},\tilde{z},T-t)\int_{T}^{\infty}e^{-r(u-t)}\frac{\partial}{\partial u}(1-\pi(u-T,\tilde{z}-z_{b}))\dd u\dd\tilde{z}\dd x\\
 & =\int_{z_{c}}^{\infty}f(t,x)(J_{b}(x)-I_{b}(x))\dd x\\
 & ~~~-\int_{z_{c}}^{\infty}\int_{z_{c}}^{\infty}f(t,x)\hat{f}(x,z_{c},\tilde{z},T-t)e^{-r(T-t)}J_{b}(\tilde{z},T)\dd\tilde{z}\dd x,
\end{align*}
where 
\begin{align}
J_{b}(x) & =\int_{t}^{\infty}e^{-r(u-t)}\frac{\partial}{\partial u}(1-\pi(u-t,x-z_{b}))\dd u\nonumber \\
 & =-\exp\left(-\frac{m(x-z_{b})+(x-z_{b})\sqrt{m^{2}+2r\sigma^{2}}}{\sigma^{2}}\right),
\end{align}
where the last line follows by taking $T\to\infty$ in Equation~(\ref{Ib}).
This leaves us with an expression for the last integral in Equation~(\ref{newterm2.1}).
\\
 \indent Similarly, the other integral satisfies 
\begin{align*}
\int_{t}^{\infty}e^{-r(u-t)}\mathbb{P}(\tau_{c}\leq T\wedge u,\tau_{b}>u|\tau_{c}>t,Y^{(n)})\dd u & =(\mathrm{III})+(\mathrm{IV}),
\end{align*}
where 
\begin{align*}
(\mathrm{III}) & =\int_{t}^{T}e^{-r(u-t)}\mathbb{P}(\tau_{c}\leq u,\tau_{b}>u|\tau_{c}>t,Y^{(n)})\dd u\\
 & =\int_{z_{c}}^{\infty}f(t,x)\int_{t}^{T}e^{-r(u-t)}\gamma(x,z_{c},z_{b},u-t,u-t)\dd u\dd x\\
 & =\int_{z_{c}}^{\infty}f(t,x)\int_{t}^{T}e^{-r(u-t)}(\pi(u-t,x-z_{c})-\pi(u-t,x-z_{b}))\dd u\dd x\\
 & =\int_{z_{c}}^{\infty}f(t,x)(\tilde{I}_{b}(x)-\tilde{I}(x))\dd x
\end{align*}
in which $\tilde{I}(x)$ is defined in the proof Theorem~\ref{thmCoCoprice1}
and $\tilde{I}_{b}(x)$ is equivalently defined as 
\begin{align}
\tilde{I}_{b}(x) & =\int_{t}^{T}e^{-r(u-t)}(1-\pi(u-t,x-z_{b}))\dd u\nonumber \\
 & =\left[-\frac{1}{r}e^{-r(u-t)}(1-\pi(u-t,x-z_{b}))\right]_{u=t}^{T}+\frac{1}{r}I_{b}(x)\nonumber \\
 & =-\frac{1}{r}e^{-r(T-t)}(1-\pi(T-t,x-z_{b}))+\frac{1}{r}+\frac{1}{r}I_{b}(x).
\end{align}
Furthermore, we have 
\begin{align*}
(\mathrm{IV}) & =\int_{T}^{\infty}e^{-r(u-t)}\mathbb{P}(\tau_{c}\leq T,\tau_{b}>u|\tau_{c}>t,Y^{(n)})\dd u\\
 & =\int_{z_{c}}^{\infty}\int_{T}^{\infty}e^{-r(u-t)}\gamma(x,z_{c},z_{b},u-t,T-t)\dd u\dd x\\
 & =\int_{z_{c}}^{\infty}f(t,x)\int_{T}^{\infty}e^{-r(u-t)}(1-\pi(u-t,x-z_{b}))\dd u\dd x\\
 & ~~~-\int_{z_{c}}^{\infty}\int_{z_{c}}^{\infty}f(t,x)\hat{f}(x,z_{c},\tilde{z},T-t)\int_{T}^{\infty}e^{-r(u-t)}(1-\pi(u-T,\tilde{z}-z_{b}))\dd u\dd\tilde{z}\dd x\\
 & =\int_{z_{c}}^{\infty}f(t,x)(\tilde{J}_{b}(x)-\tilde{I}_{b}(x))\dd x\\
 & ~~~-\int_{z_{c}}^{\infty}\int_{z_{c}}^{\infty}f(t,x)\hat{f}(x,z_{c},\tilde{z},T-t)e^{-r(T-t)}\tilde{J}_{b}(\tilde{z})\dd\tilde{z}\dd x,
\end{align*}
in which 
\begin{align}
\tilde{J}_{b}(x) & =\int_{t}^{\infty}e^{-r(u-t)}(1-\pi(u-t,x-z_{b}))\dd u\nonumber \\
 & =\left[-\frac{1}{r}(1-\pi(u-t,x-z_{b})\right]_{u=t}^{\infty}+\frac{1}{r}J_{b}(x)\nonumber \\
 & =\frac{1}{r}+\frac{1}{r}J_{b}(x).
\end{align}
Putting all the above together leads to an expression for the last
two integrals in Equation~(\ref{newterm2.1}), given by 
\begin{align}\label{finalterms}
\hspace{2em} & \hspace{-2em}-c_{1}P_{1}\int_{t}^{\infty}e^{-r(u-t)}\mathbb{P}(\tau_{c}\leq T\wedge u,\tau_{b}>u|\tau_{c}>t,Y^{(n)})\dd u\nonumber \\
\hspace{2em} & \hspace{-2em}\quad-e^{z_{b}}\int_{t}^{\infty}e^{-r(u-t)}\mathbb{P}(\tau_{c}\leq T,\tau_{b}\in\dd u|\tau_{c}>t,Y^{(n)})\nonumber \\
 & \quad=e^{z_{b}}((I)+(II))-c_{1}P_{1}((III)+(IV))\nonumber \\
 & \quad=\int_{z_{c}}^{\infty}f(t,x)\left(e^{z_{b}}J_{b}(x)+c_{1}P_{1}\tilde{I}(x)-c_{1}P_{1}\tilde{J}_{b}(x)\right)\dd x\nonumber \\
 & \quad\quad+\int_{z_{c}}^{\infty}\int_{z_{c}}^{\infty}f(t,x)\hat{f}(x,z_{c},\tilde{z},T-t)e^{-r(T-t)}(c_{1}P_{1}\tilde{J}_{b}(\tilde{z})-e^{z_{b}}J_{b}(\tilde{z}))\dd\tilde{z}\dd x.
\end{align}
\indent Finally, by combining Equations (\ref{newterm2.1}), (\ref{oldterm})
and (\ref{finalterms}), it follows that the third term in Equation~(\ref{priceeqconv}),
i.e. 
\[
\mathbb{E}\left(\frac{\Delta P_{2}}{\Delta P_{2}+1}E^{PC}(\tau_{c})e^{-r(\tau_{c}-t)}\mathbf{1}_{\{\tau_{c}\leq T\}}|\mathcal{H}_{t}\right),
\]
is given by 
\begin{align}
\int_{z_{c}}^{\infty}f(t,x)h_{1}(x)\dd x+\int_{z_{c}}^{\infty}\int_{z_{c}}^{\infty}f(t,x)\hat{f}(x,z_{c},\tilde{z},T-t)h_{2}(\tilde{z})\dd\tilde{z}\dd x,\label{finalterms2}
\end{align}
where 
\begin{align*}
 & h_{1}(x)=\frac{\Delta P_{2}}{\Delta P_{2}+1}\left(e^{z_{b}}J_{b}(x)+c_{1}P_{1}\tilde{I}(x)-c_{1}P_{1}\tilde{J}_{b}(x)-e^{z_{c}}I(x)\right),\\
 & h_{2}(\tilde{z})=\frac{\Delta P_{2}}{\Delta P_{2}+1}e^{-r(T-t)}(c_{1}P_{1}\tilde{J}_{b}(\tilde{z})-e^{z_{b}}J_{b}(\tilde{z})).
\end{align*}
\hfill{}$\square$\\
 \textbf{Proof of Proposition~\ref{prop_istep}} Denote by $\Delta t$
the time between two successive accounting reports and recall from
section~\ref{section2_model_desc} the notation $Y_{i}=Y_{t_{i}},Z_{i}=Z_{t_{i}},U_{i}=U_{t_{i}}$
and that $Y_{i}=Z_{i}+U_{i}$, where 
\[
U_{i}=\kappa U_{i-1}+\epsilon_{i},
\]
for some fixed $\kappa\in\mathbb{R}$ and independent and identically
distributed $\epsilon_{1},\epsilon_{2},\dots$, which have a normal
distribution with mean $\mu_{\epsilon}$ and variance $\sigma_{\epsilon}^{2}$,
and are independent of $Z$. This allows us to write, for $i=1,2,\dots$
\begin{align*}
\left(\begin{matrix}Y_{n+i}\\
\vdots\\
Y_{n+1}
\end{matrix}\right) & =M\left(\begin{matrix}Z_{n+i}-Z_{n+i-1}\\
\vdots\\
Z_{n+1}-Z_{n}\\
\epsilon_{n+i}\\
\vdots\\
\epsilon_{n+1}
\end{matrix}\right)+\left(\begin{matrix}\kappa^{i}\\
\vdots\\
\kappa
\end{matrix}\right)Y_{n}+\left(\begin{matrix}1-\kappa^{i}\\
\vdots\\
1-\kappa
\end{matrix}\right)Z_{n},
\end{align*}
where $M$ denotes the $(i\times2i)$-matrix defined by 
\[
M=\overbrace{\left(\begin{matrix}1 & 1 & \cdots & 1 & 1\\
0 & 1 & \cdots & 1 & 1\\
\vdots & \ddots & \ddots & \vdots & \vdots\\
0 & \cdots & 0 & 1 & 1\\
0 & \cdots & 0 & 0 & 1
\end{matrix}\right.}^{\text{ \ensuremath{i} components}}\overbrace{\left.\begin{matrix}1 & \kappa & \kappa^{2} & \cdots & \kappa^{i-1}\\
0 & 1 & \kappa & \cdots & \kappa^{i-2}\\
\vdots & \ddots & \ddots & \ddots & \vdots\\
0 & \cdots & 0 & 1 & \kappa\\
0 & \cdots & 0 & 0 & 1
\end{matrix}\right)}^{\text{ \ensuremath{i} components}}
\]
and where the vector $(Z_{n+i}-Z_{n+i-1},\dots,Z_{n+1}-Z_{n},\epsilon_{n+i},\dots,\epsilon_{n+1})$
follows a multivariate normal distribution with $2i$-dimensional
mean vector $\mu_{i}'$ and $(2i\times2i)$-dimensional covariance
matrix $\Sigma_{i}'$, defined by 
\[
\mu_{i}'=\begin{pmatrix}m\Delta t\\
\vdots\\
m\Delta t\\
\mu_{\epsilon}\\
\vdots\\
\mu_{\epsilon}
\end{pmatrix},~\Sigma_{i}'=\text{Diag}(\sigma^{2}\Delta t,\dots\sigma^{2}\Delta t,\sigma_{\epsilon}^{2},\dots,\sigma_{\epsilon}^{2}).
\]
Hence it follows that the conditional density $p_{Y}(y_{n+i},\dots,y_{n+1}|y^{(n)},z^{(n)})$
of $y_{n+i},\dots,y_{n+1}$ given $Y^{(n)}=y^{(n)}$ and $Z^{(n)}=z^{(n)}$,
is the density of a multivariate normal distribution with mean vector
\[
\hat{\mu}_{i}=M\mu_{i}'+\left(\begin{matrix}\kappa^{i}\\
\vdots\\
\kappa
\end{matrix}\right)y_{n}+\left(\begin{matrix}1-\kappa^{i}\\
\vdots\\
1-\kappa
\end{matrix}\right)z_{n}
\]
and covariance matrix 
\[
\Sigma_{i}=M\Sigma_{i}'M^{\top}.
\]
\indent The conditional density of the next $i$ accounting values,
given $Y^{(n)}$ can be written as 
\[
p_{Y}(y_{n+i},\dots,y_{n+1}|y^{(n)})=\int_{\mathbb{R}^{n}}p_{Y}(y_{n+i},\dots,y_{n+1}|y^{(n)},z^{(n)})p_{Z}(z^{(n)}|y^{(n)})\dd z^{(n)},
\]
where the conditional density $p_{Z}(z^{(n)}|y^{(n)})$ of $Z^{(n)}$,
given $Y^{(n)}=y^{(n)}$, can be computed in the same way as $b_{n}(z^{(n)}|y^{(n)})$
in section~\ref{section2_cond_dens}, which leads to 
\begin{align}
p_{Z}(z^{(n)}|y^{(n)}) & =\frac{p_{Z}(z_{n}|z_{n-1})p_{U}(y_{n}-z_{n}|y_{n-1}-z_{n-1})p_{Z}(z^{(n-1)}|y^{(n-1)})}{p_{Y}(y_{n}|y^{(n-1)})}\nonumber \\
 & =\frac{\prod_{i=1}^{n}p_{Z}(z_{i}|z_{i-1})p_{U}(y_{i}-z_{i}|y_{i-1}-z_{i-1})}{p_{Y}(y_{n}|y^{(n-1)})},
\end{align}
This leads to an expression for the survival probability until time
$t_{n+i}$, given survival until time $t_{n}\leq t<t_{n+1}$, that
is 
\begin{align}\label{survprobonlyrep1.1}
\mathbb{P}\left(\tau_{c}^{A}>t_{n+i}|Y^{(n)}=y^{(n)}\right) & =\int_{(y_{c},\infty)^{i}}p_{Y}(y_{n+i},\dots,y_{n+1}|y^{(n)})\dd y_{n+1},\dots,\dd y_{n+i}\nonumber \\
 & =\int_{\mathbb{R}^{n}}\mathbb{P}(\xi(z_{n})\in(y_{c},\infty)^{i})p_{Z}(z^{(n)}|y^{(n)})\dd z^{(n)},
\end{align}
where $\xi(z_{n})$ denotes a multivariate normal distributed random
variable with mean vector $\hat{\mu}_{i}$ and covariance matrix $\Sigma_{i}$.
\hfill{}$\square$\\
 \\
 \textbf{Proof of Theorem~\ref{thmCoCoprice3}} The CoCo's market
price is given by 
\begin{align} \label{CoCoprice_app}
C'(t) & =\mathbb{E}\left(Pe^{-r(T-t)}\mathbf{1}_{\{\tau_{c}>T\}}|\mathcal{H}_{t}\right)+\mathbb{E}\left(\int_{t}^{T}cPe^{-r(u-t)}\mathbf{1}_{\{\tau_{c}>u\}}\dd u|\mathcal{H}_{t}\right)\nonumber \\
 & ~~~+\mathbb{E}\left(RPe^{-r(\tau_{c}-t)}\mathbf{1}_{\{\tau_{c}\leq T\}}|\mathcal{H}_{t}\right).
\end{align}
For $t_{n}\leq t<t_{n+1}$, $T=t_{n+m}$ for some $m\in\mathbb{N}$
and $Y^{(n)}=y^{(n)}$, where $y_{i}>y_{c},1\leq i\leq n$, this can
be written as 
\begin{align*}
C'(t) & =Pe^{-r(T-t)}\mathbb{P}(\tau_{c}>T|Y^{(n)}=y^{(n)})+\int_{t}^{T}cPe^{-r(u-t)}\mathbb{P}(\tau_{c}>u|Y^{(n)}=y^{(n)})\dd u\\
 & ~~~~+RP\sum_{i=1}^{m}e^{-r(t_{n+i}-t)}\mathbb{P}(\tau_{c}=t_{n+i}|Y^{(n)}=y^{(n)})\\
 & =Pe^{-r(T-t)}\mathbb{P}(\tau_{c}>t_{n+m}|Y^{(n)}=y^{(n)})\\
 & ~~~~+cP\left(\sum_{i=1}^{m-1}\int_{t_{n+i}}^{t_{n+i+1}}e^{-r(u-t)}\dd u\mathbb{P}(\tau_{c}>t_{n+i}|Y^{(n)}=y^{(n)})+\int_{t}^{t_{n+i}}e^{-r(u-t)}\dd u\right)\\
 & ~~~~+RP\sum_{i=1}^{m}e^{-r(t_{n+i}-t)}\left(\mathbb{P}(\tau_{c}>t_{n+i-1}|Y^{(n)}=y^{(n)})-\mathbb{P}(\tau_{c}>t_{n+i}|Y^{(n)}=y^{(n)})\right)\\
 & =Pe^{-r(T-t)}\mathbb{P}(\tau_{c}>t_{n+m}|Y^{(n)}=y^{(n)})\\
 & ~~~~+\sum_{i=1}^{m-1}\frac{cP}{r}(e^{-r(t_{n+i}-t)}-e^{-r(t_{n+i+1}-t)})\mathbb{P}(\tau_{c}>t_{n+i}|Y^{(n)}=y^{(n)})\\
 & ~~~~+\frac{cP}{r}(1-e^{-r(t_{n+1}-t)})\\
 & ~~~~+RP\sum_{i=1}^{m}e^{-r(t_{n+i}-t)}\left(\mathbb{P}(\tau_{c}>t_{n+i-1}|Y^{(n)}=y^{(n)})-\mathbb{P}(\tau_{c}>t_{n+i}|Y^{(n)}=y^{(n)})\right)\\
 & =(1-R)Pe^{-r(T-t)}\mathbb{P}(\tau_{c}>t_{n+m}|Y^{(n)}=y^{(n)})\\
 & ~~~~+\sum_{i=1}^{m-1}\left(\frac{cP}{r}-RP\right)(e^{-r(t_{n+i}-t)}-e^{-r(t_{n+i+1}-t)})\mathbb{P}(\tau_{c}>t_{n+i}|Y^{(n)}=y^{(n)})\\
 & ~~~~+\frac{cP}{r}(1-e^{-r(t_{n+1}-t)})+RPe^{-r(t_{n+1}-t)},
\end{align*}
\hfill{}$\square$\\
 \textbf{Proof of Theorem~\ref{thmCoCoprice4}} This theorem is only
a small adaption of Theorem~\ref{thmCoCoprice3}. The second term
in Equation~(\ref{CoCoprice_app}) above, needs to be replaced by
\[
\mathbb{E}\left(\sum_{i=1}^{m-1}\int_{t_{n+i}}^{t_{n+i+1}}cPe^{-r(u-t)}\mathbf{1}_{\{\tau_{c}^{A}>u,Y_{n+i}>y_{cc}\}}\dd u+\mathbf{1}_{\{Y_{n}>y_{cc}\}}\int_{t}^{t_{n+1}}cPe^{-r(u-t)}\dd u\Big|\mathcal{H}_{t}\right),
\]
where $t_{n}\leq t<t_{n+1}$, $T=t_{n+m}$ for some $m\in\mathbb{N}$.
\\
 For $Y^{(n)}=y^{(n)}$, where $y_{i}>y_{c},1\leq i\leq n$, this
can be written as 
\begin{align*}
\lefteqn{\sum_{i=1}^{m-1}\int_{t_{n+i}}^{t_{n+i+1}}e^{-r(u-t)}\mathbb{P}(\tau_{c}^{A}>u,Y_{n+i}>y_{cc}|Y^{(n)}=y^{(n)})\dd u+\mathbf{1}_{\{Y_{n}>y_{cc}\}}\int_{t}^{t_{n+1}}cPe^{-r(u-t)}\dd u}\\
 & \quad=\sum_{i=1}^{m-1}\frac{cP}{r}(e^{-r(t_{n+i}-t)}-e^{-r(t_{n+i+1}-t)})\mathbb{P}(\tau_{c}^{A}>t_{n+i},Y_{n+i}>y_{cc}|Y^{(n)}=y^{(n)})\\
 & \quad\quad+\mathbf{1}_{\{Y_{n}>y_{cc}\}}\frac{cP}{r}(1-e^{-r(t_{n+1}-t)}),
\end{align*}
where, similar to Equation~(\ref{survprobonlyrep1.1}), 
\begin{align*}
\mathbb{P}(\tau_{c}^{A}>t_{n+i},Y_{n+i}>y_{cc}|Y^{(n)}=y^{(n)}) & =\mathbb{P}(Y_{n+1}>y_{c},\dots,Y_{n+i-1}>y_{c},Y_{n+i}>y_{cc}|Y^{(n)}=y^{(n)})\\
 & =\int_{\mathbb{R}^{n}}\mathbb{P}(\xi(z_n)\in(y_{c},\infty)^{i-1}\times(y_{cc},\infty))p_{Z}(z^{(n)}|y^{(n)})\dd z^{(n)}.\,\hfill\mbox{ \ensuremath{\square}}
\end{align*}

\end{document}